\documentclass[11pt, a4paper]{article}
\pdfoutput=1 

\usepackage{paralist}
\usepackage{array}
\usepackage[cmex10]{amsmath}
\usepackage{cite}
\usepackage{graphics}
\usepackage{amsthm}
\usepackage{amssymb}
\usepackage{units}
\usepackage{cases}
\usepackage{subfigure}
\usepackage{hyperref}
\usepackage{makeidx}

\usepackage{xspace}
\usepackage{url} 

\usepackage{bbm}
\usepackage{graphicx}

\usepackage{paralist}
\usepackage{fancyref}
\usepackage[stretch=16,shrink=16,step=4]{microtype}
\usepackage[printonlyused]{acronym}
\newlength{\figwidth}
\usepackage{paralist}
\usepackage{array}
\usepackage{cite}
\usepackage{graphics}
\usepackage{amsthm}
\usepackage{amsmath}
\usepackage{amssymb}
\usepackage{units}
\usepackage{cases}
\usepackage{subfigure}
\usepackage{hyperref}
\usepackage{makeidx}

\makeatletter
\def \footnotelabel{%
\edef\@currentlabel{\p@footnote\thefootnote}%
\label} \makeatother

\newcommand{\Figref}[1]{Figure~\ref{#1}}

\newcommand{\Ftref}[1]{Footnote~\ref{#1}}

\newcommand{\DMO}[1]{\expandafter\DeclareMathOperator*{#1}}

\def\addsymbol #1: #2#3{$#1$ \> \parbox{5in}{#2 \dotfill \pageref{#3}}\\}
 
\def\additionletter{A}

\newcommand{\addit}{\mathcal{\expandafter\uppercase\expandafter{\additionletter}}}

\def\normletter{N}
\newcommand{\normit}{\mathcal{\expandafter\uppercase\expandafter{\normletter}}}

\newcommand{\DLM}[3]{\ensuremath{\left\expandafter#1#2\right\expandafter#3}}

\newcommand{\V}[1]{\ensuremath{\expandafter\lowercase\expandafter{\mathbf{#1}}}} 

\DMO{\linearspan}{span} 

\newcommand{\M}[1]{\ensuremath{\expandafter\uppercase\expandafter{\mathbf{#1}}}} 

\newcommand{\abs}[1]{\ensuremath{\left\vert#1\right\vert}} 

\DeclareMathOperator*{\kernel}{\mathcal{K}er} 

\DeclareMathOperator*{\rank}{rank} 

\newcommand{\adj}[1]{\ensuremath{\adjoint\PP{#1}}} 

\DeclareMathOperator*{\variance}{var} 

\theoremstyle{plain}
\newtheorem{thm}{Theorem}[section]
\newtheorem{lem}[thm]{Lemma}
\newtheorem{cor}[thm]{Corollary}

\theoremstyle{definition}
\newtheorem{dfn}[thm]{Definition}
\newtheorem{xca}[thm]{Exercise}
\newtheorem{example}[thm]{Example}
	
\usepackage{amsfonts}
\usepackage{mathrsfs}
\usepackage{xspace}
\usepackage{bm}
\usepackage{fancyref}
\usepackage{textcomp}
\usepackage[Symbol]{upgreek}

\newcommand{\safemath}[2]{\newcommand{#1}{\ensuremath{#2}\xspace}}

\newcommand{\ssa}{\mathsf{a}}
\newcommand{\ssb}{\mathsf{b}}
\newcommand{\ssc}{\mathsf{c}}
\newcommand{\ssd}{\mathsf{d}}
\newcommand{\sse}{\mathsf{e}}
\newcommand{\ssf}{\mathsf{f}}
\newcommand{\ssg}{\mathsf{g}}
\newcommand{\ssh}{\mathsf{h}}
\newcommand{\ssi}{\mathsf{i}}
\newcommand{\ssj}{\mathsf{j}}
\newcommand{\ssk}{\mathsf{k}}
\newcommand{\ssl}{\mathsf{l}}
\newcommand{\ssm}{\mathsf{m}}
\newcommand{\ssn}{\mathsf{n}}
\newcommand{\sso}{\mathsf{o}}
\newcommand{\ssp}{\mathsf{p}}
\newcommand{\ssq}{\mathsf{q}}
\newcommand{\ssr}{\mathsf{r}}
\newcommand{\sss}{\mathsf{s}}
\newcommand{\sst}{\mathsf{t}}
\newcommand{\ssu}{\mathsf{u}}
\newcommand{\ssv}{\mathsf{v}}
\newcommand{\ssw}{\mathsf{w}}
\newcommand{\ssx}{\mathsf{x}}
\newcommand{\ssy}{\mathsf{y}}
\newcommand{\ssz}{\mathsf{z}}

\safemath{\bmsa}{\bm{\ssa}}
\safemath{\bmsb}{\bm{\ssb}}
\safemath{\bmsc}{\bm{\ssc}}
\safemath{\bmsd}{\bm{\ssd}}
\safemath{\bmse}{\bm{\sse}}
\safemath{\bmsf}{\bm{\ssf}}
\safemath{\bmsg}{\bm{\ssg}}
\safemath{\bmsh}{\bm{\ssh}}
\safemath{\bmsi}{\bm{\ssi}}
\safemath{\bmsj}{\bm{\ssj}}
\safemath{\bmsk}{\bm{\ssk}}
\safemath{\bmsl}{\bm{\ssl}}
\safemath{\bmsm}{\bm{\ssm}}
\safemath{\bmsn}{\bm{\ssn}}
\safemath{\bmso}{\bm{\sso}}
\safemath{\bmsp}{\bm{\ssp}}
\safemath{\bmsq}{\bm{\ssq}}
\safemath{\bmsr}{\bm{\ssr}}
\safemath{\bmss}{\bm{\sss}}
\safemath{\bmst}{\bm{\sst}}
\safemath{\bmsu}{\bm{\ssu}}
\safemath{\bmsv}{\bm{\ssv}}
\safemath{\bmsw}{\bm{\ssw}}
\safemath{\bmsx}{\bm{\ssx}}
\safemath{\bmsy}{\bm{\ssy}}
\safemath{\bmsz}{\bm{\ssz}}

\bmdefine{\bmualphad}{\upalpha}
\bmdefine{\bmubetad}{\upbeta}
\bmdefine{\bmuchid}{\upchi}
\bmdefine{\bmudeltad}{\updelta}
\bmdefine{\bmuepsilond}{\upepsilon}
\bmdefine{\bmuvarepsilond}{\upvarepsilon}
\bmdefine{\bmuetad}{\upeta}
\bmdefine{\bmugammad}{\upgamma}
\bmdefine{\bmuiotad}{\upiota}
\bmdefine{\bmukappad}{\upkappa}
\bmdefine{\bmulambdad}{\uplambda}
\bmdefine{\bmumud}{\upmu}
\bmdefine{\bmunud}{\upnu}
\bmdefine{\bmuomegad}{\upomega}
\bmdefine{\bmuphid}{\upphi}
\bmdefine{\bmuvarphid}{\upvarphi}
\bmdefine{\bmupid}{\uppi}
\bmdefine{\bmuvarpid}{\upvarpi}
\bmdefine{\bmupsid}{\uppsi}
\bmdefine{\bmurhod}{\uprho}
\bmdefine{\bmuvarrhod}{\upvarrho}
\bmdefine{\bmusigmad}{\upsigma}
\bmdefine{\bmuvarsigmad}{\upvarsigma}
\bmdefine{\bmutaud}{\uptau}
\bmdefine{\bmuthetad}{\uptheta}
\bmdefine{\bmuvarthetad}{\upvartheta}
\bmdefine{\bmuupsilond}{\upupsilon}
\bmdefine{\bmuxid}{\upxi}
\bmdefine{\bmuzetad}{\upzeta}

\safemath{\bmua}{\mathbf{a}}
\safemath{\bmub}{\mathbf{b}}
\safemath{\bmuc}{\mathbf{c}}
\safemath{\bmud}{\mathbf{d}}
\safemath{\bmue}{\mathbf{e}}
\safemath{\bmuf}{\mathbf{f}}
\safemath{\bmug}{\mathbf{g}}
\safemath{\bmuh}{\mathbf{h}}
\safemath{\bmui}{\mathbf{i}}
\safemath{\bmuj}{\mathbf{j}}
\safemath{\bmuk}{\mathbf{k}}
\safemath{\bmul}{\mathbf{l}}
\safemath{\bmum}{\mathbf{m}}
\safemath{\bmun}{\mathbf{n}}
\safemath{\bmuo}{\mathbf{o}}
\safemath{\bmup}{\mathbf{p}}
\safemath{\bmuq}{\mathbf{q}}
\safemath{\bmur}{\mathbf{r}}
\safemath{\bmus}{\mathbf{s}}
\safemath{\bmut}{\mathbf{t}}
\safemath{\bmuu}{\mathbf{u}}
\safemath{\bmuv}{\mathbf{v}}
\safemath{\bmuw}{\mathbf{w}}
\safemath{\bmux}{\mathbf{x}}
\safemath{\bmuy}{\mathbf{y}}
\safemath{\bmuz}{\mathbf{z}}

\safemath{\bmualpha}{\bmualphad}
\safemath{\bmubeta}{\bmubetad}
\safemath{\bmuchi}{\bumchid}
\safemath{\bmudelta}{\bmudeltad}
\safemath{\bmuepsilon}{\bmuepsilond}
\safemath{\bmuvarepsilon}{\bmuvarepsilond}
\safemath{\bmueta}{\bmuetad}
\safemath{\bmugamma}{\bmugammad}
\safemath{\bmuiota}{\bmuiotad}
\safemath{\bmukappa}{\bmukappad}
\safemath{\bmulambda}{\bmulambdad}
\safemath{\bmumu}{\bmumud}
\safemath{\bmunu}{\bmunud}
\safemath{\bmuomega}{\bmuomegad}
\safemath{\bmuphi}{\bmuphid}
\safemath{\bmuvarphi}{\bmuvarphid}
\safemath{\bmupi}{\bmupid}
\safemath{\bmuvarpi}{\bmuvarpid}
\safemath{\bmupsi}{\bmupsid}
\safemath{\bmurho}{\bmurhod}
\safemath{\bmuvarrho}{\bmuvarrhod}
\safemath{\bmusigma}{\bmusigmad}
\safemath{\bmuvarsigma}{\bmuvarsigmad}
\safemath{\bmutau}{\bmutaud}
\safemath{\bmutheta}{\bmuthetad}
\safemath{\bmuvartheta}{\bmuvarthetad}
\safemath{\bmuupsilon}{\bmuupsilond}
\safemath{\bmuxi}{\bmuxid}
\safemath{\bmuzeta}{\bmuzetad}

\bmdefine{\bmiad}{a}
\bmdefine{\bmibd}{b}
\bmdefine{\bmicd}{c}
\bmdefine{\bmidd}{d}
\bmdefine{\bmied}{e}
\bmdefine{\bmifd}{f}
\bmdefine{\bmigd}{g}
\bmdefine{\bmihd}{h}
\bmdefine{\bmiid}{i}
\bmdefine{\bmijd}{j}
\bmdefine{\bmikd}{k}
\bmdefine{\bmild}{l}
\bmdefine{\bmimd}{m}
\bmdefine{\bmind}{n}
\bmdefine{\bmiod}{o}
\bmdefine{\bmipd}{p}
\bmdefine{\bmiqd}{q}
\bmdefine{\bmird}{r}
\bmdefine{\bmisd}{s}
\bmdefine{\bmitd}{t}
\bmdefine{\bmiud}{u}
\bmdefine{\bmivd}{v}
\bmdefine{\bmiwd}{w}
\bmdefine{\bmixd}{x}
\bmdefine{\bmiyd}{y}
\bmdefine{\bmizd}{z}

\bmdefine{\bmialphad}{\alpha}
\bmdefine{\bmibetad}{\beta}
\bmdefine{\bmichid}{\chi}
\bmdefine{\bmideltad}{\delta}
\bmdefine{\bmiepsilond}{\epsilon}
\bmdefine{\bmivarepsilond}{\varepsilon}
\bmdefine{\bmietad}{\eta}
\bmdefine{\bmigammad}{\gamma}
\bmdefine{\bmiiotad}{\iota}
\bmdefine{\bmikappad}{\kappa}
\bmdefine{\bmivarkappad}{\varkappa}
\bmdefine{\bmilambdad}{\lambda}
\bmdefine{\bmimud}{\mu}
\bmdefine{\bminud}{\nu}
\bmdefine{\bmiomegad}{\omega}
\bmdefine{\bmiphid}{\phi}
\bmdefine{\bmivarphid}{\varphi}
\bmdefine{\bmipid}{\pi}
\bmdefine{\bmivarpid}{\varpi}
\bmdefine{\bmipsid}{\psi}
\bmdefine{\bmirhod}{\rho}
\bmdefine{\bmivarrhod}{\varrho}
\bmdefine{\bmisigmad}{\sigma}
\bmdefine{\bmivarsigmad}{\varsigma}
\bmdefine{\bmitaud}{\tau}
\bmdefine{\bmithetad}{\theta}
\bmdefine{\bmivarthetad}{\vartheta}
\bmdefine{\bmiupsilond}{\upsilon}
\bmdefine{\bmixid}{\xi}
\bmdefine{\bmizetad}{\zeta}

\safemath{\bmia}{\bmiad}
\safemath{\bmib}{\bmibd}
\safemath{\bmic}{\bmicd}
\safemath{\bmid}{\bmidd}
\safemath{\bmie}{\bmied}
\safemath{\bmif}{\bmifd}
\safemath{\bmig}{\bmigd}
\safemath{\bmih}{\bmihd}
\safemath{\bmii}{\bmiid}
\safemath{\bmij}{\bmijd}
\safemath{\bmik}{\bmikd}
\safemath{\bmil}{\bmild}
\safemath{\bmim}{\bmimd}
\safemath{\bmin}{\bmind}
\safemath{\bmio}{\bmiod}
\safemath{\bmip}{\bmipd}
\safemath{\bmiq}{\bmiqd}
\safemath{\bmir}{\bmird}
\safemath{\bmis}{\bmisd}
\safemath{\bmit}{\bmitd}
\safemath{\bmiu}{\bmiud}
\safemath{\bmiv}{\bmivd}
\safemath{\bmiw}{\bmiwd}
\safemath{\bmix}{\bmixd}
\safemath{\bmiy}{\bmiyd}
\safemath{\bmiz}{\bmizd}

\safemath{\bmialpha}{\bmialphad}
\safemath{\bmibeta}{\bmibetad}
\safemath{\bmichi}{\bmichid}
\safemath{\bmidelta}{\bmideltad}
\safemath{\bmiepsilon}{\bmiepsilond}
\safemath{\bmivarepsilon}{\bmivarepsilond}
\safemath{\bmieta}{\bmietad}
\safemath{\bmigamma}{\bmigammad}
\safemath{\bmiiota}{\bmiiotad}
\safemath{\bmikappa}{\bmikappad}
\safemath{\bmivarkappa}{\bmivarkappad}
\safemath{\bmilambda}{\bmilambdad}
\safemath{\bmimu}{\bmimud}
\safemath{\bminu}{\bminud}
\safemath{\bmiomega}{\bmiomegad}
\safemath{\bmiphi}{\bmiphid}
\safemath{\bmivarphi}{\bmivarphid}
\safemath{\bmipi}{\bmipid}
\safemath{\bmivarpi}{\bmivarpid}
\safemath{\bmipsi}{\bmipsid}
\safemath{\bmirho}{\bmirhod}
\safemath{\bmivarrho}{\bmivarrhod}
\safemath{\bmisigma}{\bmisigmad}
\safemath{\bmivarsigma}{\bmivarsigmad}
\safemath{\bmitau}{\bmitaud}
\safemath{\bmitheta}{\bmithetad}
\safemath{\bmivartheta}{\bmivarthetad}
\safemath{\bmiupsilon}{\bmiupsilond}
\safemath{\bmixi}{\bmixid}
\safemath{\bmizeta}{\bmizetad}

\bmdefine{\bmuDeltad}{\Updelta}
\bmdefine{\bmuGammad}{\Upgamma}
\bmdefine{\bmuLambdad}{\Uplambda}
\bmdefine{\bmuOmegad}{\Upomega}
\bmdefine{\bmuPhid}{\Upphi}
\bmdefine{\bmuPid}{\Uppi}
\bmdefine{\bmuPsid}{\Uppsi}
\bmdefine{\bmuSigmad}{\Upsigma}
\bmdefine{\bmuThetad}{\Uptheta}
\bmdefine{\bmuUpsilond}{\Upupsilon}
\bmdefine{\bmuXid}{\Upxi}

\safemath{\bmuA}{\mathbf{A}}
\safemath{\bmuB}{\mathbf{B}}
\safemath{\bmuC}{\mathbf{C}}
\safemath{\bmuD}{\mathbf{D}}
\safemath{\bmuE}{\mathbf{E}}
\safemath{\bmuF}{\mathbf{F}}
\safemath{\bmuG}{\mathbf{G}}
\safemath{\bmuH}{\mathbf{H}}
\safemath{\bmuI}{\mathbf{I}}
\safemath{\bmuJ}{\mathbf{J}}
\safemath{\bmuK}{\mathbf{K}}
\safemath{\bmuL}{\mathbf{L}}
\safemath{\bmuM}{\mathbf{M}}
\safemath{\bmuN}{\mathbf{N}}
\safemath{\bmuO}{\mathbf{O}}
\safemath{\bmuP}{\mathbf{P}}
\safemath{\bmuQ}{\mathbf{Q}}
\safemath{\bmuR}{\mathbf{R}}
\safemath{\bmuS}{\mathbf{S}}
\safemath{\bmuT}{\mathbf{T}}
\safemath{\bmuU}{\mathbf{U}}
\safemath{\bmuV}{\mathbf{V}}
\safemath{\bmuW}{\mathbf{W}}
\safemath{\bmuX}{\mathbf{X}}
\safemath{\bmuY}{\mathbf{Y}}
\safemath{\bmuZ}{\mathbf{Z}}

\safemath{\bmuZero}{\mathbf{0}}
\safemath{\bmuOne}{\mathbf{1}}

\safemath{\bmuDelta}{\bmuDeltad}
\safemath{\bmuGamma}{\bmuGammad}
\safemath{\bmuLambda}{\bmuLambdad}
\safemath{\bmuOmega}{\bmuOmegad}
\safemath{\bmuPhi}{\bmuPhid}
\safemath{\bmuPi}{\bmuPid}
\safemath{\bmuPsi}{\bmuPsid}
\safemath{\bmuSigma}{\bmuSigmad}
\safemath{\bmuTheta}{\bmuThetad}
\safemath{\bmuUpsilon}{\bmuUpsilond}
\safemath{\bmuXi}{\bmuXid}

\bmdefine{\bmiAd}{A}
\bmdefine{\bmiBd}{B}
\bmdefine{\bmiCd}{C}
\bmdefine{\bmiDd}{D}
\bmdefine{\bmiEd}{E}
\bmdefine{\bmiFd}{F}
\bmdefine{\bmiGd}{G}
\bmdefine{\bmiHd}{H}
\bmdefine{\bmiId}{I}
\bmdefine{\bmiJd}{J}
\bmdefine{\bmiKd}{K}
\bmdefine{\bmiLd}{L}
\bmdefine{\bmiMd}{M}
\bmdefine{\bmiOd}{N}
\bmdefine{\bmiPd}{O}
\bmdefine{\bmiQd}{P}
\bmdefine{\bmiRd}{R}
\bmdefine{\bmiSd}{S}
\bmdefine{\bmiTd}{T}
\bmdefine{\bmiUd}{U}
\bmdefine{\bmiVd}{V}
\bmdefine{\bmiWd}{W}
\bmdefine{\bmiXd}{X}
\bmdefine{\bmiYd}{Y}
\bmdefine{\bmiZd}{Z}

\bmdefine{\bmiDeltad}{\Delta}
\bmdefine{\bmiGammad}{\Gamma}
\bmdefine{\bmiLambdad}{\Lambda}
\bmdefine{\bmiOmegad}{\Omega}
\bmdefine{\bmiPhid}{\Phi}
\bmdefine{\bmiPid}{\Pi}
\bmdefine{\bmiPsid}{\Psi}
\bmdefine{\bmiSigmad}{\Sigma}
\bmdefine{\bmiThetad}{\Theta}
\bmdefine{\bmiUpsilond}{\Upsilon}
\bmdefine{\bmiXid}{\Xi}

\safemath{\bmiA}{\bmiAd}
\safemath{\bmiB}{\bmiBd}
\safemath{\bmiC}{\bmiCd}
\safemath{\bmiD}{\bmiDd}
\safemath{\bmiE}{\bmiEd}
\safemath{\bmiF}{\bmiFd}
\safemath{\bmiG}{\bmiGd}
\safemath{\bmiH}{\bmiHd}
\safemath{\bmiI}{\bmiId}
\safemath{\bmiJ}{\bmiJd}
\safemath{\bmiK}{\bmiKd}
\safemath{\bmiL}{\bmiLd}
\safemath{\bmiM}{\bmiMd}
\safemath{\bmiN}{\bmiNd}
\safemath{\bmiO}{\bmiOd}
\safemath{\bmiP}{\bmiPd}
\safemath{\bmiQ}{\bmiQd}
\safemath{\bmiR}{\bmiRd}
\safemath{\bmiS}{\bmiSd}
\safemath{\bmiT}{\bmiTd}
\safemath{\bmiU}{\bmiUd}
\safemath{\bmiV}{\bmiVd}
\safemath{\bmiW}{\bmiWd}
\safemath{\bmiX}{\bmiXd}
\safemath{\bmiY}{\bmiYd}
\safemath{\bmiZ}{\bmiZd}

\safemath{\bmiDelta}{\bmiDeltad}
\safemath{\bmiGamma}{\bmiGammad}
\safemath{\bmiLambda}{\bmiLambdad}
\safemath{\bmiOmega}{\bmiOmegad}
\safemath{\bmiPhi}{\bmiPhid}
\safemath{\bmiPi}{\bmiPid}
\safemath{\bmiPsi}{\bmiPsid}
\safemath{\bmiSigma}{\bmiSigmad}
\safemath{\bmiTheta}{\bmiThetad}
\safemath{\bmiUpsilon}{\bmiUpsilond}
\safemath{\bmiXi}{\bmiXid}

\safemath{\evA}{\mathcal{A}}
\safemath{\evB}{\mathcal{B}}
\safemath{\evC}{\mathcal{C}}
\safemath{\evD}{\mathcal{D}}
\safemath{\evE}{\mathcal{E}}
\safemath{\evF}{\mathcal{F}}
\safemath{\evG}{\mathcal{G}}
\safemath{\evH}{\mathcal{H}}
\safemath{\evI}{\mathcal{I}}
\safemath{\evJ}{\mathcal{J}}
\safemath{\evK}{\mathcal{K}}
\safemath{\evL}{\mathcal{L}}
\safemath{\evM}{\mathcal{M}}
\safemath{\evN}{\mathcal{N}}
\safemath{\evO}{\mathcal{O}}
\safemath{\evP}{\mathcal{P}}
\safemath{\evQ}{\mathcal{Q}}
\safemath{\evR}{\mathcal{R}}
\safemath{\evS}{\mathcal{S}}
\safemath{\evT}{\mathcal{T}}
\safemath{\evU}{\mathcal{U}}
\safemath{\evV}{\mathcal{V}}
\safemath{\evW}{\mathcal{W}}
\safemath{\evX}{\mathcal{X}}
\safemath{\evY}{\mathcal{Y}}
\safemath{\evZ}{\mathcal{Z}}

\safemath{\setA}{\mathcal{A}}
\safemath{\setB}{\mathcal{B}}
\safemath{\setC}{\mathcal{C}}
\safemath{\setD}{\mathcal{D}}
\safemath{\setE}{\mathcal{E}}
\safemath{\setF}{\mathcal{F}}
\safemath{\setG}{\mathcal{G}}
\safemath{\setH}{\mathcal{H}}
\safemath{\setI}{\mathcal{I}}
\safemath{\setJ}{\mathcal{J}}
\safemath{\setK}{\mathcal{K}}
\safemath{\setL}{\mathcal{L}}
\safemath{\setM}{\mathcal{M}}
\safemath{\setN}{\mathcal{N}}
\safemath{\setO}{\mathcal{O}}
\safemath{\setP}{\mathcal{P}}
\safemath{\setQ}{\mathcal{Q}}
\safemath{\setR}{\mathcal{R}}
\safemath{\setS}{\mathcal{S}}
\safemath{\setT}{\mathcal{T}}
\safemath{\setU}{\mathcal{U}}
\safemath{\setV}{\mathcal{V}}
\safemath{\setW}{\mathcal{W}}
\safemath{\setX}{\mathcal{X}}
\safemath{\setY}{\mathcal{Y}}
\safemath{\setZ}{\mathcal{Z}}
\safemath{\emptySet}{\varnothing}

\safemath{\colA}{\mathscr{A}}
\safemath{\colB}{\mathscr{B}}
\safemath{\colC}{\mathscr{C}}
\safemath{\colD}{\mathscr{D}}
\safemath{\colE}{\mathscr{E}}
\safemath{\colF}{\mathscr{F}}
\safemath{\colG}{\mathscr{G}}
\safemath{\colH}{\mathscr{H}}
\safemath{\colI}{\mathscr{I}}
\safemath{\colJ}{\mathscr{J}}
\safemath{\colK}{\mathscr{K}}
\safemath{\colL}{\mathscr{L}}
\safemath{\colM}{\mathscr{M}}
\safemath{\colN}{\mathscr{N}}
\safemath{\colO}{\mathscr{O}}
\safemath{\colP}{\mathscr{P}}
\safemath{\colQ}{\mathscr{Q}}
\safemath{\colR}{\mathscr{R}}
\safemath{\colS}{\mathscr{S}}
\safemath{\colT}{\mathscr{T}}
\safemath{\colU}{\mathscr{U}}
\safemath{\colV}{\mathscr{V}}
\safemath{\colW}{\mathscr{W}}
\safemath{\colX}{\mathscr{X}}
\safemath{\colY}{\mathscr{Y}}
\safemath{\colZ}{\mathscr{Z}}

\safemath{\opA}{\mathbb{A}}
\safemath{\opB}{\mathbb{B}}
\safemath{\opC}{\mathbb{C}}
\safemath{\opD}{\mathbb{D}}
\safemath{\opE}{\mathbb{E}}
\safemath{\opF}{\mathbb{F}}
\safemath{\opG}{\mathbb{G}}
\safemath{\opH}{\mathbb{H}}
\safemath{\opI}{\mathbb{I}}
\safemath{\opJ}{\mathbb{J}}
\safemath{\opK}{\mathbb{K}}
\safemath{\opL}{\mathbb{L}}
\safemath{\opM}{\mathbb{M}}
\safemath{\opN}{\mathbb{N}}
\safemath{\opO}{\mathbb{O}}
\safemath{\opP}{\mathbb{P}}
\safemath{\opQ}{\mathbb{Q}}
\safemath{\opR}{\mathbb{R}}
\safemath{\opS}{\mathbb{S}}
\safemath{\opT}{\mathbb{T}}
\safemath{\opU}{\mathbb{U}}
\safemath{\opV}{\mathbb{V}}
\safemath{\opW}{\mathbb{W}}
\safemath{\opX}{\mathbb{X}}
\safemath{\opY}{\mathbb{Y}}
\safemath{\opZ}{\mathbb{Z}}
\safemath{\opZero}{\mathbb{O}}
\safemath{\identityop}{\opI}

\safemath{\sca}{a}
\safemath{\scb}{b}
\safemath{\scc}{c}
\safemath{\scd}{d}
\safemath{\sce}{e}
\safemath{\scf}{f}
\safemath{\scg}{g}
\safemath{\sch}{h}
\safemath{\sci}{i}
\safemath{\scj}{j}
\safemath{\sck}{k}
\safemath{\scl}{l}
\safemath{\scm}{m}
\safemath{\scn}{n}
\safemath{\sco}{o}
\safemath{\scp}{p}
\safemath{\scq}{q}
\safemath{\scr}{r}
\safemath{\scs}{s}
\safemath{\sct}{t}
\safemath{\scu}{u}
\safemath{\scv}{v}
\safemath{\scw}{w}
\safemath{\scx}{x}
\safemath{\scy}{y}
\safemath{\scz}{z}

\safemath{\scA}{A}
\safemath{\scB}{B}
\safemath{\scC}{C}
\safemath{\scD}{D}
\safemath{\scE}{E}
\safemath{\scF}{F}
\safemath{\scG}{G}
\safemath{\scH}{H}
\safemath{\scI}{I}
\safemath{\scJ}{J}
\safemath{\scK}{K}
\safemath{\scL}{L}
\safemath{\scM}{M}
\safemath{\scN}{N}
\safemath{\scO}{O}
\safemath{\scP}{P}
\safemath{\scQ}{Q}
\safemath{\scR}{R}
\safemath{\scS}{S}
\safemath{\scT}{T}
\safemath{\scU}{U}
\safemath{\scV}{V}
\safemath{\scW}{W}
\safemath{\scX}{X}
\safemath{\scY}{Y}
\safemath{\scZ}{Z}

\safemath{\scalpha}{\alpha}
\safemath{\scbeta}{\beta}
\safemath{\scchi}{\chi}
\safemath{\scdelta}{\delta}
\safemath{\scepsilon}{\epsilon}
\safemath{\scvarepsilon}{\varepsilon}
\safemath{\sceta}{\eta}
\safemath{\scgamma}{\gamma}
\safemath{\sciota}{\iota}
\safemath{\sckappa}{\kappa}
\safemath{\scvarkappa}{\varkappa}
\safemath{\sclambda}{\lambda}
\safemath{\scmu}{\mu}
\safemath{\scnu}{\nu}
\safemath{\scomega}{\omega}
\safemath{\scphi}{\phi}
\safemath{\scvarphi}{\varphi}
\safemath{\scpi}{\pi}
\safemath{\scvarpi}{\varpi}
\safemath{\scpsi}{\psi}
\safemath{\scrho}{\rho}
\safemath{\scvarrho}{\varrho}
\safemath{\scsigma}{\sigma}
\safemath{\scvarsigma}{\varsigma}
\safemath{\sctau}{\tau}
\safemath{\sctheta}{\theta}
\safemath{\scvartheta}{\vartheta}
\safemath{\scupsilon}{\upsilon}
\safemath{\scxi}{\xi}
\safemath{\sczeta}{\zeta}

\safemath{\veca}{\mathbf{a}}
\safemath{\vecb}{\mathbf{b}}
\safemath{\vecc}{\mathbf{c}}
\safemath{\vecd}{\mathbf{d}}
\safemath{\vece}{\mathbf{e}}
\safemath{\vecf}{\mathbf{f}}
\safemath{\vecg}{\mathbf{g}}
\safemath{\vech}{\mathbf{h}}
\safemath{\veci}{\mathbf{i}}
\safemath{\vecj}{\mathbf{j}}
\safemath{\veck}{\mathbf{k}}
\safemath{\vecl}{\mathbf{l}}
\safemath{\vecm}{\mathbf{m}}
\safemath{\vecn}{\mathbf{n}}
\safemath{\veco}{\mathbf{o}}
\safemath{\vecp}{\mathbf{p}}
\safemath{\vecq}{\mathbf{q}}
\safemath{\vecr}{\mathbf{r}}
\safemath{\vecs}{\mathbf{s}}
\safemath{\vect}{\mathbf{t}}
\safemath{\vecu}{\mathbf{u}}
\safemath{\vecv}{\mathbf{v}}
\safemath{\vecw}{\mathbf{w}}
\safemath{\vecx}{\mathbf{x}}
\safemath{\vecy}{\mathbf{y}}
\safemath{\vecz}{\mathbf{z}}
\safemath{\veczero}{\mathbf{0}}
\safemath{\vecone}{\mathbf{1}}

\safemath{\vecalpha}{\upalpha}
\safemath{\vecbeta}{\upbeta}
\safemath{\vecchi}{\upchi}
\safemath{\vecdelta}{\updelta}
\safemath{\vecepsilon}{\upepsilon}
\safemath{\vecvarepsilon}{\upvarepsilon}
\safemath{\veceta}{\upeta}
\safemath{\vecgamma}{\upgamma}
\safemath{\veciota}{\upiota}
\safemath{\veckappa}{\upkappa}
\safemath{\veclambda}{\uplambda}
\safemath{\vecmu}{\text{\textmu}}
\safemath{\vecnu}{\upnu}
\safemath{\vecomega}{\upomega}
\safemath{\vecphi}{\upphi}
\safemath{\vecvarphi}{\upvarphi}
\safemath{\vecpi}{\uppi}
\safemath{\vecvarpi}{\upvarpi}
\safemath{\vecpsi}{\uppsi}
\safemath{\vecrho}{\uprho}
\safemath{\vecvarrho}{\upvarrho}
\safemath{\vecsigma}{\upsigma}
\safemath{\vecvarsigma}{\upvarsigma}
\safemath{\vectau}{\uptau}
\safemath{\vectheta}{\uptheta}
\safemath{\vecvartheta}{\upvartheta}
\safemath{\vecupsilon}{\upupsilon}
\safemath{\vecxi}{\upxi}
\safemath{\veczeta}{\upzeta}

\safemath{\vecac}{a}
\safemath{\vecbc}{b}
\safemath{\veccc}{c}
\safemath{\vecdc}{d}
\safemath{\vecec}{e}
\safemath{\vecfc}{f}
\safemath{\vecgc}{g}
\safemath{\vechc}{h}
\safemath{\vecic}{i}
\safemath{\vecjc}{j}
\safemath{\veckc}{k}
\safemath{\veclc}{l}
\safemath{\vecmc}{m}
\safemath{\vecnc}{n}
\safemath{\vecoc}{o}
\safemath{\vecpc}{p}
\safemath{\vecqc}{q}
\safemath{\vecrc}{r}
\safemath{\vecsc}{s}
\safemath{\vectc}{t}
\safemath{\vecuc}{u}
\safemath{\vecvc}{v}
\safemath{\vecwc}{w}
\safemath{\vecxc}{x}
\safemath{\vecyc}{y}
\safemath{\veczc}{z}

\safemath{\matA}{\mathbf{A}}
\safemath{\matB}{\mathbf{B}}
\safemath{\matC}{\mathbf{C}}
\safemath{\matD}{\mathbf{D}}
\safemath{\matE}{\mathbf{E}}
\safemath{\matF}{\mathbf{F}}
\safemath{\matG}{\mathbf{G}}
\safemath{\matH}{\mathbf{H}}
\safemath{\matI}{\mathbf{I}}
\safemath{\matJ}{\mathbf{J}}
\safemath{\matK}{\mathbf{K}}
\safemath{\matL}{\mathbf{L}}
\safemath{\matM}{\mathbf{M}}
\safemath{\matN}{\mathbf{N}}
\safemath{\matO}{\mathbf{O}}
\safemath{\matP}{\mathbf{P}}
\safemath{\matQ}{\mathbf{Q}}
\safemath{\matR}{\mathbf{R}}
\safemath{\matS}{\mathbf{S}}
\safemath{\matT}{\mathbf{T}}
\safemath{\matU}{\mathbf{U}}
\safemath{\matV}{\mathbf{V}}
\safemath{\matW}{\mathbf{W}}
\safemath{\matX}{\mathbf{X}}
\safemath{\matY}{\mathbf{Y}}
\safemath{\matZ}{\mathbf{Z}}
\safemath{\matzero}{\mathbf{0}}

\safemath{\matDelta}{\Updelta}
\safemath{\matGamma}{\Upgammma}
\safemath{\matLambda}{\Uplambda}
\safemath{\matOmega}{\Upomega}
\safemath{\matPhi}{\Upphi}
\safemath{\matPi}{\Uppi}
\safemath{\matPsi}{\Uppsi}
\safemath{\matSigma}{\Upsigma}
\safemath{\matTheta}{\Uptheta}
\safemath{\matUpsilon}{\Upupsilon}
\safemath{\matXi}{\Upxi}

\safemath{\matidentity}{\matI}
\safemath{\matone}{\matO}

\safemath{\matAc}{a}
\safemath{\matBc}{b}
\safemath{\matCc}{c}
\safemath{\matDc}{d}
\safemath{\matEc}{e}
\safemath{\matFc}{f}
\safemath{\matGc}{g}
\safemath{\matHc}{h}
\safemath{\matIc}{i}
\safemath{\matJc}{j}
\safemath{\matKc}{k}
\safemath{\matLc}{l}
\safemath{\matMc}{m}
\safemath{\matNc}{n}
\safemath{\matOc}{o}
\safemath{\matPc}{p}
\safemath{\matQc}{q}
\safemath{\matRc}{r}
\safemath{\matSc}{s}
\safemath{\matTc}{t}
\safemath{\matUc}{u}
\safemath{\matVc}{v}
\safemath{\matWc}{w}
\safemath{\matXc}{x}
\safemath{\matYc}{y}
\safemath{\matZc}{z}

\safemath{\rnda}{\bmia}
\safemath{\rndb}{\bmib}
\safemath{\rndc}{\bmic}
\safemath{\rndd}{\bmid}
\safemath{\rnde}{\bmie}
\safemath{\rndf}{\bmif}
\safemath{\rndg}{\bmig}
\safemath{\rndh}{\bmih}
\safemath{\rndi}{\bmii}
\safemath{\rndj}{\bmij}
\safemath{\rndk}{\bmik}
\safemath{\rndl}{\bmil}
\safemath{\rndm}{\bmim}
\safemath{\rndn}{\bmin}
\safemath{\rndo}{\bmio}
\safemath{\rndp}{\bmip}
\safemath{\rndq}{\bmiq}
\safemath{\rndr}{\bmir}
\safemath{\rnds}{\bmis}
\safemath{\rndt}{\bmit}
\safemath{\rndu}{\bmiu}
\safemath{\rndv}{\bmiv}
\safemath{\rndw}{\bmiw}
\safemath{\rndx}{\bmix}
\safemath{\rndy}{\bmiy}
\safemath{\rndz}{\bmiz}

\safemath{\rndA}{\bmiA}
\safemath{\rndB}{\bmiB}
\safemath{\rndC}{\bmiC}
\safemath{\rndD}{\bmiD}
\safemath{\rndE}{\bmiE}
\safemath{\rndF}{\bmiF}
\safemath{\rndG}{\bmiG}
\safemath{\rndH}{\bmiH}
\safemath{\rndI}{\bmiI}
\safemath{\rndJ}{\bmiJ}
\safemath{\rndK}{\bmiK}
\safemath{\rndL}{\bmiL}
\safemath{\rndM}{\bmiM}
\safemath{\rndN}{\bmiN}
\safemath{\rndO}{\bmiO}
\safemath{\rndP}{\bmiP}
\safemath{\rndQ}{\bmiQ}
\safemath{\rndR}{\bmiR}
\safemath{\rndS}{\bmiS}
\safemath{\rndT}{\bmiT}
\safemath{\rndU}{\bmiU}
\safemath{\rndV}{\bmiV}
\safemath{\rndW}{\bmiW}
\safemath{\rndX}{\bmiX}
\safemath{\rndY}{\bmiY}
\safemath{\rndZ}{\bmiZ}

\safemath{\rndalpha}{\bmialpha}
\safemath{\rndbeta}{\bmibeta}
\safemath{\rndchi}{\bmichi}
\safemath{\rnddelta}{\bmidelta}
\safemath{\rndepsilon}{\bmiepsilon}
\safemath{\rndvarepsilon}{\bmivarepsilon}
\safemath{\rndeta}{\bmieta}
\safemath{\rndgamma}{\bmigamma}
\safemath{\rndiota}{\bmiiota}
\safemath{\rndkappa}{\bmikappa}
\safemath{\rndlambda}{\bmilambda}
\safemath{\rndmu}{\bmimu}
\safemath{\rndnu}{\bminu}
\safemath{\rndomega}{\bmiomega}
\safemath{\rndphi}{\bmiphi}
\safemath{\rndvarphi}{\bmivarphi}
\safemath{\rndpi}{\bmipi}
\safemath{\rndvarpi}{\bmivarpi}
\safemath{\rndpsi}{\bmipsi}
\safemath{\rndrho}{\bmirho}
\safemath{\rndvarrho}{\bmivarrho}
\safemath{\rndsigma}{\bmisigma}
\safemath{\rndvarsigma}{\bmivarsigma}
\safemath{\rndtau}{\bmitau}
\safemath{\rndtheta}{\bmitheta}
\safemath{\rndvartheta}{\bmivartheta}
\safemath{\rndupsilon}{\bmiupsilon}
\safemath{\rndxi}{\bmixi}
\safemath{\rndzeta}{\bmizeta}

\safemath{\rveca}{\bmua}
\safemath{\rvecb}{\bmub}
\safemath{\rvecc}{\bmuc}
\safemath{\rvecd}{\bmud}
\safemath{\rvece}{\bmue}
\safemath{\rvecf}{\bmuf}
\safemath{\rvecg}{\bmug}
\safemath{\rvech}{\bmuh}
\safemath{\rveci}{\bmui}
\safemath{\rvecj}{\bmuj}
\safemath{\rveck}{\bmuk}
\safemath{\rvecl}{\bmul}
\safemath{\rvecm}{\bmum}
\safemath{\rvecn}{\bmun}
\safemath{\rveco}{\bmuo}
\safemath{\rvecp}{\bmup}
\safemath{\rvecq}{\bmuq}
\safemath{\rvecr}{\bmur}
\safemath{\rvecs}{\bmus}
\safemath{\rvect}{\bmut}
\safemath{\rvecu}{\bmuu}
\safemath{\rvecv}{\bmuv}
\safemath{\rvecw}{\bmuw}
\safemath{\rvecx}{\bmux}
\safemath{\rvecy}{\bmuy}
\safemath{\rvecz}{\bmuz}

\safemath{\rvecalpha}{\bmualpha}
\safemath{\rvecbeta}{\bmubeta}
\safemath{\rvecchi}{\bmuchi}
\safemath{\rvecdelta}{\bmudelta}
\safemath{\rvecepsilon}{\bmuepsilon}
\safemath{\rvecvarepsilon}{\bmuvarepsilon}
\safemath{\rveceta}{\bmueta}
\safemath{\rvecgamma}{\bmugamma}
\safemath{\rveciota}{\bmuiota}
\safemath{\rveckappa}{\bmukappa}
\safemath{\rveclambda}{\bmulambda}
\safemath{\rvecmu}{\bmumu}
\safemath{\rvecnu}{\bmunu}
\safemath{\rvecomega}{\bmuomega}
\safemath{\rvecphi}{\bmuphi}
\safemath{\rvecvarphi}{\bmuvarphi}
\safemath{\rvecpi}{\bmupi}
\safemath{\rvecvarpi}{\bmuvarpi}
\safemath{\rvecpsi}{\bmupsi}
\safemath{\rvecrho}{\bmurho}
\safemath{\rvecvarrho}{\bmuvarrho}
\safemath{\rvecsigma}{\bmusigma}
\safemath{\rvecvarsigma}{\bmuvarsigma}
\safemath{\rvectau}{\bmutau}
\safemath{\rvectheta}{\bmutheta}
\safemath{\rvecvartheta}{\bmuvartheta}
\safemath{\rvecupsilon}{\bmuupsilon}
\safemath{\rvecxi}{\bmuxi}
\safemath{\rveczeta}{\bmuzeta}

\safemath{\rmatA}{\bmuA}
\safemath{\rmatB}{\bmuB}
\safemath{\rmatC}{\bmuC}
\safemath{\rmatD}{\bmuD}
\safemath{\rmatE}{\bmuE}
\safemath{\rmatF}{\bmuF}
\safemath{\rmatG}{\bmuG}
\safemath{\rmatH}{\bmuH}
\safemath{\rmatI}{\bmuI}
\safemath{\rmatJ}{\bmuJ}
\safemath{\rmatK}{\bmuK}
\safemath{\rmatL}{\bmuL}
\safemath{\rmatM}{\bmuM}
\safemath{\rmatN}{\bmuN}
\safemath{\rmatO}{\bmuO}
\safemath{\rmatP}{\bmuP}
\safemath{\rmatQ}{\bmuQ}
\safemath{\rmatR}{\bmuR}
\safemath{\rmatS}{\bmuS}
\safemath{\rmatT}{\bmuT}
\safemath{\rmatU}{\bmuU}
\safemath{\rmatV}{\bmuV}
\safemath{\rmatW}{\bmuW}
\safemath{\rmatX}{\bmuX}
\safemath{\rmatY}{\bmuY}
\safemath{\rmatZ}{\bmuZ}

\safemath{\rmatDelta}{\bmuDelta}
\safemath{\rmatGamma}{\bmuGamma}
\safemath{\rmatLambda}{\bmuLambda}
\safemath{\rmatOmega}{\bmuOmega}
\safemath{\rmatPhi}{\bmuPhi}
\safemath{\rmatPi}{\bmuPi}
\safemath{\rmatPsi}{\bmuPsi}
\safemath{\rmatSigma}{\bmuSigma}
\safemath{\rmatTheta}{\bmuTheta}
\safemath{\rmatUpsilon}{\bmuUpsilon}
\safemath{\rmatXi}{\bmuXi}

\newenvironment{textbmatrix}{	\setlength{\arraycolsep}{2.5pt}%
								\big[\begin{matrix}}{\end{matrix}\big]%
								\raisebox{0.08ex}{\vphantom{M}}}

 \def\btm{\begin{textbmatrix}}
 \def\etm{\end{textbmatrix}}

\newcommand{\lefto}{\mathopen{}\left}

\DeclareMathOperator{\sinc}{sinc}			
\DeclareMathOperator{\Exop}{\opE}			
\DeclareMathOperator{\spn}{span}			 	

\safemath{\fun}{x}						
\safemath{\altfun}{y}
\safemath{\aaltfun}{\sch}
\safemath{\bel}{\sce}					
\safemath{\altbel}{\sce}					
\safemath{\frel}{g}					
\safemath{\altfrel}{g}					
\safemath{\dfrel}{\tilde{g}}					
\safemath{\altdfrel}{\tilde{g}}					

\safemath{\mat}{\matA}						
\safemath{\matc}{\matAc}						
\safemath{\altmat}{\matB}						
\safemath{\altmatc}{\matBc}						

\safemath{\vectr}{\vecu}						
\safemath{\vectrc}{\vecuc}						
\safemath{\altvectr}{\vecv}						
\safemath{\altvectrc}{\vecvc}						

\newcommand{\nullspace}{\setN}	 			
\newcommand{\rng}{\setR}		 				
\newcommand{\orth}{\perp}					

\newcommand{\vecnorm}[1]{\lVert#1\rVert}		
\newcommand{\conj}[1]{\ensuremath{#1^{*}}} 	
\newcommand{\tp}[1]{\ensuremath{#1^{\mathsf{T}}}} 		
\newcommand{\herm}[1]{\ensuremath{#1^{\mathsf{H}}}} 	
\newcommand{\ad}[1]{\ensuremath{#1^{*}}} 		
\newcommand{\inv}[1]{\ensuremath{#1^{-1}}} 	
\newcommand{\pinv}[1]{\ensuremath{#1^{\dagger}}} 	

\safemath{\dirac}{\delta}					
\safemath{\diracp}{\dirac(\time)}			
\safemath{\krond}{\dirac}					
\safemath{\indfun}{I}						
\safemath{\stepfun}{u}						

\safemath{\upto}{\uparrow}
\safemath{\downto}{\downarrow}
\safemath{\iu}{\mathrm{i}}							
\safemath{\maj}{\succ}
\newcommand{\elementof}[2]{\left[#1\right]_{#2}}		
\newcommand{\dftmat}[1]{\matF_{#1}}			
\safemath{\mdft}{\dftmat{}}					
\safemath{\runity}{\beta}					
\safemath{\eval}{\lambda}					

\safemath{\veval}{\veclambda}				
\newcommand{\cex}[1]{e^{\iu2\pi #1}}			
\newcommand{\cexn}[1]{e^{-\iu2\pi #1}}		
\safemath{\littleo}{\sco}					

\let\im\undefined
\safemath{\re}{\Re}				
\safemath{\im}{\Im}				

\safemath{\euclidspace}{\complexset}			
\safemath{\confunspace}{\setC}				
\newcommand{\banachseqspace}[1]{l^{#1}}		
\safemath{\hilseqspace}{\banachseqspace{2}}	
\newcommand{\banachfunspace}[1]{\setL^{#1}}	
\safemath{\hilfunspace}{\banachfunspace{2}}	
\safemath{\hilfunspacep}{\hilfunspace(\complexset)}	
\safemath{\schwarzspace}{\setS}				

\newcommand{\hadj}[1]{#1^{\star}}			

\safemath{\SNR}{\text{\sc snr}} 				
\safemath{\SINR}{\text{\sc sinr}} 				
\safemath{\No}{N_0}							
\safemath{\Es}{E_s}							
\safemath{\Eb}{E_b}							
\safemath{\EbNo}{\frac{\Eb}{\No}}
\safemath{\EsNo}{\frac{\Es}{\No}}
\safemath{\NoVar}{\variance}                 

\let\time\undefined
\safemath{\time}{\sct}						
\safemath{\dtime}{\sck}						
\safemath{\delay}{\sctau}					
\safemath{\ddelay}{\scl}						
\safemath{\doppler}{\scnu}					
\safemath{\ddoppler}{\scm}					
\safemath{\freq}{\scf}						
\safemath{\dfreq}{\scn}						
\safemath{\Dtime}{\Delta\time}
\safemath{\Dfreq}{\Delta\freq}
\safemath{\Ddtime}{\dtime}
\safemath{\Ddfreq}{\dfreq}
\safemath{\bandwidth}{\scB}
\safemath{\maxdoppler}{\doppler_{0}}			
\safemath{\maxdelay}{\delay_{0}}				
\safemath{\spread}{\Delta_{\CHop}}			
\DeclareMathOperator{\CHop}{\ensuremath{\opH}} 
\safemath{\kernelp}{\kernel(\time,\time')}	
\safemath{\tvir}{\rndh_{\CHop}}				
\safemath{\tvirp}{\tvir(\time,\delay)}		
\safemath{\tvirc}{\conj{\rndh}_{\CHop}}		
\safemath{\tvtf}{\rndl_{\CHop}}				
\safemath{\tvtfp}{\tvtf(\time,\freq)}			
\safemath{\tvtfc}{\conj{\rndl}_{\CHop}}		
\safemath{\spf}{\rnds_{\CHop}}				
\safemath{\spfp}{\spf(\doppler,\delay)}		
\safemath{\spfc}{\conj{\rnds}_{\CHop}}		
\safemath{\bff}{\rndb_{\CHop}}				
\safemath{\bffp}{\bff(\doppler,\freq)}		

\safemath{\irc}{\scr_{\rndh}}				
\safemath{\tfc}{\scr_{\rndl}}				
\safemath{\spc}{\scr_{\rnds}}				
\safemath{\bfc}{\scr_{\rndb}}				

\safemath{\scaf}{\scc_{\rnds}}				
\safemath{\scafp}{\scaf(\doppler,\delay)}		
\safemath{\ccf}{\scc_{\rndl}}				
\safemath{\ccfp}{\ccf(\Dtime,\Dfreq)}			
\safemath{\cic}{\scc_{\rndh}}				
\safemath{\cicp}{\cic(\Dtime,\delay)}			

\safemath{\mi}{\scI}							
\safemath{\capacity}{\scC}					

\DeclareMathOperator{\Prob}{\opP}		

\safemath{\normal}{\mathcal{N}}			
\safemath{\jpg}{\mathcal{CN}}			
\safemath{\uniform}{\mathcal{U}}				
\safemath{\mchain}{\leftrightarrow}		
\newcommand{\given}{\,\vert\,}				

\safemath{\dB}{\,\mathrm{dB}}
\safemath{\dBm}{\,\mathrm{dBm}}
\safemath{\Hz}{\,\mathrm{Hz}}
\safemath{\kHz}{\,\mathrm{kHz}}
\safemath{\MHz}{\,\mathrm{MHz}}
\safemath{\GHz}{\,\mathrm{GHz}}
\safemath{\s}{\,\mathrm{s}}
\safemath{\ms}{\,\mathrm{ms}}
\safemath{\mus}{\,\mathrm{\text{\textmu}s}}
\safemath{\ns}{\,\mathrm{ns}}
\safemath{\ps}{\,\mathrm{ps}}
\safemath{\meter}{\,\mathrm{m}}
\safemath{\mm}{\,\mathrm{mm}}
\safemath{\cm}{\,\mathrm{cm}}
\safemath{\m}{\,\mathrm{m}}
\safemath{\W}{\,\mathrm{W}}
\safemath{\mW}{\, \mathrm{mW}}
\safemath{\J}{\,\mathrm{J}}
\safemath{\K}{\,\mathrm{K}}
\safemath{\bit}{\,\mathrm{bit}}
\safemath{\nat}{\,\mathrm{nat}}

\safemath{\define}{\triangleq}					

\providecommand{\inner}[2]{\ensuremath{\left\langle#1,#2\right\rangle}}
\safemath{\equivalent}{\sim}
\safemath{\distas}{\sim}					
\safemath{\sdiff}{\Delta}				
\safemath{\setdiff}{\setminus}				

\safemath{\reals}{\mathbb R}
\safemath{\positivereals}{\reals^{+}}
\safemath{\integers}{\mathbb Z}
\safemath{\posint}{\integers^{+}}
\safemath{\naturals}{\mathbb N}
\safemath{\posnaturals}{\naturals^{+}}
\safemath{\complexset}{\mathbb C}
\safemath{\rationals}{\mathbb Q}

\safemath{\iSet}{\setI}
\safemath{\rel}{\bowtie}					
\safemath{\eqrel}{\sim}					
\safemath{\rlord}{\leq}					
\safemath{\slord}{<}						
\safemath{\rpord}{\preceq}				
\safemath{\rrpord}{\succeq}				
\safemath{\spord}{\prec}					
\safemath{\sig}{\sigma}					
\safemath{\metric}{d}					

\safemath{\setfun}{\Phi}					
\safemath{\measure}{\mu}					
\safemath{\altmeasure}{\lambda}					
\newcommand{\outerm}[1]{#1^{\star}}		
\newcommand{\innerm}[1]{#1_{\star}}		
\safemath{\omeasure}{\outerm{\measure}}		
\safemath{\imeasure}{\innerm{\measure}}		
\safemath{\aecol}{\colS^{\star}_{\measure}} 
\safemath{\emeasure}{\bar{\measure}_{0}}	
\safemath{\rmeasure}{\tilde{\measure}}	
\safemath{\bmeasure}{\measure_{0}}		
\safemath{\glength}{\measure_{\altfun}}	
\safemath{\lebmea}{\lambda}				
\safemath{\blebmea}{\lebmea_{0}}			
\safemath{\sfun}{s}						
\safemath{\absintspace}{\colL^{1}}		
\safemath{\sqintspace}{\banachfunspace{2}}		
\safemath{\abssumspace}{l^{1}}		
\safemath{\sqsumspace}{l^{2}}		

\safemath{\sfield}{\setF}				
\safemath{\vectors}{\setV}				
\safemath{\vecspace}{(\vectors,\sfield)}	
\safemath{\linspace}{\setV}				
\safemath{\clinspace}{(\linspace,\sfield)} 
\safemath{\nspace}{\setU}				
\safemath{\metspace}{\setM}				
\safemath{\bspace}{\setB}				
\safemath{\ipspace}{\setG}				
\safemath{\hilspace}{\setH}				
\safemath{\blospace}{\setG}				
\safemath{\lop}{\opT}					
\safemath{\altlop}{\opS}					
\safemath{\nullsp}{\nullspace(\lop)}		
\safemath{\lfun}{l}						
\safemath{\altlfun}{g}					
\newcommand{\dual}[1]{#1^{'}}			
\newcommand{\ocomp}[1]{#1^{\orth}}		
\safemath{\dsum}{\oplus}					
\safemath{\funspace}{\colL}				
\renewcommand{\adj}[1]{#1^{\times}}		
\safemath{\adjlop}{\adj{\lop}}			
\safemath{\hadjlop}{\hadj{\lop}}			
\safemath{\tow}{\xrightarrow{w}}			
\safemath{\tows}{\xrightarrow{w^{*}}}		
\safemath{\cparam}{\lambda}
\safemath{\lopl}{\lop_{\cparam}}		
\safemath{\iop}{\opI}					
\safemath{\resolop}{\opR}				
\safemath{\resolvent}{\resolop_{\cparam}(\lop)}	
\safemath{\reset}{\setQ}
\safemath{\spectrum}{\setS}
\safemath{\resolset}{\reset(\lop)}		
\safemath{\lopspec}{\spectrum(\lop)}		
\safemath{\pspec}{\spectrum_{p}(\lop)}	
\safemath{\cspec}{\spectrum_{c}(\lop)}	
\safemath{\rspec}{\spectrum_{r}(\lop)}	
\newcommand{\specrad}[1]{r_{#1}}			
\safemath{\lopsrad}{\specrad{\lop}}		
\safemath{\pop}{\opP}					

\safemath{\specfam}{\colE}				
\safemath{\specop}{\opE_{\cparam}}		
\safemath{\altspecop}{\opE_{\mu}}		
\safemath{\vmulti}{\vecone}				
\safemath{\unitaryop}{\opU}				
\safemath{\sval}{\sigma}					
\safemath{\corrcoef}{\rho}				
\safemath{\sangle}{\theta}				

\let\time\undefined
\safemath{\iset}{\setI}				
\safemath{\shift}{\nu}
\safemath{\scale}{\alpha}
\safemath{\time}{t}
\safemath{\specfreq}{\theta}	
\newcommand{\transopgen}[1]{\opT_{#1}} 
\safemath{\transop}{\transopgen{\delay}}
\newcommand{\modopgen}[1]{\opM_{#1}}	
\safemath{\modop}{\modopgen{\shift}}
\newcommand{\dilopgen}[1]{\opD_{#1}}	
\safemath{\dilop}{\dilopgen{\scale}}
\safemath{\fram}{\setF}				
\safemath{\dfram}{\dual{\fram}}		
\safemath{\ufb}{B}					
\safemath{\lfb}{A}					
\safemath{\sop}{\hadj{\aop}}				
\safemath{\aop}{\opT}			
\safemath{\fop}{\opS}				
\safemath{\daop}{\tilde\opT}			
\safemath{\dsop}{\hadj{\tilde\opT}}				

\safemath{\ifop}{\inv{\fop}}			
\safemath{\rifop}{\fop^{-1/2}}			
\newcommand{\ft}[1]{\widehat{#1}}	
\newcommand{\ftd}[1]{\widehat{#1}_d}	
\safemath{\transeq}{\setT}			
\safemath{\nfun}{\Phi}				
\safemath{\funvec}{\vecf}			
\safemath{\altfunvec}{\vecg}

\safemath{\samplespace}{\Omega}
\safemath{\probspace}{(\samplespace,\sfield,\Prob)}	
\safemath{\ccoef}{\rho}			

\safemath{\infstate}{\vecpi}				
\safemath{\typset}{\setA_{\epsilon}^{(N)}}	
\safemath{\expequal}{\doteq}				
\safemath{\code}{C}						
\safemath{\dstringset}{\setD^{\star}}		
\safemath{\cwlength}{l}					
\safemath{\elength}{L}					
\safemath{\extension}{C^{\star}}			
\safemath{\approaches}{\rightarrow}		
\safemath{\evnt}{\setA}					
\safemath{\altevnt}{\setB}					
\safemath{\altrv}{\rndy}					
\safemath{\complexrv}{\rndu}					
\safemath{\altcrv}{\rndv}				
\safemath{\rvec}{\rvecx}					
\safemath{\altrvec}{\rvecy}				
\safemath{\crvec}{\rvecu}				
\safemath{\altcrvec}{\rvecv}				
\safemath{\map}{T}						
\safemath{\jacobian}{J}					
\safemath{\wvec}{\rvecw}					
\safemath{\wrv}{\rndw}					
\safemath{\orthmat}{\matQ}				
\safemath{\evmat}{\matLambda}			
\safemath{\identity}{\matidentity}		
\safemath{\innovec}{\vecv}				
\safemath{\convas}{\xrightarrow{\text{a.s.}}}	
\safemath{\convr}{\xrightarrow{\text{r}}}	
\safemath{\convp}{\xrightarrow{\text{P}}}	
\safemath{\convd}{\xrightarrow{\text{D}}}	
\safemath{\ltis}{\opL}				
\safemath{\ir}{h}					
\safemath{\tf}{\MakeUppercase{\ir}}	

\newcommand*{\fancyrefparlabelprefix}{par}		
\newcommand*{\fancyrefchalabelprefix}{cha}		
\newcommand*{\fancyrefapplabelprefix}{app}		
\newcommand*{\fancyrefthmlabelprefix}{thm}		
\newcommand*{\fancyreflemlabelprefix}{lem}		
\newcommand*{\fancyrefcorlabelprefix}{cor}		
\newcommand*{\fancyrefdeflabelprefix}{dfn}		
\newcommand*{\fancyrefproplabelprefix}{prop}		
\newcommand*{\fancyrefexlabelprefix}{ex}		
\newcommand*{\fancyrefxcalabelprefix}{xca}		

\frefformat{vario}{\fancyrefparlabelprefix}{Part~#1}
\frefformat{vario}{\fancyrefchalabelprefix}{Chapter~#1}
\frefformat{vario}{\fancyrefseclabelprefix}{Section~#1}
\frefformat{vario}{\fancyrefthmlabelprefix}{Theorem~#1}
\frefformat{vario}{\fancyreflemlabelprefix}{Lemma~#1}
\frefformat{vario}{\fancyrefcorlabelprefix}{Corollary~#1}
\frefformat{vario}{\fancyrefdeflabelprefix}{Definition~#1}
\frefformat{vario}{\fancyreffiglabelprefix}{Figure~#1}
\frefformat{vario}{\fancyrefapplabelprefix}{Appendix~#1}
\frefformat{vario}{\fancyrefeqlabelprefix}{(#1)}
\frefformat{vario}{\fancyrefproplabelprefix}{Property~#1}
\frefformat{vario}{\fancyrefexlabelprefix}{Example~#1}
\frefformat{vario}{\fancyrefxcalabelprefix}{Exercise~#1}

\safemath{\vsignal}{\vecx}				
\safemath{\altvsignal}{\vecy}			
\safemath{\csignal}{\fun}				
\safemath{\altcsignal}{\altfun}			
\safemath{\analop}{\opT}				
\safemath{\lianalop}{\opL}				
\safemath{\dualanalop}{\tilde\opT}				
\safemath{\manalop}{\matT}				
\safemath{\analcrossop}{{\opT^{\times}}}			
\safemath{\synthdualop}{\ad{\tilde\opT}}		
\safemath{\msynthdualop}{\tp{\tilde\matT}}		
\safemath{\msynthdualopc}{\herm{\tilde\matT}}		
\safemath{\msynthop}{\tp{\matT}}		
\safemath{\analdualop}{\tilde\opT}		
\safemath{\vbasisel}{\vece}			
\safemath{\vsynthbasisel}{\tilde\vece}	
\safemath{\synthbasisel}{\tilde e}	
\safemath{\basisel}{e}	
\safemath{\vanalframeel}{\vecg}			
\safemath{\vanalframeelc}{g}			
\safemath{\vsynthframeel}{\tilde\vecg}	
\safemath{\analframeel}{g}			
\safemath{\analframeelrev}{\tilde g}			
\safemath{\synthframeel}{\tilde g}	
\safemath{\synthframeelrev}{g}	
\safemath{\cframeel}{g}					
\safemath{\vcoef}{\vecc}				
\safemath{\coeflt}{\{c_k\}_{k\in\frameset}}				
\safemath{\coefaltlt}{\{a_k\}_{k\in\frameset}}				
\safemath{\coef}{c}						
\safemath{\coefalt}{a}					
\safemath{\dimension}{M}				
\safemath{\framesize}{N}				
\safemath{\frameset}{\setK}				
\safemath{\frameA}{A}					
\safemath{\frameB}{B}					
\safemath{\frameAdual}{\tilde A}					
\safemath{\frameBdual}{\tilde B}					
\safemath{\frameop}{\opS}				
\safemath{\mframeop}{\matS}				
\safemath{\dualframeop}{\tilde\opS}				
\safemath{\dualframeoprev}{\opS}				
\safemath{\sinterval}{T}				
\safemath{\diractrain}{\dirac_T}        
\safemath{\ftdiractrain}{\ft{\dirac}_T}        
\safemath{\funsampled}{\fun_T}        	
\safemath{\ftfunsampled}{\ft{\fun}_T}        	
\safemath{\LPfiltert}{h_{\text{LP}}} 
\safemath{\arbfiltert}{h}				
\safemath{\outfiltert}{h_{\text{out}}} 
\safemath{\outfilterf}{{\ft h}_{\text{out}}} 

\safemath{\LPfilterf}{\ft{h}_{\text{LP}}} 
\safemath{\cutoff}{f_c} 				
\safemath{\oversampling}{K} 			
\safemath{\sqintspaceBL}{\hilfunspace(\bandwidth)} 	
\safemath{\qnoise}{w} 					
\safemath{\qnoisevar}{\sigma^2} 					
\safemath{\mseo}{\sigma^2_{\text{oversampling}}} 					
\safemath{\msec}{\sigma^2_{\text{critical}}} 					
\safemath{\mseb}{\sigma^2_{\text{basis}}} 					
\safemath{\msef}{\sigma^2_{\text{frame}}} 					
\safemath{\setframe}{\{\analframeel_k\in\hilspace\}_{k\in\frameset}} 					
\safemath{\setframem}{\{\analframeel_k\}_{k\in\frameset},\, \analframeel_k\in\hilspace,\, k\in\frameset,}					
\safemath{\setframemm}{\{\analframeel_k\}_{k\in\frameset}}					
\safemath{\setdualframe}{\{\synthframeel_k\in\hilspace\}_{k\in\frameset}} 					
\safemath{\setdualframemm}{\{\synthframeel_k\}_{k\in\frameset}} 					
\safemath{\setframerev}{\{\analframeelrev_k\in\hilspace\}_{k\in\frameset}} 					
\safemath{\setframerevmm}{\{\analframeelrev_k\}_{k\in\frameset}} 					
\safemath{\setdualframerev}{\{\synthframeelrev_k\in\hilspace\}_{k\in\frameset}} 					
\safemath{\setdualframerevmm}{\{\synthframeelrev_k\}_{k\in\frameset}} 					
\safemath{\weylop}{\opW}				
\safemath{\window}{g}				
\safemath{\waveop}{\opV}				
\safemath{\arb}{\mathrm{arb}}				

\safemath{\duration}{D}			
\safemath{\dtdf}{\dtime,\dfreq}  
\safemath{\altdtdf}{\altdtime,\altdfreq}
\safemath{\altdtime}{l}				
\safemath{\altdfreq}{m}				

\safemath{\dtimetilde}{\tilde{\dtime}}				
\safemath{\dfreqtilde}{\tilde{\dfreq}}				
\safemath{\altdtimetilde}{\tilde{\altdtime}}				
\safemath{\altdfreqtilde}{\tilde{\altdfreq}}				

\safemath{\maxDoppler}{\doppler_{0}}		
\safemath{\altmaxDoppler}{\widetilde{\maxDoppler}}
\safemath{\maxDelay}{\delay_{0}}			
\safemath{\altmaxDelay}{\widetilde{\maxDelay}}
\safemath{\altspread}{\widetilde{\Delta}_{\CHop}}
\safemath{\tstep}{T}				
\safemath{\fstep}{F}				
\safemath{\alttstep}{\widetilde{\tstep}}				
\safemath{\altfstep}{\widetilde{\fstep}}				
\safemath{\tfstep}{\tstep\fstep}	
\safemath{\srtfstep}{\sqrt{\tfstep}}
\safemath{\fslots}{N}			
\safemath{\tslots}{K}			
\safemath{\tslotstot}{\widetilde{\tslots}}
\safemath{\fslotstot}{\widetilde{\fslots}}
\safemath{\tslotsguard}{K'}
\safemath{\tfslots}{\tslots\fslots}	
\safemath{\tfsamples}{\dtime\tstep,\dfreq\fstep}	
\safemath{\WHset}{(\logon,\tstep,\fstep)}
\safemath{\WHsetsquare}{(\logon,\srtfstep,\srtfstep)}
\safemath{\spreadset}{\setD}
\safemath{\altspreadset}{\setE}
\safemath{\spreadsquareset}{\widetilde{\spreadset}}

\safemath{\logon}{g}						
\safemath{\logonp}{\logon(\time)}
\safemath{\altlogon}{\widetilde{\logon}}						
\safemath{\altaltlogon}{e}
\safemath{\altlogonp}{\altlogon(\time)}
\safemath{\altaltlogonp}{\altaltlogon(\time)}
\safemath{\slogon}{\logon_{\dtime,\dfreq}}	
\safemath{\slogonp}{\slogon(\time)}	
\safemath{\slogonset}{\left\{\logon(\time-\dtime\srtfstep)\cex{\dfreq\srtfstep\time}\right\}}
\safemath{\eftime}{a_{0}}	
\safemath{\efband}{b_{0}} 
\safemath{\logonf}{G} 
\safemath{\logonfp}{\logonf(\freq)}
\safemath{\logonalt}{f}
\safemath{\logonaltp}{\logonalt(\time)}
\safemath{\slogonct}{\logon_{(\alpha,\beta)}}
\safemath{\slogonctalt}{\logon_{(\alpha',\beta')}}
\safemath{\af}{A}						
\safemath{\afp}{\af_{\logon}(\doppler,\delay)}	
\safemath{\setloc}{\setG}
\safemath{\setloclogon}{\widetilde{\setG}}
\safemath{\dertau}{b}
\safemath{\dernu}{a}
\safemath{\dertaup}{\dertau_{\dtdf}}
\safemath{\dernup}{\dernu_{\dtdf}}
\safemath{\constM}{c_{M}}
\safemath{\constm}{c_{m}}
\safemath{\corr}{d}
\safemath{\corrmat}{\matD}

\safemath{\sumdtime}{\sum_{\dtime=-\infty}^{\infty}} 
\safemath{\sumdfreq}{\sum_{\dfreq=-\infty}^{\infty}} 

\safemath{\sumdtimelimited}{\sum_{\dtime=-\tslots}^{\tslots}} 
\safemath{\sumdfreqlimited}{\sum_{\dfreq=-\fslots}^{\fslots}} 

\safemath{\altsumdtimelimited}{\sum_{\altdtime=-\tslots}^{\tslots}} 
\safemath{\altsumdfreqlimited}{\sum_{\altdfreq=-\fslots}^{\fslots}} 

 \safemath{\limintime}{\lim_{\tslotstot\to\infty}}	

\safemath{\euler}{\gamma}	
\safemath{\intdiscrete}{\int_{-1/2}^{1/2}}
\safemath{\intdiscretetwod}{\intdiscrete\intdiscrete}
\safemath{\hilfuntimeband}{\hilfunspace(\bandwidth,\duration,\spillover)}
\safemath{\spillover}{\eta}
\safemath{\hilfunband}{\hilfunspace(\bandwidth)}
\newcommand{\timetruncfilter}[1]{p_{#1}}		
\safemath{\alttimetruncfilter}{p_{\duration+2\maxDelay}}		
\safemath{\timetruncfilterp}{\timetruncfilter(\time)}
\safemath{\onb}{\phi}
\safemath{\indexonb}{m}
\safemath{\numonb}{M}
\safemath{\onbp}{\onb_{\indexonb}(\time)}
\safemath{\funloc}{f}
\safemath{\funlocp}{\funloc(\dd)}
\safemath{\ofun}{\phi}
\safemath{\altofun}{\ofun'}
\safemath{\oindex}{m}
\safemath{\ofunp}{\ofun_\oindex(\time)}
\safemath{\altofunp}{\altofun_\oindex(\time)}
\safemath{\onum}{M} 

\newcommand{\citepar}[1]{\hspace{-0.06mm}\cite{#1}}

\usepackage[a4paper,left=2cm,hmarginratio=1:1,vmargin=3cm]{geometry}
\begin{document}


\title{A Short Course on Frame Theory}



 \author{Veniamin I. Morgenshtern and  Helmut B\"{o}lcskei \vspace{0.3cm}
\\
ETH Zurich, 8092 Zurich, Switzerland\\
 E-mail: \{vmorgens, boelcskei\}@nari.ee.ethz.ch \vspace{0.2cm}
}

\maketitle
\acrodef{ONB}[ONB]{orthonormal basis}
\acrodef{ada}[A/D]{analog-to-digital}
\acrodef{daa}[D/A]{digital-to-analog}
\acrodef{dfta}[DFT]{discrete Fourier transform}
\acrodef{dtfta}[DTFT]{discrete-time Fourier transform}
\acrodef{cdmaa}[CDMA]{code division multiple access}
\acrodef{ofdma}[OFDM]{orthogonal frequency division multiplexing}
\acrodef{wha}[WH]{Weyl-Heisenberg}

\newcommand{\onbac}{\ac{ONB}\xspace}
\newcommand{\onbacp}{\acp{ONB}\xspace}
\newcommand{\adac}{\ac{ada}\xspace}
\newcommand{\daac}{\ac{daa}\xspace}
\newcommand{\dftac}{\ac{dfta}\xspace}
\newcommand{\dtftac}{\ac{dtfta}\xspace}
\newcommand{\dtftacp}{\acp{dtfta}\xspace}
\newcommand{\cdmaac}{\ac{cdmaa}\xspace}
\newcommand{\ofdmac}{\ac{ofdma}\xspace}
\newcommand{\whac}{\ac{wha}\xspace}


\label{sec:introduction}
Hilbert spaces~\cite[Def. 3.1-1]{kreyszig89} and the associated concept of orthonormal bases 
are of fundamental importance in signal processing, communications, control, and information theory. However, linear independence and orthonormality of the basis elements impose constraints that often make it difficult to have the basis elements satisfy additional desirable properties. This calls for a theory of signal decompositions that is flexible enough to accommodate decompositions into possibly nonorthogonal  and redundant signal sets. The theory of frames provides such a tool.

This chapter is an introduction to the theory of frames, which was developed by Duffin and Schaeffer~\cite{duffin52} and popularized mostly through~\cite{daubechies90-09, daubechies92,heil89-12, young80}. Meanwhile frame theory, in particular the aspect of redundancy in signal expansions, has found numerous applications such as, e.g., denoising~\cite{donoho95-03,donoho94-08}, \cdmaac~\cite{rupf94-07},
 \ofdmac 
systems~\cite{sandell96-09}, coding theory~\cite{heath02-12,rudelson05}, quantum information theory~\cite{eldar02-03}, \adac converters~\cite{bolcskei01-01,bolcskei97-11,benedetto06-05}, and compressive sensing~\cite{donoho06-04, candes06-12, donoho02-03}. A more extensive list of relevant references can be found in~\cite{kovacevic08}.
For a  comprehensive treatment of frame theory we refer to the excellent textbook~\cite{christensen96}.

\section{Examples of Signal Expansions}
We start by considering some simple motivating examples.
\begin{example}[Orthonormal basis in $\reals^{2}$]
\label{ex:orthnormalsystem}
Consider the \onbac 
\begin{equation*}
\vbasisel_1=\begin{bmatrix} 1 \\ 0 \end{bmatrix}, \quad \quad \vbasisel_2=\begin{bmatrix} 0 \\ 1 \end{bmatrix}
\end{equation*}
in $\reals^{2}$ (see~\Figref{fig:orthr2}).
\begin{figure}[t]
	\centering
	\includegraphics[]{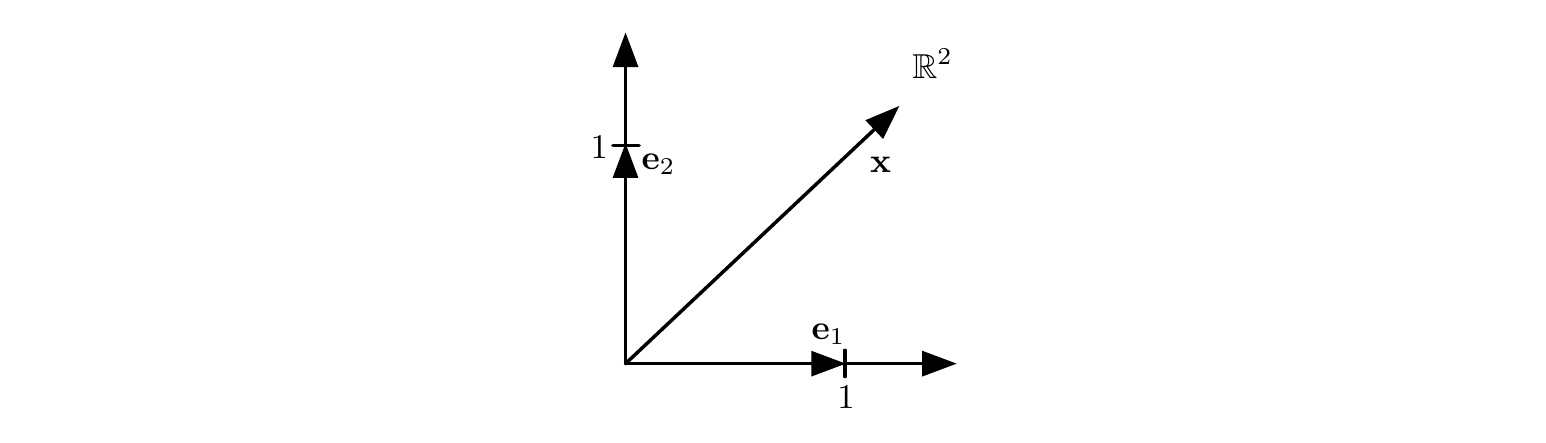}
	\caption{Orthonormal basis in $\reals^2$.}
	\label{fig:orthr2}
\end{figure}
We can represent every signal $\vsignal\in\reals^{2}$ as the following linear combination of the basis vectors~$\vbasisel_1$ and~$\vbasisel_2$:
\begin{equation}
\label{eq:decompr2}
\vsignal=\inner{\vsignal}{\vbasisel_1}\vbasisel_1+\inner{\vsignal}{\vbasisel_2}\vbasisel_2.
\end{equation}
To rewrite~\fref{eq:decompr2} in vector-matrix notation, we start by defining the vector of expansion coefficients as
\begin{equation*}
	\vcoef=\begin{bmatrix} \coef_1 \\ \coef_2 \end{bmatrix}\define	\begin{bmatrix} \inner{\vsignal}{\vbasisel_1} \\ \inner{\vsignal}{\vbasisel_2} \end{bmatrix}=\begin{bmatrix} \tp{\vbasisel_1} \\ \tp{\vbasisel_2} \end{bmatrix} \vsignal =\begin{bmatrix} 1&0\\0&1 \end{bmatrix} \vsignal.
\end{equation*}
It is convenient to define the matrix
\begin{equation*}
	\manalop\define\begin{bmatrix} \tp{\vbasisel_1} \\ \tp{\vbasisel_2} \end{bmatrix} =\begin{bmatrix} 1&0\\0&1 \end{bmatrix}.
\end{equation*}
Henceforth we call \manalop the \emph{analysis matrix}; it  multiplies the signal \vsignal to produce the expansion coefficients
\begin{equation*}
	\vcoef=\manalop \vsignal.
\end{equation*}
Following~\fref{eq:decompr2}, we can reconstruct the signal \vsignal from the coefficient vector $\vcoef$ according to 
\begin{equation}
	\label{eq:synthstage}
	\vsignal=\msynthop\vcoef=\begin{bmatrix} \vbasisel_1 & \vbasisel_2 \end{bmatrix}\vcoef=\begin{bmatrix} \vbasisel_1 & \vbasisel_2 \end{bmatrix}\begin{bmatrix} \inner{\vsignal}{\vbasisel_1} \\ \inner{\vsignal}{\vbasisel_2} \end{bmatrix} =\inner{\vsignal}{\vbasisel_1}\vbasisel_1+\inner{\vsignal}{\vbasisel_2}\vbasisel_2.
\end{equation}
We call
\begin{equation}
	\label{eq:msynthop1}
	\msynthop=\begin{bmatrix} \vbasisel_1 & \vbasisel_2 \end{bmatrix}=\begin{bmatrix} 1&0\\0&1 \end{bmatrix}
\end{equation}
the \emph{synthesis matrix}; it multiplies the coefficient vector $\vcoef$ to recover the signal $\vsignal$.
It follows from~\fref{eq:synthstage} that~\fref{eq:decompr2} is equivalent to 
\begin{equation}
	\label{eq:signaldecomporth}
	\vsignal=\msynthop\manalop \vsignal=\begin{bmatrix} 1&0\\0&1 \end{bmatrix}\begin{bmatrix} 1&0\\0&1 \end{bmatrix}\vsignal.
\end{equation}
\end{example}

The introduction of the analysis and the synthesis matrix in the example above may seem artificial and may appear as complicating matters unnecessarily. After all, both \manalop and \msynthop are equal to the identity matrix in this example. We will, however, see shortly that this notation paves the way to developing a unified framework for nonorthogonal and redundant signal expansions. Let us now look at a somewhat more interesting example. 

\begin{example}[Biorthonormal bases in $\reals^{2}$]
	\label{ex:biorthogonal}
Consider two noncollinear unit norm vectors in $\reals^{2}$. For concreteness, take (see \Figref{fig:nonorthr2})
\begin{equation*}
\vbasisel_1=\begin{bmatrix} 1 \\ 0 \end{bmatrix}, \quad \quad \vbasisel_2=\frac{1}{\sqrt{2}}\begin{bmatrix} 1 \\ 1 \end{bmatrix}.
\end{equation*}
\begin{figure}[t]
	\centering
	\includegraphics[]{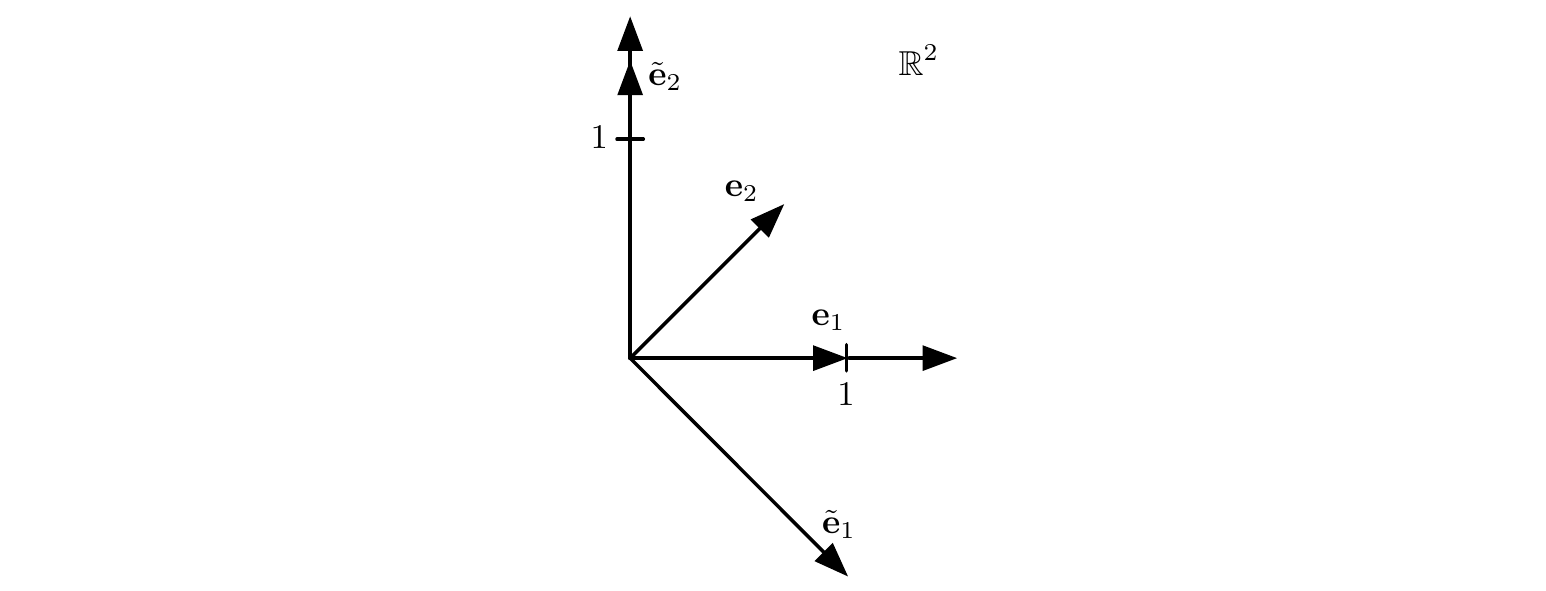}
	\caption{Biorthonormal bases in $\reals^2$.}
	\label{fig:nonorthr2}
\end{figure}
For an arbitrary signal $\vsignal\in\reals^2$, we can compute the expansion coefficients 
\begin{align*}
	\coef_1&\define\inner{\vsignal}{\vbasisel_1}\\
	\coef_2&\define\inner{\vsignal}{\vbasisel_2}.
\end{align*}
As in  \fref{ex:orthnormalsystem} above, we stack the expansion coefficients into a vector so that 
\begin{equation}
	\label{eq:vcoefnonorth}
	\vcoef=\begin{bmatrix} \coef_1 \\ \coef_2 \end{bmatrix}=\begin{bmatrix} \inner{\vsignal}{\vbasisel_1} \\ \inner{\vsignal}{\vbasisel_2} \end{bmatrix}=\begin{bmatrix} \tp{\vbasisel_1} \\ \tp{\vbasisel_2} \end{bmatrix} \vsignal =\begin{bmatrix} 1&0\\1/\sqrt{2}&1/\sqrt{2} \end{bmatrix} \vsignal.
\end{equation}
Analogously to \fref{ex:orthnormalsystem}, we can define the analysis matrix 
\begin{equation*}
	\manalop\define \begin{bmatrix} \tp{\vbasisel_1} \\ \tp{\vbasisel_2} \end{bmatrix}=\begin{bmatrix} 1&0\\1/\sqrt{2}&1/\sqrt{2} \end{bmatrix}
\end{equation*}
and rewrite \fref{eq:vcoefnonorth} as
\begin{equation*}
	\vcoef=\manalop\vsignal.
\end{equation*}
Now, obviously, the vectors $\vbasisel_1$ and $\vbasisel_2$ are not orthonormal (or, equivalently, $\manalop$ is not unitary) so that we cannot write $\vsignal$ in the form~\fref{eq:decompr2}.
We could, however, try to find a decomposition of \vsignal of the form
\begin{equation}
\vsignal=\inner{\vsignal}{\vbasisel_1}\vsynthbasisel_1+\inner{\vsignal}{\vbasisel_2}\vsynthbasisel_2
\label{eq:basisdecompositionreq}
\end{equation}
with $\vsynthbasisel_1, \vsynthbasisel_2\in\reals^2$. 
That this is, indeed, possible is easily seen by rewriting~\fref{eq:basisdecompositionreq} according to
\begin{equation}
\vsignal = \begin{bmatrix} \vsynthbasisel_1 & \vsynthbasisel_2 \end{bmatrix} \manalop\vsignal
	\label{eq:signalmatrixdecomposition}
\end{equation}
and choosing the vectors $\vsynthbasisel_1$ and $\vsynthbasisel_2$ to be given by the columns of $\manalop^{-1}$ according to
\begin{equation}
	\label{eq:dualisinv}
	\begin{bmatrix} \vsynthbasisel_1 & \vsynthbasisel_2 \end{bmatrix} =\manalop^{-1}.
\end{equation}
Note that $\manalop$ is invertible as a consequence of $\vbasisel_1$ and $\vbasisel_2$ not being collinear. For the specific example at hand we find%
\begin{equation*}
\begin{bmatrix}  \vsynthbasisel_1 & \vsynthbasisel_2 \end{bmatrix}= \manalop^{-1}
	= \begin{bmatrix} 1&0\\-1&\sqrt{2}\end{bmatrix}
\end{equation*}
and therefore (see~\Figref{fig:nonorthr2})
\begin{equation*}
\vsynthbasisel_1=\begin{bmatrix} 1 \\ -1 \end{bmatrix}, \quad \quad \vsynthbasisel_2=\begin{bmatrix} 0 \\ \sqrt{2} \end{bmatrix}.
\end{equation*}
Note that~\fref{eq:dualisinv} implies that $\manalop\begin{bmatrix} \vsynthbasisel_1 & \vsynthbasisel_2 \end{bmatrix}=\identity_2$, which is equivalent to
\begin{equation*}
	\begin{bmatrix} \tp{\vbasisel_1} \\ \tp{\vbasisel_2} \end{bmatrix}\begin{bmatrix} \vsynthbasisel_1 & \vsynthbasisel_2 \end{bmatrix}=\identity_2.
\end{equation*}
More directly the two sets of vectors  $\{\vbasisel_1, \vbasisel_2\}$ and $\{\vsynthbasisel_1, \vsynthbasisel_2\}$ satisfy a ``biorthonormality'' property according to
\begin{equation*}
\inner{\vbasisel_j}{\vsynthbasisel_k} = \begin{cases} 1, \quad &j=k\\ 0,  &\text{else} \end{cases},\quad\quad j,k=1,2.
\end{equation*}
We say that $\{\vbasisel_1, \vbasisel_2\}$ and $\{\vsynthbasisel_1, \vsynthbasisel_2\}$ are biorthonormal bases. Analogously to~\fref{eq:msynthop1}, we can now define the synthesis matrix as follows: 
\begin{equation*}
\msynthdualop\define\begin{bmatrix} \vsynthbasisel_1 & \vsynthbasisel_2 \end{bmatrix} = \begin{bmatrix} 1&0\\-1&\sqrt{2}\end{bmatrix}.
\end{equation*}
Our observations can be summarized according to 
\begin{align}
	\label{eq:signaldecompnorth}
\vsignal&=\inner{\vsignal}{\vbasisel_1} \vsynthbasisel_1+\inner{\vsignal}{\vbasisel_2} \vsynthbasisel_2\nonumber\\
&=\msynthdualop\vcoef=\msynthdualop\manalop\vsignal\nonumber\\
&=\begin{bmatrix} 1&0\\-1&\sqrt{2}\end{bmatrix} \begin{bmatrix} 1&0\\1/\sqrt{2}&1/\sqrt{2} \end{bmatrix} \vsignal=\begin{bmatrix} 1&0\\0&1\end{bmatrix}\vsignal.
\end{align}
Comparing \fref{eq:signaldecompnorth} to \fref{eq:signaldecomporth}, we observe the following: To synthesize $\vsignal$ from the expansion coefficients \vcoef corresponding to the nonorthogonal set $\{\vbasisel_1, \vbasisel_2\}$, we need to use the synthesis matrix $\msynthdualop$ obtained from the set $\{ \vsynthbasisel_1,  \vsynthbasisel_2\}$, which forms a biorthonormal pair with $\{\vbasisel_1, \vbasisel_2\}$. In~\fref{ex:orthnormalsystem} $\{\vbasisel_1, \vbasisel_2\}$ is an \onbac and hence $\tilde\manalop=\manalop$, or, equivalently, $\{ \vbasisel_1, \vbasisel_2\}$ forms a biorthonormal pair with itself.

 As the vectors $\vbasisel_1$ and $\vbasisel_2$ are  linearly independent, the $2\times 2$ analysis matrix \manalop has full rank and is hence invertible, i.e., there is a \emph{unique} matrix $\manalop^{-1}$ that satisfies $\manalop^{-1}\manalop=\identity_2$. According to \fref{eq:signalmatrixdecomposition} this means that for each analysis set $\{\vbasisel_1, \vbasisel_2\}$ 
 there is precisely one synthesis set $\{ \vsynthbasisel_1, \vsynthbasisel_2\}$ such that \fref{eq:basisdecompositionreq} is satisfied for all $\vsignal\in\reals^2$.
\end{example}

So far we considered nonredundant signal expansions where the number of expansion coefficients is equal to the dimension of the Hilbert space. 
Often, however, redundancy in the expansion is desirable. 

\begin{example}[Overcomplete expansion in $\reals^{2}$, {\cite[Ex. 3.1]{kovacevic08}}]
	\label{ex:overcomplete}

Consider the following three vectors in $\reals^{2}$ (see~\Figref{fig:overcomplete}):
\begin{figure}[t]
	\centering
	\includegraphics[]{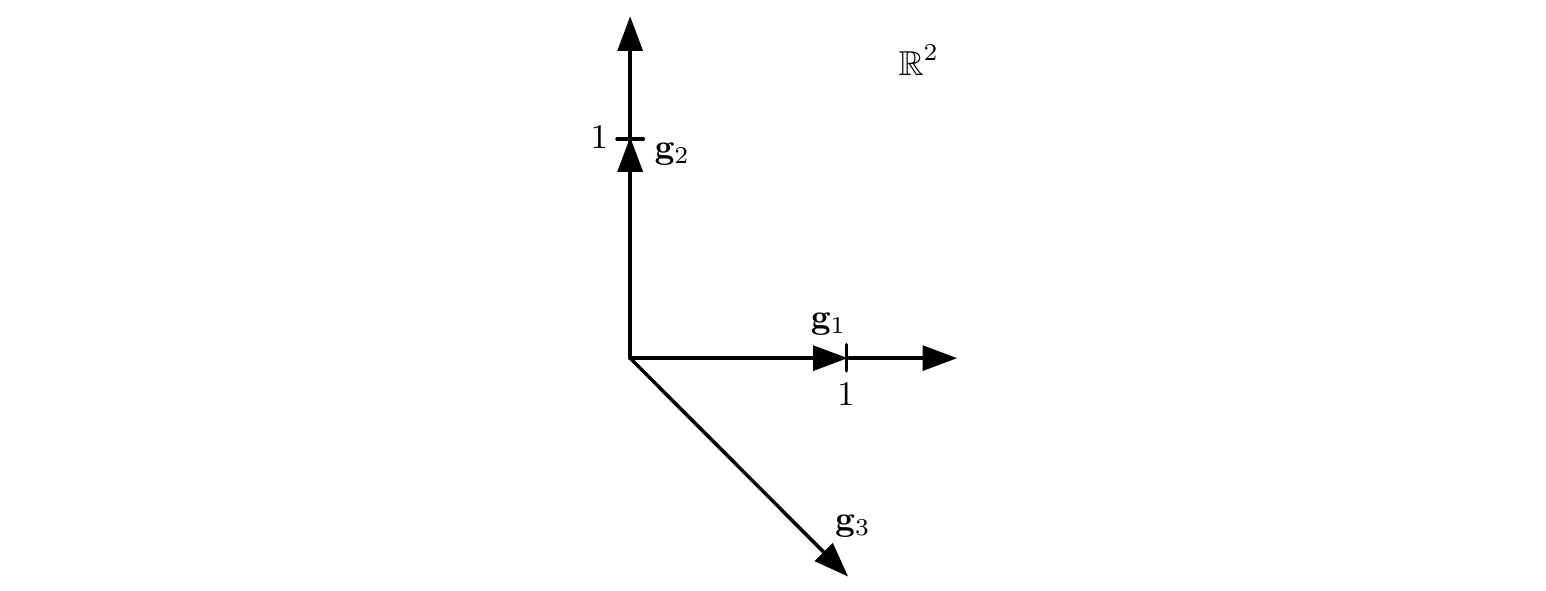}
	\caption{Overcomplete set of vectors in $\reals^2$.}
	\label{fig:overcomplete}
\end{figure}
\begin{equation*}
\vanalframeel_1=\begin{bmatrix} 1\\0\end{bmatrix}, \quad \vanalframeel_2=\begin{bmatrix} 0 \\1 \end{bmatrix}, \quad \vanalframeel_3=\begin{bmatrix} 1\\-1\end{bmatrix}.
\end{equation*}
Three vectors in a two-dimensional space are always linearly dependent. In particular, in this example we have $\vanalframeel_3=\vanalframeel_1-\vanalframeel_2$. 
Let us compute the expansion coefficients \vcoef corresponding to $\{\vanalframeel_1, \vanalframeel_2, \vanalframeel_3\}$:
\begin{equation}
	\label{eq:vcoefred}
	\vcoef=\begin{bmatrix} \coef_1 \\ \coef_2 \\ \coef_3 \end{bmatrix}\define	\begin{bmatrix} \inner{\vsignal}{\vanalframeel_1} \\ \inner{\vsignal}{\vanalframeel_2} \\ \inner{\vsignal}{\vanalframeel_3} \end{bmatrix}=\begin{bmatrix} \tp{\vanalframeel_1} \\ \tp{\vanalframeel_2} \\ \tp{\vanalframeel_3}\end{bmatrix} \vsignal =\begin{bmatrix} 1&0\\0 & 1 \\ 1 & -1 \end{bmatrix} \vsignal.
\end{equation}
Following Examples \ref{ex:orthnormalsystem} and \ref{ex:biorthogonal}, we define the analysis matrix 
\begin{equation*}
	\manalop\define\begin{bmatrix} \tp{\vanalframeel_1} \\ \tp{\vanalframeel_2} \\ \tp{\vanalframeel_3}\end{bmatrix}\define\begin{bmatrix} 1&0\\0 & 1 \\ 1 & -1 \end{bmatrix}
\end{equation*}
and rewrite \fref{eq:vcoefred} as
\begin{equation*}
	\vcoef=\manalop\vsignal.
\end{equation*}
Note that here, unlike in Examples \ref{ex:orthnormalsystem} and \ref{ex:biorthogonal}, \vcoef is a redundant representation of \vsignal as we have \emph{three} expansion coefficients for a \emph{two}-dimensional signal~\vsignal.

We next ask if $\vsignal$ can be represented as a linear combination of the form
\begin{equation}
	\label{eq:vsignalredundantdecomp}
	\vsignal=\underbrace{\inner{\vsignal}{\vanalframeel_1}}_{\coef_1} \vsynthframeel_1+\underbrace{\inner{\vsignal}{\vanalframeel_2}}_{\coef_2} \vsynthframeel_2+\underbrace{\inner{\vsignal}{\vanalframeel_3}}_{\coef_3}\vsynthframeel_3
\end{equation}
 with 
 $\vsynthframeel_1, \vsynthframeel_2, \vsynthframeel_3\in\reals^2$? 
To answer this question (in the affirmative) we first note that the vectors $\vanalframeel_1, \vanalframeel_2$ form an \onbac for $\reals^{2}$. We therefore  know that the following is true:
\begin{equation}
	\label{eq:vsignalbasisdec}
\vsignal=\inner{\vsignal}{\vanalframeel_1}\vanalframeel_1+\inner{\vsignal}{\vanalframeel_2}\vanalframeel_2.
\end{equation}
Setting 
\begin{equation*}
	\vsynthframeel_1=\vanalframeel_1,\ \vsynthframeel_2=\vanalframeel_2,\ \vsynthframeel_3=\veczero
\end{equation*} 
obviously yields a representation of the form~\fref{eq:vsignalredundantdecomp}. It turns out, however, that this representation is not unique and that an alternative representation of the form~\fref{eq:vsignalredundantdecomp} can be obtained as follows. We start by adding zero to the right-hand side of~\fref{eq:vsignalbasisdec}:
\begin{equation*}
\vsignal=\inner{\vsignal}{\vanalframeel_1}\vanalframeel_1+\inner{\vsignal}{\vanalframeel_2}\vanalframeel_2+\underbrace{\inner{\vsignal}{\vanalframeel_1-\vanalframeel_2}(\vanalframeel_1-\vanalframeel_1)}_\veczero.
\end{equation*}
Rearranging terms in this expression, we obtain
\begin{equation}
	\label{eq:tempexpansion}
\vsignal=\inner{\vsignal}{\vanalframeel_1}2\vanalframeel_1+\inner{\vsignal}{\vanalframeel_2}(\vanalframeel_2-\vanalframeel_1)-\inner{\vsignal}{\vanalframeel_1-\vanalframeel_2}\vanalframeel_1.
\end{equation}
We recognize that $\vanalframeel_1-\vanalframeel_2=\vanalframeel_3$ and set
\begin{equation}
	\label{eq:otherdual}
\vsynthframeel_1=2\vanalframeel_1, \ \vsynthframeel_2=\vanalframeel_2-\vanalframeel_1, \  \vsynthframeel_3=-\vanalframeel_1.
\end{equation}
This allows us to rewrite \fref{eq:tempexpansion} as
\begin{equation*}
	\vsignal=\inner{\vsignal}{\vanalframeel_1}\vsynthframeel_1+\inner{\vsignal}{\vanalframeel_2} \vsynthframeel_2+\inner{\vsignal}{\vanalframeel_3} \vsynthframeel_3.
\end{equation*}
The redundant set of vectors $\{\vanalframeel_1,\vanalframeel_2,\vanalframeel_3\}$ is called a \emph{frame}. The set $\{\vsynthframeel_1,\vsynthframeel_2,\vsynthframeel_3\}$ in \fref{eq:otherdual} 
is called a \emph{dual frame} to the frame $\{\vanalframeel_1,\vanalframeel_2,\vanalframeel_3\}$. Obviously another dual frame is given by  $\vsynthframeel_1=\vanalframeel_1$, $\vsynthframeel_2=\vanalframeel_2$, and~$\vsynthframeel_3=\veczero$. In fact, there are infinitely many dual frames. To see this, we first define the synthesis matrix corresponding to a dual frame $\{\vsynthframeel_1,\vsynthframeel_2,\vsynthframeel_3\}$ as
\begin{equation}
	\label{eq:msynthdualop}
	\msynthdualop\define\begin{bmatrix}\vsynthframeel_1 & \vsynthframeel_2&\vsynthframeel_3 \end{bmatrix}.
\end{equation}
It then follows that we can write
\begin{align*}
\vsignal&=\inner{\vsignal}{\vanalframeel_1} \vsynthframeel_1+\inner{\vsignal}{\vanalframeel_2} \vsynthframeel_2+\inner{\vsignal}{\vanalframeel_3}\vsynthframeel_3\nonumber\\
&=\msynthdualop\vcoef=\msynthdualop\manalop\vsignal, 
\end{align*}
which implies that setting $\msynthdualop=[\vsynthframeel_1 \ \vsynthframeel_2\ \vsynthframeel_3]$ to be a left-inverse of $\manalop$ yields a valid dual frame.
Since $\manalop$ is a $3\times 2$ (``tall'') matrix, its left-inverse is not unique. In fact, $\manalop$ has infinitely many left-inverses (two of them were found above). Every left-inverse of $\manalop$ leads to a dual frame according to~\fref{eq:msynthdualop}.

Thanks to the redundancy of the frame $\{\vanalframeel_1,\vanalframeel_2,\vanalframeel_3\}$, we obtain design freedom: In order to synthesize the signal $\vsignal$ from its expansion coefficients $\coef_k=\inner{\vsignal}{\vanalframeel_k},\ k=1,2,3,$ in the frame $\{\vanalframeel_1,\vanalframeel_2,\vanalframeel_3\}$, we can choose between infinitely many dual frames $\{\vsynthframeel_1,\vsynthframeel_2,\vsynthframeel_3\}$. In practice the particular choice of the dual frame is usually dictated by the requirements of the specific problem at hand. We shall discuss this issue in detail in the context of sampling theory in \fref{sec:designfreedom}.
\end{example}

\section{Signal Expansions in Finite-Dimensional Hilbert Spaces}
\label{sec:signalexpfindim}
Motivated by the examples above, we now consider general signal expansions in finite-dimensional Hilbert spaces. As in the previous section, we first review the concept of an \onbac, we then consider arbitrary (nonorthogonal) bases, and, finally, we discuss  redundant vector sets --- frames. While the discussion in this section is confined to the finite-dimensional case, we develop the general (possibly infinite-dimensional) case in~\fref{sec:framesfd}.

\subsection{Orthonormal Bases}
\label{sec:orthobasis}
We start by reviewing the concept of an \onbac. 
\begin{dfn}
	\label{dfn:ONB}
The set of vectors $\{\vbasisel_k\}_{k=1}^\dimension,\, \vbasisel_k\in\complexset^\dimension,$ is called an \onbac for $\complexset^\dimension$ if
\begin{enumerate}
\item $\spn\{\vbasisel_k\}_{k=1}^\dimension=\{\coef_1\vbasisel_1+\coef_2\vbasisel_2+\ldots+\coef_\dimension\vbasisel_\dimension \given \coef_1,\coef_2,\ldots,\coef_\dimension \in \complexset\}=\complexset^\dimension$
\item 
\begin{equation*}
\inner{\vbasisel_k}{\vbasisel_j} =\begin{cases}1,\quad &k=j\\0,&k\neq j\end{cases}\quad\quad k,j=1,\ldots,\dimension.
\end{equation*}
\end{enumerate}
\end{dfn}

When $\{\vbasisel_k\}_{k=1}^\dimension$ is an ONB, thanks to the spanning property in~\fref{dfn:ONB},
every $\vsignal\in\complexset^\dimension$ can be decomposed as
\begin{equation}
	\label{eq:onbexp}
\vsignal=\sum_{k=1}^\dimension\coef_k\vbasisel_k.
\end{equation}
The expansion coefficients $\{\coef_k\}_{k=1}^\dimension$ in \eqref{eq:onbexp} can be found through the following calculation:
\begin{equation*}
\inner{\vsignal}{\vbasisel_j}=\inner{\sum_{k=1}^\dimension \coef_k \vbasisel_k}{\vbasisel_j}=\sum_{k=1}^\dimension\coef_k \inner{\vbasisel_k}{\vbasisel_j}=\coef_j.
\end{equation*}
In summary, we have the decomposition
\begin{equation*}
\vsignal=\sum_{k=1}^\dimension\inner{\vsignal}{\vbasisel_k}\vbasisel_k.
\end{equation*}
Just like in \fref{ex:orthnormalsystem}, in the previous section, we define the analysis matrix
\begin{equation*}
	\manalop\define\begin{bmatrix} \herm{\vbasisel_1} \\ \vdots \\ \herm{\vbasisel_\dimension} \end{bmatrix}.
\end{equation*}
If we organize the inner products $\{\inner{\vsignal}{\vbasisel_k}\}_{k=1}^\dimension$ into the vector \vcoef, we have
\begin{equation*}
\vcoef\define\begin{bmatrix} \inner{\vsignal}{\vbasisel_1} \\ \vdots \\ \inner{\vsignal}{\vbasisel_\dimension} \end{bmatrix} = \manalop\vsignal=\begin{bmatrix} \herm{\vbasisel_1} \\ \vdots \\ \herm{\vbasisel_\dimension} \end{bmatrix} \vsignal.
\end{equation*}
Thanks to the orthonormality of the vectors $\vbasisel_1, \vbasisel_2, \ldots, \vbasisel_\dimension$ the matrix $\manalop$
is unitary, i.e., $\herm\manalop=\manalop^{-1}$ and hence
\begin{equation*}
	\manalop\herm\manalop=\begin{bmatrix} \herm{\vbasisel_1}\\\vdots \\ \herm{\vbasisel_\dimension} \end{bmatrix}\begin{bmatrix} \vbasisel_1 & \ldots  &\vbasisel_\dimension \end{bmatrix}=\begin{bmatrix} \inner{\vbasisel_1}{\vbasisel_1} &\cdots & \inner{\vbasisel_\dimension}{\vbasisel_1}  \\ \vdots & \ddots & \vdots \\ \inner{\vbasisel_1}{\vbasisel_\dimension} & \cdots & \inner{\vbasisel_\dimension}{\vbasisel_\dimension} \end{bmatrix}= \identity_{\dimension}=\herm\manalop\manalop.
\end{equation*}
Thus,  if we multiply the vector $\vcoef$ by $\herm\manalop$, we synthesize $\vsignal$ according to
\begin{equation}
	\label{eq:matrixanalsynth}
	\herm\manalop\vcoef=\herm\manalop\manalop\vsignal=\sum_{k=1}^\dimension\inner{\vsignal}{\vbasisel_k}\vbasisel_k=\identity_{\dimension}\vsignal=\vsignal.
\end{equation}
We shall therefore call the matrix $\herm\manalop$ the synthesis matrix, corresponding to the analysis matrix $\manalop$. 
In the \onbac case considered here the synthesis matrix is simply the Hermitian adjoint of the analysis matrix.

\subsection{General Bases}
\label{sec:genbasis}
We next relax the orthonormality property, i.e., the second condition in \fref{dfn:ONB}, and consider general bases.
\label{sec:signalexpansionsfd}
\begin{dfn}
	\label{dfn:basis}
The set of vectors $\{\vbasisel_k\}_{k=1}^\dimension,\, \vbasisel_k\in\complexset^\dimension,$ is a basis for $\complexset^\dimension$ if
\begin{enumerate}
\item $\spn\{\vbasisel_k\}_{k=1}^\dimension=\{\coef_1\vbasisel_1+\coef_2\vbasisel_2+\ldots+\coef_\dimension\vbasisel_\dimension \given \coef_1,\coef_2,\ldots,\coef_\dimension \in \complexset\}=\complexset^\dimension$
\item $\{\vbasisel_k\}_{k=1}^\dimension$ is a linearly independent set, i.e., if $\sum_{k=1}^\dimension \coef_k\vbasisel_k=\veczero$ for some scalar coefficients $\{\coef_k\}_{k=1}^\dimension$, then necessarily $\coef_k=0$ for all $k=1,\ldots,\dimension.$
\end{enumerate}
\end{dfn}

Now consider a signal $\vsignal \in \complexset^\dimension$ and compute the expansion coefficients 
\begin{equation}
	\label{eq:basisanalcoef}
	\coef_k\define\inner{\vsignal}{\vbasisel_k}, \quad k=1,\ldots, \dimension.
\end{equation}
Again, it is convenient to introduce the analysis matrix 
\begin{equation*}
	\manalop\define\begin{bmatrix} \herm\vbasisel_1\\ \vdots \\ \herm\vbasisel_\dimension \end{bmatrix}
\end{equation*}
and to stack the coefficients $\{\coef_k\}_{k=1}^\dimension$ in the vector $\vcoef$. Then~\fref{eq:basisanalcoef} can be written as
\begin{equation*}
	\vcoef=\manalop\vsignal.
\end{equation*}
Next, let us ask how we can find a set of vectors $\{\vsynthbasisel_1,\ldots,\vsynthbasisel_\dimension\},\, \vsynthbasisel_k\in \complexset^\dimension,\, k=1,\ldots,\dimension,$ that is dual to the set $\{\vbasisel_1,\ldots,\vbasisel_\dimension\}$ in the sense that
\begin{equation}
	\label{eq:basisgendecomp1}
	\vsignal=\sum_{k=1}^\dimension \coef_k \vsynthbasisel_k=\sum_{k=1}^\dimension\inner{\vsignal}{\vbasisel_k}\vsynthbasisel_k 
\end{equation}
for all $\vsignal\in\complexset^\dimension$. If we introduce the synthesis matrix 
\begin{equation*}
	\msynthdualopc\define[\vsynthbasisel_1  \ \cdots \ \vsynthbasisel_\dimension],
\end{equation*}
we can rewrite \fref{eq:basisgendecomp1} in vector-matrix notation as follows
\begin{equation*}
	\vsignal=\msynthdualopc \vcoef=\msynthdualopc \manalop \vsignal.
\end{equation*}
This shows that finding vectors $\vsynthbasisel_1,\ldots,\vsynthbasisel_\dimension$ that satisfy \fref{eq:basisgendecomp1} is equivalent to finding  the inverse of the analysis matrix  $\manalop$ and setting $\msynthdualopc=\manalop^{-1}$. 
Thanks to the linear independence of the vectors $\{\vbasisel_k\}_{k=1}^\dimension$, the matrix $\manalop$ has full rank and is, therefore, invertible.  

Summarizing our findings, we conclude that in the case of a basis~$\{\vbasisel_k\}_{k=1}^\dimension$, the analysis matrix and the synthesis matrix are inverses of each other, i.e., $\msynthdualopc\manalop=\manalop\msynthdualopc=\identity_{\dimension}$. Recall that in the case of an \onbac the analysis matrix $\manalop$ is  \emph{unitary} and hence its inverse is simply given by $\herm\manalop$ [see \fref{eq:matrixanalsynth}], so that in this case $\tilde\manalop=\manalop$.

Next, note  that $\manalop\msynthdualopc=\identity_{\dimension}$ is equivalent to 
\begin{equation*}
\begin{bmatrix}\herm{\vbasisel_1} \\ \vdots \\ \herm{\vbasisel_\dimension}\end{bmatrix}\begin{bmatrix}\vsynthbasisel_1 &\ldots& \vsynthbasisel_\dimension \end{bmatrix}=\begin{bmatrix}
	\inner{\vsynthbasisel_1}{\vbasisel_1}&\cdots&\inner{\vsynthbasisel_\dimension}{\vbasisel_1}\\
	\vdots & \ddots & \vdots \\
	\inner{\vsynthbasisel_1}{\vbasisel_\dimension}&\cdots&\inner{\vsynthbasisel_\dimension}{\vbasisel_\dimension}
\end{bmatrix}=\identity_{\dimension} 
\end{equation*}
or equivalently
\begin{equation}
	\label{eq:biorthog}
	\inner{\vbasisel_k}{\vsynthbasisel_j}=\begin{cases}1,\quad &k=j \\0, &\text{else}\end{cases}, \quad\quad k,j=1,\ldots,\dimension.
\end{equation}
The sets $\{\vbasisel_k\}_{k=1}^\dimension$ and $\{\vsynthbasisel_k\}_{k=1}^\dimension$ are biorthonormal bases. 
\onbacp are biorthonormal to themselves in this terminology, as already noted in~\fref{ex:biorthogonal}. We emphasize that it is the fact that  $\manalop$ and $\msynthdualopc$ are square and full-rank that allows us to conclude that $\msynthdualopc\manalop=\identity_{\dimension}$ implies $\manalop\msynthdualopc=\identity_{\dimension}$ and hence to conclude that~\fref{eq:biorthog} holds.
 We shall see below that for redundant expansions $\manalop$ is a tall matrix and $\msynthdualopc\manalop\ne\manalop\msynthdualopc$ ($\msynthdualopc\manalop$ and $\manalop\msynthdualopc$ have different dimensions) so that dual frames will not be biorthonormal.

As $\manalop$ is a square matrix and of full rank, its inverse is  \emph{unique}, which means that for a  given analysis set $\{\vbasisel_k\}_{k=1}^\dimension$, the synthesis set $\{\vsynthbasisel_k\}_{k=1}^\dimension$ is \emph{unique}.  Alternatively, for a given synthesis set $\{\vsynthbasisel_k\}_{k=1}^\dimension$, there is a \emph{unique} analysis set $\{\vbasisel_k\}_{k=1}^\dimension$. This uniqueness property is not always desirable. For example, one may want to impose certain structural properties on the synthesis set $\{\vsynthbasisel_k\}_{k=1}^\dimension$ in which case having freedom in choosing the synthesis set as in~\fref{ex:biorthogonal} is helpful.

An important property of \onbacp is that they are norm-preserving: The norm of the coefficient vector \vcoef is equal to the norm of the signal \vsignal. This can  easily be seen by noting that  
\begin{equation}
	\label{eq:ONBnormpres}
\vecnorm{\vcoef}^2 = \herm\vcoef\vcoef=\herm\vsignal\herm\manalop\manalop\vsignal= \herm\vsignal\identity_{\dimension}\vsignal=\vecnorm{\vsignal}^2,
\end{equation}
where we used \fref{eq:matrixanalsynth}.
Biorthonormal bases are \emph{not} norm-preserving, in general. Rather, 
the equality in \fref{eq:ONBnormpres} is relaxed to a double-inequality, by application of the Rayleigh-Ritz theorem~\cite[Sec. 9.7.2.2]{lutkepohl96} according to 
\begin{equation}
	\label{eq:RR}
 \lambda_{\text{min}}\lefto(\herm{\manalop}\manalop\right)\vecnorm{\vsignal}^2 \leq \vecnorm{\vcoef}^2=\herm\vsignal \herm{\manalop}\manalop\vsignal\leq \lambda_{\text{max}}\lefto(\herm{\manalop}\manalop\right) \vecnorm{\vsignal}^2.
\end{equation}

\subsection{Redundant Signal Expansions}
\label{sec:redsigexp}
The signal expansions we considered so far are non-redundant in the sense that the number of expansion coefficients equals the dimension of the Hilbert space. Such signal expansions have a number of disadvantages. 
First, corruption or loss of expansion coefficients can result in significant reconstruction errors. Second, the reconstruction process is very rigid: As we have seen in \fref{sec:genbasis}, for each set of analysis vectors,  there is a \emph{unique}  set of synthesis vectors. 
In practical applications it is often desirable to impose additional constraints on the reconstruction functions, such as smoothness properties or structural properties that allow for computationally efficient reconstruction.

Redundant expansions  allow to overcome many of these problems as they offer design freedom and robustness to corruption or loss of expansion coefficients.
 We already saw in~\fref{ex:overcomplete} that in the case of redundant expansions, for a given  set of analysis vectors
the set of synthesis vectors that allows perfect recovery of a signal from its expansion coefficients is not unique; in fact, there are infinitely many sets of synthesis vectors, in general. This results in design freedom and provides robustness. Suppose that the expansion coefficient $\coef_3=\inner{\vsignal}{\vanalframeel_3}$ in~\fref{ex:overcomplete} is corrupted or even completely lost. We can still reconstruct $\vsignal$ \emph{exactly} from \fref{eq:vsignalbasisdec}.

Now, let us turn to developing the general theory of redundant signal expansions in finite-dimensional Hilbert spaces. Consider a set of $\framesize$ vectors $\{\vanalframeel_1,\ldots,\vanalframeel_\framesize\},\, \vanalframeel_k\in \complexset^\dimension,\, k=1,\ldots,\framesize,$ with $\framesize\ge\dimension$. Clearly, when $\framesize$ is \emph{strictly} greater than $\dimension$, the vectors $\vanalframeel_1,\ldots,\vanalframeel_\framesize$ must be linearly dependent. 
Next, consider a signal $\vsignal\in\complexset^\dimension$ and compute the expansion coefficients 
\begin{equation}
	\label{eq:redanalcoef}
	\coef_k=\inner{\vsignal}{\vanalframeel_k}, \quad k=1,\ldots, \framesize.
\end{equation}
Just as before, it is convenient to introduce the analysis matrix 
\begin{equation}
	\label{eq:manalopdfn}
	\manalop\define\begin{bmatrix} \herm\vanalframeel_1\\ \vdots \\ \herm\vanalframeel_\framesize \end{bmatrix}
\end{equation}
and to stack the coefficients $\{\coef_k\}_{k=1}^\framesize$ in the vector $\vcoef$. Then~\fref{eq:redanalcoef} can be written as
\begin{equation}
	\label{eq:redundantcoeff}
	\vcoef=\manalop\vsignal.
\end{equation}
Note that $\vcoef\in\complexset^\framesize$ and $\vsignal\in\complexset^\dimension$. Differently from \onbacp and biorthonormal bases considered in Sections~\ref{sec:orthobasis} and~\ref{sec:genbasis}, respectively, in the case of redundant expansions, the signal \vsignal and the expansion coefficient vector~\vcoef will, in general, belong to different Hilbert spaces. 
 
The question now is: How can we find a set of vectors $\{\vsynthframeel_1,\ldots,\vsynthframeel_\framesize\},\ \vsynthframeel_k\in \complexset^\dimension,\, k=1,\ldots,\framesize,$ such that
\begin{equation}
	\label{eq:framedecomp1}
	\vsignal=\sum_{k=1}^\framesize \coef_k \vsynthframeel_k=\sum_{k=1}^\framesize\inner{\vsignal}{\vanalframeel_k}\vsynthframeel_k
\end{equation}
for all $\vsignal\in\complexset^\dimension$?
If we introduce the synthesis matrix 
\begin{equation*}
	\msynthdualopc\define[\vsynthframeel_1  \ \cdots \ \vsynthframeel_\framesize],
\end{equation*}
we can rewrite \fref{eq:framedecomp1} in vector-matrix notation as follows
\begin{equation}
	\label{eq:framedecomp1op}
	\vsignal=\msynthdualopc \vcoef=\msynthdualopc \manalop \vsignal.
\end{equation}
Finding vectors $\vsynthframeel_1,\ldots,\vsynthframeel_\framesize$ that satisfy \fref{eq:framedecomp1} for all $\vsignal\in\complexset^\dimension$ is therefore equivalent to finding a left-inverse $\msynthdualopc$ of $\manalop$, i.e.,
\begin{equation*}
	\msynthdualopc\manalop=\identity_{\dimension}.
\end{equation*}
First note that~$\manalop$ is left-invertible if and only if $\complexset^\dimension =\spn\{\vanalframeel_k\}_{k=1}^\framesize$, i.e., if and only if the set of vectors~$\{\vanalframeel_k\}_{k=1}^\framesize$ spans~$\complexset^\dimension$.
Next observe that when $\framesize>\dimension$, the $\framesize\times\dimension$ matrix $\manalop$ is a ``tall'' matrix, and therefore 
its left-inverse is, in general, not  unique. In fact, there are infinitely many left-inverses. The following theorem~\cite[Ch. 2, Th.~1]{ben-israel02} provides a convenient parametrization of all these left-inverses. 
\begin{thm}
\label{thm:leftinv}
	Let $\mat\in\complexset^{\framesize\times \dimension}$, $\framesize\ge\dimension$. Assume that $\rank(\mat)=\dimension$. 
	Then $\pinv \mat\define (\herm \mat \mat)^{-1}\herm \mat$ is a left-inverse of $\mat$, i.e., $\pinv \mat\mat=\identity_{\dimension}$.
	Moreover, the general solution $\matL\in\complexset^{\dimension\times\framesize}$ of the equation $\matL\mat=\identity_{\dimension}$ is given by
	\begin{equation}
		\label{eq:generalli}
		\matL=\pinv \mat+\matM(\identity_{\framesize}-\mat\pinv \mat),
	\end{equation}
	where $\matM\in\complexset^{\dimension\times\framesize}$ is an arbitrary matrix.
\end{thm}
\begin{proof}
	Since $\rank(\mat)=\dimension$, the matrix $\herm \mat \mat$ is invertible and hence $\pinv \mat$ is well defined.  
	Now, let us verify that $\pinv \mat$ is, indeed, a left-inverse of $\mat$:
	\begin{equation}
		\label{eq:pseudoinv}
		\pinv \mat\mat=(\herm \mat \mat)^{-1}\herm \mat \mat=\identity_{\dimension}.
	\end{equation}
	The matrix $\pinv \mat$ is called the Moore-Penrose inverse of \mat.
	
	Next, we show that every matrix $\matL$ of the form~\fref{eq:generalli} is a valid left-inverse of $\mat$:
	\begin{align*}
		\matL\mat&=\left(\pinv \mat+\matM(\identity_{\framesize}-\mat\pinv \mat)\right)\mat\nonumber\\
		&=\underbrace{\pinv \mat\mat}_{\identity_{\dimension}}+\matM\mat-\matM\mat\underbrace{\pinv \mat\mat}_{\identity_{\dimension}}\nonumber\\
		&=\identity_{\dimension}+\matM\mat-\matM\mat=\identity_{\dimension},
	\end{align*} 
	where we used \fref{eq:pseudoinv} twice. 
	
	Finally, assume that $\matL$ is a valid left-inverse of \mat, i.e., $\matL$ is a solution of the equation $\matL\mat=\identity_{\dimension}$. We show that $\matL$ can be written in the form~\fref{eq:generalli}.  Multiplying the equation $\matL\mat=\identity_{\dimension}$ by $\pinv \mat$ from the right, we have
	\begin{equation*}
		\matL\mat\pinv \mat=\pinv \mat.
	\end{equation*}
	Adding \matL to both sides of this equation and rearranging terms yields
	\begin{equation*}
		\matL=\pinv \mat+\matL-\matL\mat\pinv \mat=\pinv \mat+\matL(\identity_{\framesize}-\mat\pinv \mat),
	\end{equation*}
	which shows that \matL can be written in the form \fref{eq:generalli} (with~$\matM=\matL$), as required.
\end{proof}
 We conclude that for each redundant set of vectors $\{\vanalframeel_1,\ldots,\vanalframeel_\framesize\}$ that spans $\complexset^\dimension$, there are infinitely many dual sets $\{\vsynthframeel_1,\ldots,\vsynthframeel_\framesize\}$ such that the decomposition \fref{eq:framedecomp1} holds for all $\vsignal\in\complexset^\dimension$.
These dual sets are obtained by identifying $\{\vsynthframeel_1,\ldots,\vsynthframeel_\framesize\}$ with the columns of $\matL$ according to 
\begin{equation*}
		[\vsynthframeel_1  \ \cdots \ \vsynthframeel_\framesize]=\matL,
\end{equation*}
where $\matL$ can be written as  follows
\begin{equation*}
	\matL=\pinv \manalop+\matM(\identity_\framesize-\manalop\pinv \manalop)
\end{equation*}
and $\matM\in\complexset^{\dimension\times\framesize}$ is an arbitrary matrix.

The dual set $\{\vsynthframeel_1,\ldots,\vsynthframeel_\framesize\}$ corresponding to the Moore-Penrose inverse $\matL=\pinv\manalop$ of the matrix $\manalop$, i.e.,
\begin{equation*}
	[\vsynthframeel_1  \ \cdots \ \vsynthframeel_\framesize]=\pinv\manalop=(\herm\manalop\manalop)^{-1}\herm\manalop
\end{equation*}
is called the \emph{canonical dual} of $\{\vanalframeel_1,\ldots,\vanalframeel_\framesize\}$. Using \fref{eq:manalopdfn}, we see that in this case 
\begin{equation}
	\label{eq:finitedimcanondual}
	\vsynthframeel_k=(\herm\manalop\manalop)^{-1}\vanalframeel_k,\quad k=1,\ldots,\framesize.
\end{equation}

Note that \emph{unlike} in the case of a basis, the equation $\msynthdualopc\manalop=\identity_\dimension$ \emph{does not} imply that the sets~$\{\vsynthframeel_k\}_{k=1}^\framesize$  and~$\{\vanalframeel_k\}_{k=1}^\framesize$ are biorthonormal. This is because the matrix $\manalop$ is \emph{not} a square matrix, and thus, $\msynthdualopc\manalop\ne\manalop\msynthdualopc$ ($\msynthdualopc\manalop$ and $\manalop\msynthdualopc$ have different dimensions).

Similar to biorthonormal bases, redundant sets of vectors are, in general, not norm-preserving. Indeed, from \fref{eq:redundantcoeff} we see that
\begin{equation*}
	\vecnorm{\vcoef}^2=\herm\vsignal \herm{\manalop}\manalop\vsignal
\end{equation*}
and thus, by the Rayleigh-Ritz theorem~\cite[Sec. 9.7.2.2]{lutkepohl96}, we have
\begin{equation}
	\label{eq:framematspecbounds}
 \lambda_{\text{min}}\lefto(\herm{\manalop}\manalop\right)\vecnorm{\vsignal}^2 \leq \vecnorm{\vcoef}^2\leq \lambda_{\text{max}}\lefto(\herm{\manalop}\manalop\right) \vecnorm{\vsignal}^2
\end{equation}
as in the case of biorthonormal bases.
	
We already saw some of the basic issues that a theory of orthonormal, biorthonormal, and redundant signal expansions should address. It should  account for the signals and the expansion coefficients belonging, potentially, to different Hilbert spaces; it should account for the fact that for a given analysis set, the synthesis set is not unique in the redundant case, it should prescribe how synthesis vectors can be obtained from the analysis vectors. Finally, it should apply not only to finite-dimensional Hilbert spaces, as considered so far, but also to infinite-dimensional Hilbert spaces. We now proceed to develop this general theory, known as the theory of frames.

\section{Frames for General Hilbert Spaces}
\label{sec:framesfd}
Let  $\{\analframeel_k\}_{k\in\frameset}$ ($\frameset$ is a countable set) be a set of elements taken from the Hilbert space \hilspace. Note that this set need not be orthogonal. 

In developing a general theory of signal expansions in Hilbert spaces, as outlined at the end of the previous section, we start by noting that the central quantity  in~\fref{sec:signalexpfindim} was the analysis matrix  $\manalop$ associated to the (possibly nonorthogonal or redundant) set of vectors $\{\vanalframeel_1,\ldots,\vanalframeel_\framesize\}$. Now matrices are nothing but linear operators in finite-dimensional Hilbert spaces. In formulating frame theory for general (possibly infinite-dimensional) Hilbert spaces, it is therefore sensible to define the analysis operator \analop that assigns to each signal $\fun\in\hilspace$ the sequence of inner products $\analop\fun=\{\inner{\fun}{\analframeel_k}\}_{k\in\frameset}$. Throughout this section, we assume that $\{\analframeel_k\}_{k\in\frameset}$ is a Bessel sequence, i.e., $\sum_{k\in\frameset} \abs{\inner{\fun}{\analframeel_k}}^2<\infty$ for all $\fun\in\hilspace$.
\begin{dfn}
	The linear operator \analop is defined as the operator that maps the Hilbert space \hilspace into the space \hilseqspace of square-summable complex sequences\footnote{The fact that the range space of \analop is contained in  \hilseqspace is a consequence of $\{\analframeel_k\}_{k\in\frameset}$ being a Bessel sequence.}, 
	$\analop:\hilspace\to\hilseqspace$, 
	by assigning to each signal $\fun\in\hilspace$ the sequence of inner products $\inner{\fun}{\analframeel_k}$ according to
	\begin{equation*}
		\analop:\fun\to\{\inner{\fun}{\analframeel_k}\}_{k\in\frameset}.
	\end{equation*}
\end{dfn}

Note that $\vecnorm{\analop\fun}^2=\sum_{k\in\frameset}\abs{\inner{\fun}{\analframeel_k}}^2$, i.e., the energy $\vecnorm{\analop\fun}^2$ of $\analop\fun$ can be expressed as 
\begin{equation}
	\label{eq:energyT}
	\vecnorm{\analop\fun}^2=\sum_{k\in\frameset}\abs{\inner{\fun}{\analframeel_k}}^2.
\end{equation}
We shall next formulate properties that the set  $\{\analframeel_k\}_{k\in\frameset}$ and hence the operator \analop should satisfy if we have signal expansions in mind:
\begin{enumerate}
	\label{req:1}
\item The signal \fun can be perfectly reconstructed from the coefficients $\{\inner{\fun}{\analframeel_k}\}_{k\in\frameset}$. This means that we want $\inner{\fun}{\analframeel_k}=\inner{\altfun}{\analframeel_k}$, for all $k\in\frameset,$ (i.e., $\analop\fun=\analop\altfun$) to imply that $\fun=\altfun$, for all $\fun,\altfun\in\hilspace$. In other words, the operator $\analop$ has to be left-invertible, which means that $\analop$ is invertible on its range space $\rng(\analop)=\{\altfun\in \hilseqspace:\altfun = \analop\fun,\ \fun \in \hilspace\}$.

This requirement will clearly be satisfied if we demand that there exist a constant $\frameA>0$ such that for all $\fun, \altfun\in \hilspace$ we have
\begin{equation*}
	\frameA\vecnorm{\fun-\altfun}^2\le  \vecnorm{\analop\fun-\analop\altfun}^2.
\end{equation*}
Setting $z=\fun-\altfun$ and using the linearity of $\analop$, we see that this condition is equivalent to 
\begin{equation}
	\label{eq:lowerfr}
	\frameA\vecnorm{z}^2\le  \vecnorm{\analop z}^2
\end{equation}
for all $z\in\hilspace$ with $\frameA>0$. 
\item The energy in the sequence of expansion coefficients $\analop\fun=\{\inner{\fun}{\analframeel_k}\}_{k\in\frameset}$  should be related to the energy in the signal~\fun. For example, we saw  in \eqref{eq:ONBnormpres} that if $\{\vbasisel_k\}_{k=1}^\dimension$ is an \onbac for $\complexset^\dimension$, then 
\begin{equation}
	\label{eq:parseval}
	\vecnorm{\manalop\vsignal}^2=\sum_{k=1}^\dimension \abs{\inner{\vsignal}{\vbasisel_k}}^2=	\vecnorm{\vsignal}^2,\ \text{for all } \vsignal\in\complexset^\dimension.
\end{equation}
This property is a consequence of the unitarity of $\analop=\manalop$ and it is clear that it will not hold for general sets $\{\analframeel_k\}_{k\in\frameset}$ (see the discussion around~\fref{eq:RR} and~\fref{eq:framematspecbounds}). Instead, we will relax~\fref{eq:parseval} to demand that for all $\fun\in \hilspace$  there exist a finite constant $\frameB$ such that\footnote{Note that if~\fref{eq:upperfr} is satisfied with $\frameB<\infty$, then $\{\analframeel_k\}_{k\in\frameset}$ is a Bessel sequence. }
\begin{equation}
	\label{eq:upperfr}
	\vecnorm{\analop\fun}^2=\sum_{k\in\frameset} \abs{\inner{\fun}{\analframeel_k}}^2\le	\frameB\vecnorm{\fun}^2.
\end{equation}
\end{enumerate}
Together with~\fref{eq:lowerfr} this ``sandwiches'' the quantity $\vecnorm{\analop\fun}^2$ according to
\begin{equation*}
	\frameA\vecnorm{\fun}^2\le\vecnorm{\analop\fun}^2\le \frameB\vecnorm{\fun}^2.
\end{equation*}
We are now ready to  formally define a frame for the Hilbert space~\hilspace.
\begin{dfn}
	\label{dfn:frame}
A set of elements $\setframem$  is called a frame for the Hilbert space \hilspace if
\begin{equation}
	\label{eq:framedef}
\frameA\vecnorm{\fun}^2\leq\sum_{k\in\frameset} \abs{\inner{\fun}{\analframeel_k}}^2 \leq \frameB\vecnorm{\fun}^2, \quad \text{for all}  \quad\fun \in \hilspace,
\end{equation}
with $\frameA,\frameB\in \reals$ and $0<\frameA\leq \frameB<\infty$. Valid constants $\frameA$ and $\frameB$ are called frame bounds.
The largest valid constant $\frameA$ and the smallest valid constant $\frameB$ are called \emph{the} (tightest possible) \emph{frame bounds}.
\end{dfn}

Let us next consider some simple examples of frames.
\begin{example}[\citepar{christensen96}]
	\label{ex:tightframeido}
Let $\{\basisel_k\}_{k=1}^\infty$ be an \onbac for an infinite-dimensional Hilbert space \hilspace.
By repeating each element in $\{\basisel_k\}_{k=1}^\infty$ once, we obtain the redundant set
	\begin{equation*}
	\{\analframeel_k\}_{k=1}^\infty=\{\basisel_1,\basisel_1,\basisel_2,\basisel_2,\ldots\}.
	\end{equation*}
	To see that this set is a frame for~\hilspace, we note that because $\{\basisel_k\}_{k=1}^\infty$ is an \onbac, for all $\fun\in\hilspace$, we have
	\begin{equation*}
		\sum_{k=1}^\infty \abs{\inner{\fun}{\basisel_k}}^2=\vecnorm{\fun}^2
	\end{equation*}
	and therefore 
	\begin{equation*}
		\sum_{k=1}^\infty \abs{\inner{\fun}{\analframeel_k}}^2=\sum_{k=1}^\infty \abs{\inner{\fun}{\basisel_k}}^2+\sum_{k=1}^\infty \abs{\inner{\fun}{\basisel_k}}^2=2\vecnorm{\fun}^2.
	\end{equation*}
	This verifies the frame condition \fref{eq:framedef} and shows that the frame bounds are given by 
	$\frameA=\frameB=2$.
\end{example}
\begin{example}[\citepar{christensen96}]
	\label{ex:tightframeidt}
	Starting from the \onbac  $\{\basisel_k\}_{k=1}^\infty,$ we can construct another redundant set as follows
	\begin{equation*}
		\{\analframeel_k\}_{k=1}^\infty=\left\{\basisel_1,\frac{1}{\sqrt 2}\basisel_2,\frac{1}{\sqrt 2}\basisel_2,\frac{1}{\sqrt 3}\basisel_3,\frac{1}{\sqrt 3}\basisel_3,\frac{1}{\sqrt 3}\basisel_3,\ldots\right\}.
\end{equation*}
To see that the set $\{\analframeel_k\}_{k=1}^\infty$ is a frame for~\hilspace, take an arbitrary $\fun\in\hilspace$ and  note that 
\begin{equation*}
	\sum_{k=1}^\infty \abs{\inner{\fun}{\analframeel_k}}^2=\sum_{k=1}^\infty k\abs{\inner{\fun}{\frac{1}{\sqrt k}\basisel_k}}^2 = \sum_{k=1}^\infty k \frac{1}{k} \abs{\inner{\fun}{\basisel_k}}^2=\sum_{k=1}^\infty \abs{\inner{\fun}{\basisel_k}}^2=\vecnorm{\fun}^2.
	\end{equation*}
	We conclude that $\{\analframeel_k\}_{k=1}^\infty$ is a frame with the frame bounds $\frameA=\frameB=1$.
\end{example}
From \eqref{eq:energyT} it follows that an equivalent formulation of  \eqref{eq:framedef} is
\begin{equation*}
\frameA\vecnorm{\fun}^2\leq \vecnorm{\analop\fun}^2\leq \frameB\vecnorm{\fun}^2, \quad \text{for all}  \quad\fun \in \hilspace.	
\end{equation*}

This means that the energy in the coefficient sequence $\analop\fun$ is bounded above and below by bounds that are proportional to the signal energy. The existence of a lower frame bound $\frameA>0$ guarantees that the linear operator \analop is left-invertible,~i.e., our first requirement above is satisfied. Besides that it also guarantees completeness of the set $\setframemm$ for $\hilspace$, as we shall see next. To this end, we first recall the following definition: 
\begin{dfn}
\label{dfn:completeness}
	A set of elements $\setframem$ is complete for the Hilbert space $\hilspace$ if~$\inner{\fun}{\analframeel_k}=0$ for all~$k\in\frameset$ and with~$\fun \in \hilspace$ implies~$\fun=0$, i.e., the only element in~$\hilspace$ that is orthogonal to
	all $\analframeel_k$, is $\fun=0$. 
\end{dfn}

To see that the frame~$\setframemm$ is complete for \hilspace, take an arbitrary signal~$\fun\in\hilspace$ and assume that~$\inner{\fun}{\analframeel_k}=0$ for all~$k\in\frameset$. Due to the existence of a lower frame bound~$\frameA>0$ we have
\begin{equation*}
	\frameA\vecnorm{\fun}^2\le\sum_{k\in\frameset} \abs{\inner{\fun}{\analframeel_k}}^2=0,
\end{equation*}
which implies $\vecnorm{\fun}^2=0$ and hence $\fun=0.$

 Finally, note that the existence of an upper frame bound $\frameB<\infty$ guarantees  that $\analop$ is a bounded linear operator\footnote{Let \hilspace and $\hilspace'$ be  Hilbert spaces and $\opA: \hilspace\to\hilspace'$ a linear operator. The operator $\opA$ is said to be \emph{bounded}  if there exists a finite number $c$ such that for all $\fun\in\hilspace$, $\vecnorm{\opA\fun}\le c\vecnorm{\fun}$.} (see~\cite[Def. 2.7-1]{kreyszig89}), and, therefore (see~\cite[Th. 2.7-9]{kreyszig89}), continuous\footnote{Let \hilspace and $\hilspace'$ be Hilbert spaces and $\opA: \hilspace\to\hilspace'$ a linear operator. The operator $\opA$ is said to be \emph{continuous} at a point $\fun_0\in\hilspace$ if for every $\epsilon>0$ there is a $\delta>0$ such that for all $\fun\in\hilspace$ satisfying $\vecnorm{\fun-\fun_0}<\delta$ it follows that $\vecnorm{\opA\fun-\opA\fun_0}< \epsilon$. The operator \opA is said to be continuous on \hilspace, if it is continuous at every point $\fun_0\in\hilspace$.} (see~\cite[Sec. 2.7]{kreyszig89}).

Recall that we would like to find a general method to reconstruct a signal $\fun\in\hilspace$ from its expansion coefficients $\{\inner{\fun}{\analframeel_k}\}_{k\in\frameset}$. 
In~\fref{sec:redsigexp} we saw that in the finite-dimensional case, this can be accomplished according to
\begin{equation*}
	\vsignal=\sum_{k=1}^\framesize\inner{\vsignal}{\vanalframeel_k}\vsynthframeel_k.
\end{equation*}
Here $\{\vsynthframeel_1,\ldots,\vsynthframeel_\framesize\}$ can be chosen to be the canonical dual set to the set $\{\vanalframeel_1,\ldots,\vanalframeel_\framesize\}$
and can be computed as follows: $\vsynthframeel_k=(\herm\manalop\manalop)^{-1}\vanalframeel_k,\, k=1,\ldots,\framesize$.
We already know that $\analop$ is the generalization of $\manalop$ to the infinite-dimensional setting.
Which operator will then correspond to~$\herm\manalop$?  
To answer this question we start with a definition.

\begin{dfn}
\label{dfn:analcrossop}
	The linear operator $\analcrossop$ is defined as 
	\begin{align*}
		&\analcrossop: \hilseqspace\to \hilspace\\
		& \analcrossop: \coeflt\to\sum_{k\in\frameset}\coef_k \analframeel_k.
	\end{align*}
\end{dfn}
Next, we recall the definition of the adjoint of an operator.
\begin{dfn}
	Let $\opA:\hilspace\to\hilspace'$ be a bounded linear operator between the Hilbert spaces $\hilspace$ and $\hilspace'$. The unique bounded linear operator $\ad\opA:\hilspace'\to\hilspace$ that satisfies 
\begin{equation}
	\label{eq:defadj}
	\inner{\opA\fun}{\altfun}=\inner{\fun}{\ad\opA\altfun}
\end{equation}
for all $\fun\in\hilspace $ and all $\altfun\in\hilspace'$ is called the adjoint of \opA.
\end{dfn}
Note that the concept of the adjoint of an operator directly generalizes that of the Hermitian transpose of a matrix: if $\mat\in\complexset^{N\times M}$, $\vecx\in\complexset^M$, $\vecy\in\complexset^N$, then 
\begin{equation*}
	\inner{\mat\vecx}{\vecy}=\herm\vecy\mat\vecx=\herm{(\herm\mat\vecy)}\vecx=\inner{\vecx}{\herm\mat\vecy},
\end{equation*}
which, comparing to~\fref{eq:defadj}, shows that $\herm\mat$ corresponds to $\ad\opA$.

We shall next show that the operator $\analcrossop$ defined above is nothing but the adjoint~$\ad{\analop}$ of the operator~$\analop$. 
To see this consider an arbitrary sequence~$\coeflt\in \hilseqspace$ and an arbitrary signal~$\fun\in\hilspace$. We have to prove that 
\begin{equation*}
	\inner{\analop\fun}{\coeflt}=\inner{\fun}{\analcrossop\coeflt}.
\end{equation*}
This can be established by noting that 
\begin{align*}
	\inner{\analop\fun}{\coeflt}&=\sum_{k\in\frameset} \inner{\fun}{\analframeel_k}\conj{\coef_k}\\
	\inner{\fun}{\analcrossop\coeflt}&=\inner{\fun}{\sum_{k\in\frameset}\coef_k \analframeel_k}=\sum_{k\in\frameset} \conj{\coef_k} \inner{\fun}{\analframeel_k}.
\end{align*}
We therefore showed that the adjoint operator of $\analop$ is $\analcrossop$, i.e.,
\begin{equation*}
	\analcrossop=\ad\analop.
\end{equation*}
In what follows, we shall always write $\ad\analop$ instead of $\analcrossop$. As pointed out above the concept of the adjoint of an operator generalizes the concept of the Hermitian transpose of a matrix to the infinite-dimensional case. Thus, $\ad\analop$ is the generalization of $\herm\manalop$ to the infinite-dimensional setting.

\subsection{The Frame Operator}
Let us return to the discussion we had immediately before~\fref{dfn:analcrossop}. We saw that in the finite-dimensional case, the canonical dual set $\{\vsynthframeel_1,\ldots,\vsynthframeel_\framesize\}$ to the set $\{\vanalframeel_1,\ldots,\vanalframeel_\framesize\}$
 can be computed as follows: $\vsynthframeel_k=(\herm\manalop\manalop)^{-1}\vanalframeel_k,\, k=1,\ldots,\framesize$. We know that $\analop$ is the generalization of $\manalop$ to the infinite-dimensional case and we have just seen that $\ad\analop$ is the generalization of $\herm\manalop$. It is now obvious that the operator $\ad\analop\analop$ must correspond to $\herm\manalop\manalop$. The operator $\ad\analop\analop$ is of central importance in frame theory.
\begin{dfn}
	\label{dfn:frameop}
	Let~$\{\analframeel_k\}_{k\in\frameset}$ be a frame for the Hilbert space $\hilspace$. The operator $\frameop:\hilspace\to\hilspace$ defined as
	\begin{align}
		\label{eq:frameop}
		&\frameop=\ad\analop\analop,\\
		&\frameop\fun=\sum_{k\in \frameset } \inner{\fun}{\analframeel_k}\analframeel_k\nonumber
	\end{align}
	is called the frame operator.
\end{dfn}
We note that 
\begin{equation}
	\label{eq:frameopanalopconnection}
	\sum_{k\in\frameset} \abs{\inner{\fun}{\analframeel_k}}^2=\vecnorm{\analop\fun}^2=\inner{\analop\fun}{\analop\fun}=\inner{\ad\analop\analop\fun}{\fun}=\inner{\frameop\fun}{\fun}.
\end{equation}
We are now able to formulate the frame condition in terms of the frame operator $\frameop$ by simply noting that~\fref{eq:framedef} can be written as
\begin{equation}
	\label{eq:framecondS}
	\frameA\vecnorm{\fun}^2\leq\inner{\frameop\fun}{\fun} \leq \frameB\vecnorm{\fun}^2.
\end{equation}

We  shall next discuss the properties of $\frameop$. 
\begin{thm}
The frame operator $\frameop$ satisfies the properties:
\begin{enumerate}
	\item \frameop is linear and bounded;
	\item \label{prop:selfadj}\frameop is self-adjoint, i.e., $\ad{\frameop}=\frameop$;
	\item \label{prop:posdef} \frameop is positive definite, i.e., $\inner{\frameop\fun}{\fun}>0$ for all $\fun\in\hilspace$;
	\item \label{prop:sqrt}\frameop has a unique self-adjoint positive definite square root (denoted as $\frameop^{1/2}$).
\end{enumerate}	
\end{thm}
\begin{proof}
	\begin{enumerate}
	\item
	Linearity and boundedness of \frameop follow from the fact that \frameop is obtained by cascading a bounded linear operator and its adjoint (see~\fref{eq:frameop}). 
	\item
	To see that \frameop is self-adjoint simply note that
	\begin{equation*}
		\ad\frameop = \ad{(\ad\analop\analop)} = \ad\analop\analop = \frameop.
	\end{equation*}
	\item
	To see that \frameop is positive definite note that, with \fref{eq:framecondS}
	\begin{equation*}
		\inner{\frameop\fun}{\fun}\ge \frameA\vecnorm{\fun}^2>0
	\end{equation*}
	for all $\fun\in\hilspace,\ \fun\ne 0$.
	\item
	Recall the following basic fact from functional analysis~\cite[Th. 9.4-2]{kreyszig89}.
	\begin{lem} 
	\label{lem:possqrt}
	Every self-adjoint positive definite bounded  operator $\opA:\hilspace\to\hilspace$ has a unique self-adjoint positive definite square root, i.e., there exists a unique self-adjoint positive-definite operator \opB such that $\opA=\opB\opB$.  The operator \opB commutes with the operator \opA, i.e., $\opB\opA=\opA\opB$.  
	\end{lem}
	Property \ref{prop:sqrt} now follows directly form Property \ref{prop:selfadj}, Property \ref{prop:posdef}, and~\fref{lem:possqrt}.
	\end{enumerate} 
\end{proof}	
We next show that the tightest possible frame bounds $\frameA$ and $\frameB$ are given by the smallest and the largest spectral value~\cite[Def. 7.2-1]{kreyszig89} of the frame operator \frameop, respectively.
\begin{thm}
	Let $\frameA$ and $\frameB$ be the tightest possible frame bounds for a frame with frame operator $\frameop$. Then
	\begin{equation}
		\label{eq:framaABeval}
		\frameA=\eval_{\min}\quad \text{and}\quad \frameB= \eval_{\max},
	\end{equation}
	where  $\eval_{\min}$ and $\eval_{\max}$ denote the smallest and the largest spectral value of $\frameop$, respectively. 
\end{thm}
\begin{proof}
	By standard results on the spectrum of self-adjoint operators~\cite[Th. 9.2-1, Th. 9.2-3, Th. 9.2-4]{kreyszig89}, we have
	\begin{equation}
		\label{eq:frameopspecbounds}
		\eval_{\min}=\inf_{\fun\in\hilspace}\frac{\inner{\frameop\fun}{\fun}}{\vecnorm{\fun}^2} \ \ \ \text{and} \ \ \   \eval_{\max}=\sup_{\fun\in\hilspace}\frac{\inner{\frameop\fun}{\fun}}{\vecnorm{\fun}^2}.
	\end{equation}
	This means that $\eval_{\min}$ and $\eval_{\max}$ are, respectively, the largest and the smallest constants such that
	\begin{equation}
		\label{eq:frameopspecbounds2}
		\eval_{\min}\vecnorm{\fun}^2\leq\inner{\frameop\fun}{\fun} \leq \eval_{\max}\vecnorm{\fun}^2 
	\end{equation}
	is satisfied for every $\fun\in\hilspace$.
	According to \eqref{eq:framecondS} this implies that $\eval_{\min}$ and $\eval_{\max}$ are the tightest possible frame bounds.
\end{proof}	

It is instructive to compare \fref{eq:frameopspecbounds2} to \fref{eq:framematspecbounds}. Remember that  $\frameop=\ad\analop\analop$ corresponds to the matrix $\herm\manalop\manalop$ in the finite-dimensional case considered in \fref{sec:redsigexp}. Thus, $\vecnorm{\vcoef}^2=\herm\vsignal \herm{\manalop}\manalop\vsignal=\inner{\mframeop\vsignal}{\vsignal}$, which upon insertion into \fref{eq:framematspecbounds}, shows that \fref{eq:frameopspecbounds2} is simply a generalization of \fref{eq:framematspecbounds} to the infinite-dimensional case.

\subsection{The Canonical Dual Frame}
\label{sec:dualframe}
 Recall that in the finite-dimensional case considered in \fref{sec:redsigexp}, the canonical dual frame $\{\vsynthframeel_k\}_{k=1}^\framesize$ of the frame $\{\vanalframeel_k\}_{k=1}^\framesize$ can be used to reconstruct the signal $\vsignal$ from the expansion coefficients $\{\inner{\vsignal}{\vanalframeel_k}\}_{k=1}^\framesize$ according to
\begin{equation*}
		\vsignal=\sum_{k=1}^\framesize\inner{\vsignal}{\vanalframeel_k}\vsynthframeel_k.
\end{equation*}
In~\fref{eq:finitedimcanondual} we saw that the canonical dual frame can be computed as follows: 
\begin{equation}
	\label{eq:finitedimcanondual1}
	\vsynthframeel_k=(\herm\manalop\manalop)^{-1}\vanalframeel_k,\quad k=1,\ldots,\framesize.
\end{equation}
We already pointed out that the frame operator $\frameop=\ad\analop\analop$ is represented by the matrix $\herm\manalop\manalop$ in the finite-dimensional case. The matrix $(\herm\manalop\manalop)^{-1}$  therefore corresponds to  the operator $\frameop^{-1}$, which will be studied next.

From~\fref{eq:framaABeval} it follows that $\eval_{\min}$, the smallest spectral value of \frameop, satisfies $\eval_{\min}>0$ if $\setframemm$  is a frame. This implies that zero is a regular value~\cite[Def. 7.2-1]{kreyszig89} of \frameop and hence \frameop is invertible on \hilspace, i.e., there exists a unique operator $\frameop^{-1}$ such  that $\frameop\frameop^{-1}=\frameop^{-1}\frameop=\iop_\hilspace$.  
Next, we summarize the properties of $\frameop^{-1}$.

\begin{thm}
\label{thm:smone}
	The following properties hold:
	\begin{enumerate}
		\item $\frameop^{-1}$ is self-adjoint, i.e., $\ad{\left(\frameop^{-1}\right)}=\frameop^{-1}$;
		\item $\frameop^{-1}$ satisfies
			\begin{equation}
			\label{eq:frameSinvorder1}
			\frac{1}{\frameB}=\inf_{\fun\in\hilspace}\frac{\inner{\frameop^{-1}\fun}{\fun}}{\vecnorm{\fun}^2} \ \ \ \text{and} \ \ \   \frac{1}{\frameA}=\sup_{\fun\in\hilspace}\frac{\inner{\frameop^{-1}\fun}{\fun}}{\vecnorm{\fun}^2},
		\end{equation}
		where $\frameA$ and $\frameB$ are the tightest possible frame bounds of $\frameop$;
		\item $\frameop^{-1}$ is positive definite.
	\end{enumerate}
\end{thm}
\begin{proof}
	\begin{enumerate}
	\item
	To prove that $\frameop^{-1}$ is  self-adjoint we write
	\begin{equation*}
		\ad{(\frameop\frameop^{-1})}=\ad{(\frameop^{-1})}\ad\frameop=\iop_\hilspace.
	\end{equation*}
	Since \frameop is self-adjoint, i.e., $\frameop=\ad\frameop$, we conclude that
	\begin{equation*}
		\ad{(\frameop^{-1})}\frameop=\iop_\hilspace. 
	\end{equation*}
	Multiplying by $\frameop^{-1}$ from the right, we finally obtain
	\begin{equation*}
		\ad{(\frameop^{-1})}=\frameop^{-1}.
	\end{equation*}
	\item
	To prove the first equation in~\fref{eq:frameSinvorder1} we write 
	\begin{align}
		\label{eq:Smoproof}
	\frameB=\sup_{\fun\in\hilspace}\frac{\inner{\frameop\fun}{\fun}}{\vecnorm{\fun}^2}
	&=\sup_{\altfun\in\hilspace}\frac{\inner{\frameop\frameop^{1/2}\frameop^{-1}\altfun}{\frameop^{1/2}\frameop^{-1}\altfun}}{\inner{\frameop^{1/2}\frameop^{-1}\altfun}{\frameop^{1/2}\frameop^{-1}\altfun}}\nonumber\\
	&=\sup_{\altfun\in\hilspace}\frac{\inner{\frameop^{-1}\frameop^{1/2}\frameop\frameop^{1/2}\frameop^{-1}\altfun}{\altfun}}{\inner{\frameop^{-1}\frameop^{1/2}\frameop^{1/2}\frameop^{-1}\altfun}{\altfun}}=\sup_{\altfun\in\hilspace}\frac{\inner{\altfun}{\altfun}}{\inner{\frameop^{-1}\altfun}{\altfun}}
	\end{align}
	where the first equality follows from~\fref{eq:framaABeval} and \fref{eq:frameopspecbounds}; in the second equality we used the fact that the operator $\frameop^{1/2}\frameop^{-1}$ is one-to-one on $\hilspace$ and  changed variables according to $\fun=\frameop^{1/2}\frameop^{-1}\altfun$; in the third equality we used the fact that $\frameop^{1/2}$ and $\frameop^{-1}$ are self-adjoint, and in the fourth equality we used $\frameop=\frameop^{1/2}\frameop^{1/2}$. 
The first equation in~\fref{eq:frameSinvorder1} is now obtained by noting that~\fref{eq:Smoproof} implies
	\begin{equation*}
		\frac{1}{\frameB}=1\left/\left(\sup_{\altfun\in\hilspace}\frac{\inner{\altfun}{\altfun}}{\inner{\frameop^{-1}\altfun}{\altfun}}\right)\right.=
		\inf_{\altfun\in\hilspace}\frac{\inner{\frameop^{-1}\altfun}{\altfun}}{\inner{\altfun}{\altfun}}.
	\end{equation*}
	The second equation in~\fref{eq:frameSinvorder1} is proved analogously.
			\item
		Positive-definiteness of $\frameop^{-1}$ follows from the first equation in~\fref{eq:frameSinvorder1} and the fact that $\frameB<\infty$ so that $1/\frameB>0$.	
	\end{enumerate}
\end{proof}

We are now ready to generalize~\fref{eq:finitedimcanondual1} and state the main result on canonical dual frames in the case of general (possibly infinite-dimensional) Hilbert spaces.
\begin{thm}
\label{thm:dualframe}
	Let $\setframemm$ be a frame for the Hilbert space \hilspace with the frame bounds \frameA and \frameB, and let $\frameop$ be the corresponding frame operator. Then, the set $\setdualframemm$ given by
	\begin{equation}
		\label{eq:canonicaldual}
		\synthframeel_k=\frameop^{-1} \analframeel_k, \quad k\in\frameset,
	\end{equation}
	is a frame for \hilspace with the frame bounds $\frameAdual=1/\frameB$ and $\frameBdual=1/\frameA$. 
	
	The analysis operator associated to $\setdualframemm$ 
	defined as
	\begin{align*}
		&\analdualop:\hilspace \to \hilseqspace\\
		&\analdualop: \fun\to \{\inner{\fun}{\synthframeel_k}\}_{k\in\frameset} 
	\end{align*}
	satisfies
	\begin{equation}
		\label{eq:analdualopdef}
		\analdualop=\analop\frameop^{-1}=\analop\left(\ad\analop\analop\right)^{-1}.
	\end{equation}
\end{thm}
\begin{proof}
	Recall that $\frameop^{-1}$ is self-adjoint. 
	Hence, we have $\inner{\fun}{\synthframeel_k}=\inner{\fun}{\frameop^{-1}\analframeel_k}=\inner{\frameop^{-1}\fun}{\analframeel_k}$ for all $\fun\in\hilspace$.
	Thus, using \fref{eq:frameopanalopconnection}, we obtain
	\begin{align*}
	\sum_{k\in\frameset} \abs{\inner{\fun}{\synthframeel_k}}^2
	&
	= \sum_{k\in\frameset} \abs{\inner{\frameop^{-1}\fun}{\analframeel_k}}^2 \\
	&=\inner{\frameop(\frameop^{-1}\fun)}{\frameop^{-1}\fun}=\inner{\fun}{\frameop^{-1}\fun}=\inner{\frameop^{-1}\fun}{\fun}.
	\end{align*}
	Therefore, we conclude from \fref{eq:frameSinvorder1} that
	\begin{equation*}
		\frac{1}{\frameB}\vecnorm{\fun}^2\le \sum_{k\in\frameset}  \abs{\inner{\fun}{\synthframeel_k}}^2\le\frac{1}{\frameA}\vecnorm{\fun}^2,
	\end{equation*}
	i.e., the set $\setdualframemm$ constitutes a frame for \hilspace with frame bounds $\frameAdual=1/\frameB$ and $\frameBdual=1/\frameA$; moreover, it follows from \fref{eq:frameSinvorder1} that $\frameAdual=1/\frameB$ and $\frameBdual=1/\frameA$ are the tightest possible frame bounds. It remains to show that $\analdualop=\analop\frameop^{-1}$:
	\begin{equation*}
	\analdualop\fun=\left\{\inner{\fun}{\synthframeel_k}\right\}_{k\in\frameset}=\left\{\inner{\fun}{\frameop^{-1}\analframeel_k}\right\}_{k\in\frameset}=\left\{\inner{\frameop^{-1}\fun}{\analframeel_k}\right\}_{k\in\frameset}=\analop\frameop^{-1}\fun.
	\end{equation*}
\end{proof}

We call $\setdualframemm$ the \emph{canonical dual frame} associated to the frame $\setframemm$. 
It is convenient to introduce the \emph{canonical dual frame operator}:
\begin{dfn}
	The frame operator associated to the canonical dual frame, 
	\begin{equation}
		\label{eq:dfndualframeop}
		\dualframeop=\ad\analdualop\analdualop,\quad \dualframeop\fun=\sum_{k\in\frameset} \inner{\fun}{\synthframeel_k}\synthframeel_k
	\end{equation}
	is called the canonical dual frame operator.
\end{dfn}
\begin{thm}
	The canonical dual frame operator $\dualframeop$ satisfies $\dualframeop=\frameop^{-1}.$
\end{thm}
\begin{proof}
	For every $\fun\in\hilspace$, we have
\begin{align*}
	\dualframeop\fun&=\sum_{k\in\frameset} \inner{\fun}{\synthframeel_k}\synthframeel_k=\sum_{k\in\frameset} \inner{\fun}{\frameop^{-1}\analframeel_k}\frameop^{-1}\analframeel_k\\
	&=\frameop^{-1}\sum_{k\in\frameset} \inner{\frameop^{-1}\fun}{\analframeel_k}\analframeel_k=\frameop^{-1}\frameop\frameop^{-1}\fun=\frameop^{-1}\fun,
\end{align*}
where in the first equality we used \fref{eq:dfndualframeop}, in the second we used \fref{eq:canonicaldual}, in the third we made use of the fact that $\frameop^{-1}$ is self-adjoint, and in the fourth we used the definition of $\frameop$.
\end{proof}
Note that canonical duality is a reciprocity relation. If the frame $\setdualframemm$ is the canonical dual of the frame $\setframemm$, then $\setframemm$ is the canonical dual of the frame $\setdualframemm$. 
This can be seen by noting that 
\begin{equation*}
	\dualframeop^{-1}\synthframeel_k=(\frameop^{-1})^{-1}\frameop^{-1}\analframeel_k=\frameop\frameop^{-1}\analframeel_k=\analframeel_k.
\end{equation*}

\subsection{Signal Expansions}
\label{sec:signalexpansions}
The following theorem can be considered as one of the \emph{central results in frame theory.} It states that every signal $\fun\in\hilspace$ can be expanded into a frame. The expansion coefficients can be chosen as the inner products of $\fun$ with the canonical dual frame elements. 
\begin{thm}
	\label{thm:frameexpansion}
	Let $\setframemm$ and $\setdualframemm$ be canonical dual frames for the Hilbert space \hilspace. Every signal $\fun\in\hilspace$ can be decomposed as follows
	\begin{align}
		&\fun=\ad\analop\analdualop\fun=\sum_{k\in\frameset} \inner{\fun}{\synthframeel_k}\analframeel_k\nonumber\\
		\label{eq:signalexp2}
		&\fun=\ad\analdualop\analop\fun=\sum_{k\in\frameset} \inner{\fun}{\analframeel_k}\synthframeel_k.
	\end{align}
	Note that, equivalently, we have
	\begin{equation*}
		\ad\analop\analdualop=\ad\analdualop\analop=\iop_\hilspace.
	\end{equation*}
\end{thm}
\begin{proof}
	We have 
	\begin{align*}
		\ad\analop\analdualop\fun&=\sum_{k\in\frameset} \inner{\fun}{\synthframeel_k}\analframeel_k=\sum_{k\in\frameset} \inner{\fun}{\frameop^{-1}\analframeel_k}\analframeel_k\\
		&=\sum_{k\in\frameset} \inner{\frameop^{-1}\fun}{\analframeel_k}\analframeel_k=\frameop\frameop^{-1}\fun=\fun.
	\end{align*}
	This proves that $\ad\analop\analdualop=\iop_\hilspace$. The proof of $\ad\analdualop\analop=\iop_\hilspace$ is similar.
\end{proof}
Note that \fref{eq:signalexp2}  corresponds to the decomposition \fref{eq:framedecomp1}  we found in the finite-dimensional case.

It is now natural to ask whether reconstruction of $\fun$ from the coefficients $\inner{\fun}{\analframeel_k},\, k\in\frameset,$ according to~\fref{eq:signalexp2} is the only way of recovering $\fun$ from $\inner{\fun}{\analframeel_k},\, k\in\frameset$. Recall that we showed in the finite-dimensional case (see~\fref{sec:redsigexp})
that for each complete and  redundant set of vectors $\{\vanalframeel_1,\ldots,\vanalframeel_\framesize\}$, there are infinitely many dual sets $\{\vsynthframeel_1,\ldots,\vsynthframeel_\framesize\}$ that can be used to reconstruct a signal $\vsignal$  from the coefficients 
$\inner{\vsignal}{\vanalframeel_k},\, k=1,\ldots,\framesize,$
according to~\fref{eq:framedecomp1}.
These dual sets are obtained by identifying $\{\vsynthframeel_1,\ldots,\vsynthframeel_\framesize\}$ with the columns of $\matL$, where $\matL$ is a left-inverse of the analysis matrix~\manalop. In the infinite-dimensional case the question of finding all dual frames for a given frame boils down to finding, for a given analysis operator \analop, all linear operators \lianalop that satisfy
\begin{equation*}
	\lianalop\analop\fun=\fun
\end{equation*}
for all $\fun\in\hilspace$. In other words, we want to identify all left-inverses $\lianalop$ of the analysis operator $\analop$.
The answer to this question is the infinite-dimensional version of \fref{thm:leftinv} 
that we state here without proof.
\begin{thm}
\label{thm:leftinvop}
	Let $\opA:\hilspace\to\hilseqspace$ be a bounded linear  operator. Assume that $\ad\opA\opA:\hilspace\to\hilspace$ is invertible on $\hilspace$. Then, the operator $\pinv \opA: \hilseqspace\to\hilspace$ defined as $\pinv \opA\define (\ad \opA \opA)^{-1}\ad \opA$ is a left-inverse of $\opA$, i.e., $\pinv \opA\opA=\iop_{\hilspace}$, where $\iop_{\hilspace}$ is the identity operator on \hilspace.
	Moreover, the general solution $\opL$ of the equation $\opL\opA=\iop_{\hilspace}$ is given by
	\begin{equation*}
		\opL=\pinv \opA+\opM(\iop_{\hilseqspace}-\opA\pinv \opA)
	\end{equation*}
	where $\opM:\hilseqspace\to\hilspace$ is an arbitrary  bounded linear operator and $\iop_{\hilseqspace}$ is the identity operator on \hilseqspace.
\end{thm}
Applying this theorem to the operator $\analop$ we see that all left-inverses of $\analop$ can be written as 
\begin{equation}
	\label{eq:generalliopframes}
	\opL=\pinv \analop+\opM(\iop_{\hilseqspace}-\analop\pinv \analop)
\end{equation}
where $\opM:\hilseqspace\to\hilspace$ is an arbitrary  bounded linear operator and 
\begin{equation*}
	\pinv \analop=(\ad \analop \analop)^{-1}\ad \analop. 
\end{equation*}
Now, using \fref{eq:analdualopdef}, we obtain the following important identity: 
\begin{equation*}
	\pinv \analop=(\ad \analop \analop)^{-1}\ad \analop=\frameop^{-1}\ad \analop=\ad\analdualop.
\end{equation*}
This shows that reconstruction according to \fref{eq:signalexp2}, i.e., by applying the operator $\ad\analdualop$ to the coefficient sequence $\analop\fun=\{\inner{\fun}{\analframeel_k}\}_{k\in\frameset}$ is nothing but applying the infinite-dimensional analog of the Moore-Penrose inverse $\pinv \manalop=(\herm \manalop \manalop)^{-1}\herm \manalop$.
As already noted in the finite-dimensional case the existence of infinitely many left-inverses of the operator \analop provides us with freedom in designing dual frames. 

We close this discussion with a geometric interpretation of the parametrization~\fref{eq:generalliopframes}. First observe the following.
\begin{thm}
The operator
\begin{equation*}
\pop:\sqsumspace\to\rng(\analop) \subseteq \sqsumspace
\end{equation*} 
defined as
\begin{equation*}
\pop=\analop\frameop^{-1}\ad\analop
\end{equation*} 
satisfies the following properties:
\begin{enumerate}
	\item \pop is the identity  operator $\iop_{\sqsumspace}$ on~$\rng(\analop)$.
	\item \pop is the zero operator on $\ocomp{\rng(\analop)}$, where $\ocomp{\rng(\analop)}$ denotes the orthogonal complement of the space~$\rng(\analop)$.
\end{enumerate}
In other words, \pop is the orthogonal projection operator onto 
$\rng(\analop)=\{\coeflt\given\coeflt=\analop \fun, \fun\in\hilspace\}$, the range space of the operator \analop.
\end{thm}
\begin{proof}
\begin{enumerate}
	\item Take a sequence $\coeflt\in\rng(\analop)$ and note that it can be written as $\coeflt=\analop\fun$, where $\fun\in\hilspace$. Then, we have 
	\begin{equation*}
	\pop \coeflt=\analop\frameop^{-1}\ad\analop\analop\fun =\analop\frameop^{-1}\frameop\fun=\analop\iop_\hilspace\fun=\analop\fun=\coeflt.
	\end{equation*}
	This proves that \pop is the identity operator on $\rng(\analop)$.
	\item Next, take a sequence $\coeflt\in\ocomp{\rng(\analop)}$. As the orthogonal complement of the range space of an operator is the null space of its adjoint, 	 we have $\ad\analop \coeflt=0$ and therefore
	\begin{equation*}
	\pop \coeflt= \analop\frameop^{-1}\ad\analop \coeflt=0.
	\end{equation*}
	This proves that \pop is the zero operator on $\ocomp{\rng(\analop)}$.
\end{enumerate}
\end{proof} 
Now using that 
$\analop\pinv \analop=\analop\frameop^{-1}\ad\analop=\pop$ and $\pinv \analop=\frameop^{-1}\ad\analop=\frameop^{-1}\frameop\frameop^{-1}\ad\analop=\frameop^{-1}\ad\analop\analop\frameop^{-1}\ad\analop=\ad\analdualop\pop,$
we can rewrite \fref{eq:generalliopframes} as follows
\begin{equation}
	\label{eq:generalliopframes1}
	\opL=\ad\analdualop \pop +\opM(\iop_{\hilseqspace}-\pop).
\end{equation}
Next, we show that $(\iop_{\hilseqspace}-\pop):\hilseqspace\to\hilseqspace$ is the orthogonal projection onto $\ocomp{\rng(\analop)}$.
Indeed, we can directly verify the following: For every $\coeflt\in \ocomp{\rng(\analop)}$, we have $(\iop_{\hilseqspace}-\pop)\coeflt=\iop_{\hilseqspace}\coeflt-0=\coeflt$, i.e., $\iop_{\hilseqspace}-\pop$ is the identity operator on $\ocomp{\rng(\analop)}$; for every $\coeflt\in \ocomp{(\ocomp{\rng(\analop)})}=\rng(\analop)$, we have $(\iop_{\hilseqspace}-\pop)\coeflt=\iop_{\hilseqspace}\coeflt-\coeflt=0$, i.e., $\iop_{\hilseqspace}-\pop$ is the zero operator on $\ocomp{(\ocomp{\rng(\analop)})}$.  

We are now ready to re-interpret~\fref{eq:generalliopframes1} as follows. Every left-inverse~$\opL$ of $\analop$ acts as $\ad\analdualop$ (the synthesis operator of the canonical dual frame) on the range space of the analysis operator $\analop$, and can act in an arbitrary linear and bounded fashion on the orthogonal complement of the range space of the analysis operator $\analop$.

\subsection{Tight Frames}
The frames considered in Examples~\ref{ex:tightframeido} and~\ref{ex:tightframeidt} above have an interesting property: In both cases the tightest possible frame bounds \frameA and \frameB are equal. Frames with this property  are called tight frames.
\begin{dfn}
\label{dfn:tightframe}
	A frame $\setframemm$ with tightest possible frame bounds $\frameA=\frameB$ is called a tight frame.
\end{dfn}
Tight frames are of significant practical interest because of the following central fact.
\begin{thm}
\label{thm:tightframe}
	Let $\setframemm$ be a frame for the Hilbert space \hilspace. The frame  $\setframemm$ is tight with frame bound $\frameA$ if and only if its corresponding frame  operator satisfies $\frameop=\frameA\iop_\hilspace$,
	or equivalently, if
	\begin{equation}
		\label{eq:rectight}
		\fun=\frac{1}{\frameA} \sum_{k\in\frameset}\inner{\fun}{\analframeel_k}\analframeel_k
	\end{equation}
	for all $\fun\in\hilspace$.
\end{thm}
\begin{proof}
	First observe that $\frameop=\frameA\iop_\hilspace$ is equivalent to $\frameop\fun = \frameA\iop_\hilspace\fun= \frameA\fun$ for all $\fun\in\hilspace$, which, in turn, is equivalent to~\fref{eq:rectight} by definition of the frame operator.
	
	To prove that tightness of $\setframemm$ implies $\frameop=\frameA\iop_\hilspace$, note that by~\fref{dfn:tightframe}, using \fref{eq:framecondS} we can write
	\begin{equation*}
		\inner{\frameop\fun}{\fun} = \frameA\inner{\fun}{\fun}, \ \text{for all}\ \fun\in\hilspace.
	\end{equation*}
	Therefore 
	\begin{equation*}
		\inner{(\frameop-\frameA\iop_\hilspace)\fun}{\fun} = 0,  \ \text{for all}\ \fun\in\hilspace,
	\end{equation*}
	which implies $\frameop=\frameA\iop_\hilspace$.
	
	To prove that $\frameop=\frameA\iop_\hilspace$ implies tightness of $\setframemm$, we take the inner product with \fun on both sides of~\fref{eq:rectight} to obtain
	\begin{equation*}
		\inner{\fun}{\fun}=\frac{1}{\frameA} \sum_{k\in\frameset}\inner{\fun}{\analframeel_k}\inner{\analframeel_k}{\fun}.
	\end{equation*}
	This is equivalent to 
	\begin{equation*}
		\frameA	\vecnorm{\fun}^2= \sum_{k\in\frameset}\abs{\inner{\fun}{\analframeel_k}}^2,
	\end{equation*}
	which shows that $\setframemm$ is a tight frame for \hilspace with frame bound equal to \frameA.
\end{proof}

The practical importance of tight frames lies in the fact that they make the computation of the canonical dual frame, which in the general case requires inversion of an operator and application of this inverse to  all frame elements, particularly simple. Specifically, we have: 
\begin{equation*}
	\synthframeel_k=\frameop^{-1}\analframeel_k=\frac{1}{\frameA}\iop_\hilspace\analframeel_k=\frac{1}{\frameA}\analframeel_k.
\end{equation*}

A well-known example of a tight frame for $\reals^2$ is the following:
\begin{example}[The Mercedes-Benz frame~\cite{kovacevic08}]
	\label{ex:MBframetight}
	The Mercedes-Benz frame (see \Figref{fig:MBframe}) is given by the  following three vectors in $\reals^2$:
	\begin{equation}
		\label{eq:mbframevectors}
		\vanalframeel_1=\begin{bmatrix} 0\\1\end{bmatrix}, \quad \vanalframeel_2=\begin{bmatrix} -\sqrt{3}/2 \\ -1/2 \end{bmatrix}, \quad \vanalframeel_3=\begin{bmatrix} \sqrt{3}/2\\ -1/2\end{bmatrix}.
	\end{equation}
	\begin{figure}[t]
		\centering
		\includegraphics[]{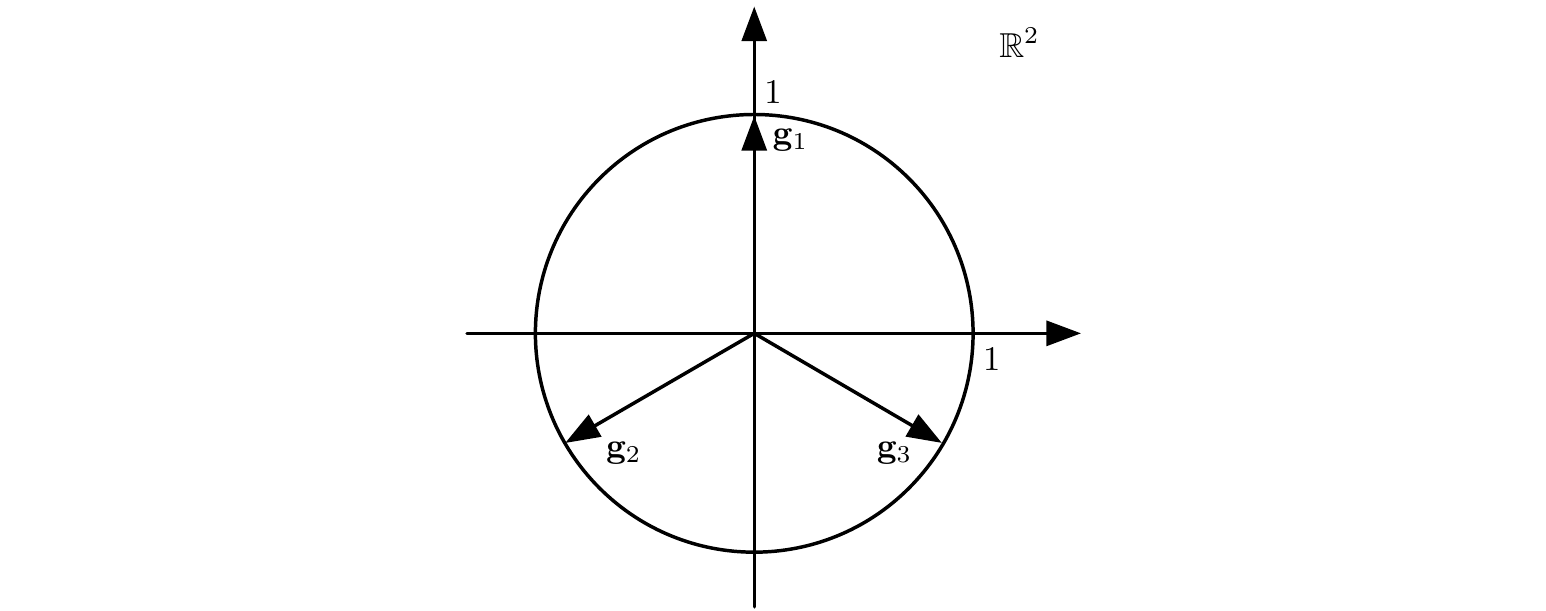}
		\caption{The Mercedes-Benz frame.}
		\label{fig:MBframe}
	\end{figure}
	To see that this frame is indeed tight, note that its analysis operator $\analop$ is given by the matrix
	\begin{equation*}
	\manalop=\begin{bmatrix} 0&1\\-\sqrt{3}/2 & -1/2 \\ \sqrt{3}/2 & -1/2 \end{bmatrix}.
	\end{equation*}
	The adjoint $\ad\analop$ of the analysis operator is given by the matrix
	\begin{equation*}
	\herm\manalop=\begin{bmatrix} 0 & -\sqrt{3}/2 & \sqrt{3}/2\\1 & -1/2 &  -1/2  \end{bmatrix}.	
	\end{equation*}
	Therefore, the frame operator $\frameop$ is represented by the matrix
	\begin{equation*}
		\mframeop=\herm\manalop\manalop=\begin{bmatrix} 0 & -\sqrt{3}/2 & \sqrt{3}/2\\ 1 & -1/2 &  -1/2  \end{bmatrix}\begin{bmatrix} 0&1\\-\sqrt{3}/2 & -1/2 \\ \sqrt{3}/2 & -1/2 \end{bmatrix} =\frac{3}{2} \begin{bmatrix}
			1 & 0\\0 & 1
		\end{bmatrix}=\frac{3}{2} \identity_2,
	\end{equation*}
and hence $\frameop=\frameA\iop_{\reals^2}$ with  $\frameA=3/2$, which implies, by~\fref{thm:tightframe}, that $\{\vanalframeel_1,\vanalframeel_2,\vanalframeel_3\}$ is a tight frame (for $\reals^2$).
\end{example}

The design of tight frames is challenging in general. It is hence interesting to devise simple systematic methods for obtaining tight frames. 
The following theorem shows how we can obtain a tight frame from a given general frame.
\begin{thm}
Let $\setframemm$ be a frame for the Hilbert space \hilspace with frame operator \frameop. Denote the positive definite square root of $\frameop^{-1}$ by $\frameop^{-1/2}$. Then $\{\frameop^{-1/2}\analframeel_k\}_{k\in\frameset}$  is a tight frame for \hilspace with frame bound $\frameA=1$, i.e.,
\begin{equation*}
\fun=\sum_{k\in\frameset} \inner{\fun}{\frameop^{-1/2}\analframeel_k}\frameop^{-1/2}\analframeel_k,\quad\text{for all } \fun\in\hilspace.
\end{equation*}
\end{thm}
\begin{proof}
	Since $\frameop^{-1}$ is self-adjoint and positive definite by~\fref{thm:smone}, it has, by~\fref{lem:possqrt}, a unique self-adjoint positive definite square root $\frameop^{-1/2}$ that commutes with $\frameop^{-1}$. Moreover $\frameop^{-1/2}$ also commutes with $\frameop$, which can be seen as follows: 
	\begin{align*}
	\frameop^{-1/2}\frameop^{-1}=\frameop^{-1}\frameop^{-1/2}\\
	\frameop\frameop^{-1/2}\frameop^{-1}=\frameop^{-1/2}\\
	\frameop\frameop^{-1/2}=\frameop^{-1/2}\frameop.
	\end{align*}
The proof is then effected by noting the following:
	\begin{align*}
	\fun&=\frameop^{-1}\frameop\fun=\frameop^{-1/2}\frameop^{-1/2}\frameop\fun\\
	&=\frameop^{-1/2}\frameop\frameop^{-1/2}\fun\\
	&=\sum_{k\in\frameset}\inner{\frameop^{-1/2}\fun}{\analframeel_k}\frameop^{-1/2}\analframeel_k\\
	&=\sum_{k\in\frameset}\inner{\fun}{\frameop^{-1/2}\analframeel_k}\frameop^{-1/2}\analframeel_k.
	\end{align*}
\end{proof}

It is evident that every \onbac is a tight frame with $\frameA=1$. Note, however, that conversely a tight frame (even with $\frameA=1$) need not be an orthonormal or orthogonal basis, as can be seen from \fref{ex:tightframeidt}. However, as the next theorem shows, a tight frame with~$\frameA=1$ and $\vecnorm{\analframeel_k}=1,$ for all~$k\in\frameset,$ is necessarily an \onbac.

\begin{thm} 
	\label{thm:tightframebasis}
A tight frame \setframemm for the Hilbert space \hilspace with $\frameA=1$ and $\vecnorm{\analframeel_k}=1$, for all~$k\in\frameset,$ is an~\onbac for~\hilspace.
\end{thm}
\begin{proof}
	Combining 
	\begin{equation*}
		\inner{\frameop\analframeel_k}{\analframeel_k}=\frameA\vecnorm{\analframeel_k}^2=\vecnorm{\analframeel_k}^2
	\end{equation*}
	with
	\begin{equation*}
		\inner{\frameop\analframeel_k}{\analframeel_k}=\sum_{j\in\frameset} \abs{\inner{\analframeel_k}{\analframeel_j}}^2=\vecnorm{\analframeel_k}^4+\sum_{j\ne k}\abs{\inner{\analframeel_k}{\analframeel_j}}^2
	\end{equation*}
	we obtain
	\begin{equation*}
		\vecnorm{\analframeel_k}^4+\sum_{j\ne k}\abs{\inner{\analframeel_k}{\analframeel_j}}^2=\vecnorm{\analframeel_k}^2.
	\end{equation*}
	Since $\vecnorm{\analframeel_k}^2=1,$ for all $k\in\frameset,$ it follows that $\sum_{j\ne k}\abs{\inner{\analframeel_k}{\analframeel_j}}^2=0$, for all $k\in\frameset$. This implies that the elements of $\{\analframeel_j\}_{j\in\frameset}$ are necessarily orthogonal to each other. 
\end{proof}

There is an elegant result that tells us that every tight frame with  frame bound  $\frameA=1$ can be realized as an orthogonal projection of an \onbac from a space with larger dimension. This result is known as Naimark's theorem. Here we state the finite-dimensional version of this theorem, for the infinite-dimensional version see\cite{han00}.
\begin{thm}[Naimark, {\cite[Prop. 1.1]{han00}}]
\label{thm:naimark}
	Let $\framesize>\dimension$. Suppose that the set $\{\vanalframeel_1,\ldots,\vanalframeel_\framesize\}$, $\vanalframeel_k\in\hilspace,\, k=1,\ldots,\framesize$, is a tight frame for an $\dimension$-dimensional Hilbert space \hilspace  with frame bound $\frameA=1$. Then, there exists an $\framesize$-dimensional Hilbert space $\setK\supset\hilspace$ and an \onbac $\{\vbasisel_1,\ldots,\vbasisel_\framesize\}$ for $\setK$ 
such that $\opP\vbasisel_k=\vanalframeel_k$, $k=1,\ldots,\framesize$, where $\opP:\setK\to\setK$ is the orthogonal projection onto $\hilspace$.
\end{thm}
 We omit the proof and illustrate the theorem by an example instead.
\begin{example}
	\label{ex:naimark}
	Consider the Hilbert space $\setK=\reals^3$,  and assume that $\hilspace\subset\setK$ is the plane spanned by the vectors~$\tp{[1\ 0\ 0]}$ and~$\tp{[0\ 1\ 0]}$, i.e.,
	\begin{equation*}
		\hilspace=\spn\left\{\tp{[1\ 0\ 0]}, \tp{[0\ 1\ 0]} \right\}.
	\end{equation*} 
We can construct a tight frame for $\hilspace$ with three elements and frame bound $\frameA=1$ if we rescale the Mercedes-Benz frame from~\fref{ex:MBframetight}. Specifically, consider the vectors $\vanalframeel_k,\, k=1,2,3,$ defined in \fref{eq:mbframevectors} and let $\vanalframeel'_k\define\sqrt{2/3}\,\vanalframeel_k,\, k=1,2,3$. In the following, we think about the two-dimensional vectors $\vanalframeel'_k$ as being embedded into the three-dimensional space \setK with the third coordinate (in the standard basis of \setK) being equal to zero. Clearly, $\{\vanalframeel'_k\}_{k=1}^3$ is a tight frame for $\hilspace$  with  frame bound $\frameA=1$. 
Now consider the following three vectors in $\setK$:
\begin{equation*}
		\vbasisel_1=\begin{bmatrix} 0\\\sqrt{2/3}\\-1/\sqrt{3}\end{bmatrix}, \quad \vbasisel_2=\begin{bmatrix} -1/\sqrt{2}\\-1/\sqrt{6} \\-1/\sqrt{3} \end{bmatrix}, \quad \vbasisel_3=\begin{bmatrix} 1/\sqrt{2}\\-1/\sqrt{6} \\-1/\sqrt{3} \end{bmatrix}.
\end{equation*}
Direct calculation reveals that $\{\vbasisel_k\}_{k=1}^3$ is an \onbac for $\setK$.
Observe that the frame vectors $\vanalframeel'_k,\, k=1,2,3,$ can be obtained from the \onbac vectors $\vbasisel_k,\, k=1,2,3,$ by applying the   orthogonal projection from $\setK$ onto $\hilspace$:
\begin{equation*}
	\matP\define\begin{bmatrix}
		1& 0& 0\\
		0& 1& 0\\
		0& 0& 0
	\end{bmatrix},
\end{equation*}
according to $\vanalframeel'_k=\matP\vbasisel_k,\ k=1,2,3$. This illustrates Naimark's  theorem. 
\end{example}

\subsection{Exact Frames and Biorthonormality}
\label{sec:exactframes}
In~\fref{sec:genbasis} we studied expansions of signals in $\complexset^\dimension$ into (not necessarily orthogonal) bases.  The main results we established in this context can  be summarized as follows:
\begin{enumerate}
	\item The number of vectors in a basis is always equal to the dimension of the Hilbert space under consideration. Every set of vectors that spans $\complexset^\dimension$ and has more than $M$ vectors is necessarily redundant, i.e., the vectors in this set are linearly dependent. Removal of an arbitrary vector from a basis for $\complexset^\dimension$ leaves a set that no longer spans $\complexset^\dimension$. 
	\item
	For a given basis $\{\vbasisel_k\}_{k=1}^\dimension$ every signal $\vsignal\in\complexset^\dimension$ has a \emph{unique} representation according to 
	\begin{equation}
	\vsignal = \sum_{k=1}^\dimension \inner{\vsignal}{\vbasisel_k}\vsynthbasisel_k.
	\label{eq:vsexpansion}
	\end{equation}
	The basis $\{\vbasisel_k\}_{k=1}^\dimension$ and its dual basis $\{\vsynthbasisel_k\}_{k=1}^\dimension$ satisfy the biorthonormality relation \fref{eq:biorthog}.
\end{enumerate}

The theory of \onbacp in infinite-dimensional spaces is well-developed. 
In this section, we ask how the concept of general (i.e., not necessarily orthogonal) bases can be extended to infinite-dimensional spaces. Clearly, in the infinite-dimensional case, we can not simply say that the number of elements in a basis must be equal to the dimension of the Hilbert space. However, we can use the property that removing an element from a basis, leaves us with an incomplete set of vectors to motivate the following definition.
\begin{dfn}
Let  $\setframemm$ be a frame for the Hilbert space $\hilspace$. We call the frame $\setframemm$ \emph{exact} if, for all $m\in\frameset$, the set $\{\synthframeelrev_{k}\}_{k\ne m}$ is incomplete for $\hilspace$; 
 we call the frame $\setframemm$ \emph{inexact} if there is at least one element $\analframeel_m$ that can be removed from the frame, so that the  set $\{\analframeel_k\}_{k\ne m}$ is again a frame for $\hilspace$.	
\end{dfn}

There are two more properties of general bases in finite-dimensional spaces that carry over to the infinite-dimensional case, namely uniqueness of representation in the sense of \eqref{eq:vsexpansion} and biorthonormality between the frame and its canonical dual. To show that representation of a signal in an exact frame is unique and that an exact frame is biorthonormal to its canonical dual frame, we will need the following two lemmas.

Let $\setdualframerevmm$ and $\setframerevmm$
be canonical dual frames. The first lemma below states that for a fixed~$\fun\in\hilspace$, among all possible expansion coefficient 
sequences $\coeflt$ satisfying $\fun=\sum_{k\in\frameset} \coef_{k} \synthframeelrev_k$, the coefficients $\coef_{k}=\inner{\fun}{\analframeelrev_k}$ have
minimum $\sqsumspace$-norm. 
\begin{lem}[{\citepar{heil89-12}}]
\label{lem:minimumnorm}
Let $\setdualframerevmm$ be a frame for the Hilbert space \hilspace and $\setframerevmm$  its canonical dual frame. For a fixed~$\fun\in\hilspace$, let $\coef_{k}=\inner{\fun}{\analframeelrev_k}$ so that $\fun=\sum_{k\in\frameset} \coef_{k}\synthframeelrev_k$.
If it is possible to find scalars $\{\coefalt_{k}\}_{k\in\frameset}\ne \{\coef_{k}\}_{k\in\frameset}$ such that
$\fun=\sum_{k\in\frameset} \coefalt_{k}\synthframeelrev_k$, then we must have 
\begin{equation}
\sum_{k\in\frameset} \abs{\coefalt_{k}}^{2}=\sum_{k\in\frameset} \abs{\coef_{k}}^{2}+\sum_{k\in\frameset} \abs{\coef_{k}-\coefalt_{k}}^2.
\label{eq:min1}
\end{equation}	
\end{lem}

\begin{proof}
We have 
\begin{equation*}
\coef_{k}=\inner{\fun}{\analframeelrev_k}=\inner{\fun}{\dualframeoprev^{-1}\synthframeelrev_k}=\inner{\dualframeoprev^{-1}\fun}{\synthframeelrev_k}=\inner{\tilde\fun}{\synthframeelrev_k} 
\end{equation*}
with $\tilde\fun= \dualframeoprev^{-1} \fun$. Therefore,
\begin{equation*}
\inner{\fun}{\tilde\fun}=\inner{\sum_{k\in\frameset} \coef_{k} \synthframeelrev_k }{\tilde\fun}=\sum_{k\in\frameset} \coef_{k} \inner{\synthframeelrev_k }{\tilde\fun}=\sum_{k\in\frameset} \coef_{k} \conj\coef_{k}=\sum_{k\in\frameset} \abs{\coef_{k}}^2
\end{equation*}
and 
\begin{equation*}
\inner{\fun}{\tilde\fun}=\inner{\sum_{k\in\frameset} \coefalt_{k} \synthframeelrev_k }{\tilde\fun}=\sum_{k\in\frameset} \coefalt_{k} \inner{\synthframeelrev_k }{\tilde\fun}=\sum_{k\in\frameset} \coefalt_{k} \conj\coef_{k}.
\end{equation*}
We can therefore conclude that
\begin{equation}
	\label{eq:aux1}
\sum_{k\in\frameset} \abs{\coef_{k}}^2=\sum_{k\in\frameset} \coefalt_{k} \conj\coef_{k}=\sum_{k\in\frameset} \conj\coefalt_{k} \coef_{k}.
\end{equation}
Hence,
\begin{align*}
\sum_{k\in\frameset} \abs{\coef_{k}}^2+\sum_{k\in\frameset} \abs{\coef_{k}-\coefalt_{k}}^2 
&=\sum_{k\in\frameset} \abs{\coef_{k}}^2+\sum_{k\in\frameset} \left(\coef_{k}-\coefalt_{k}\right)\left(\conj\coef_{k}-\conj\coefalt_{k}\right)\\
&=\sum_{k\in\frameset} \abs{\coef_{k}}^2+\sum_{k\in\frameset}\abs{\coef_{k}}^2-\sum_{k\in\frameset}\coef_{k}\conj\coefalt_{k}-\sum_{k\in\frameset}\conj\coef_{k}\coefalt_{k}+\sum_{k\in\frameset} \abs{\coefalt_{k}}^2.
\end{align*}
Using \fref{eq:aux1}, we get
\begin{equation*}
	\sum_{k\in\frameset} \abs{\coef_{k}}^2+\sum_{k\in\frameset} \abs{\coef_{k}-\coefalt_{k}}^2 =\sum_{k\in\frameset} \abs{\coefalt_{k}}^2.
\end{equation*}
\end{proof} 

Note that this lemma implies $\sum_{k\in\frameset} \abs{\coefalt_{k}}^{2} > \sum_{k\in\frameset} \abs{\coef_{k}}^{2}$, i.e., the coefficient sequence $\{\coefalt_{k}\}_{k\in\frameset}$ has larger $l^{2}$-norm than the coefficient sequence $\{\coef_{k}=\inner{\fun}{\analframeelrev_k}\}_{k\in\frameset}$.

\begin{lem}[{\citepar{heil89-12}}]
\label{lem:innerprodexpr}
Let $\setdualframerevmm$ be a frame for the Hilbert space \hilspace and $\setframerevmm$ its canonical dual frame. Then for each $m\in\frameset$,
we have
\begin{equation*}
\sum_{k \ne m}\abs{\inner{\synthframeelrev_m}{\analframeelrev_k}}^2= \frac{	1-\abs{\inner{\synthframeelrev_m}{\analframeelrev_m}}^2-\abs{1-\inner{\synthframeelrev_m}{\analframeelrev_m}}^2}{2}.
\end{equation*}
\end{lem}

\begin{proof} 
We can represent $\synthframeelrev_{m}$ in two different ways. Obviously $\synthframeelrev_{m}=\sum_{k\in\frameset} \coefalt_{k} \synthframeelrev_{k}$ with $\coefalt_{m}=1$ and $\coefalt_{k}=0$ for $k\ne m$, so that
$\sum_{k\in\frameset} \abs{\coefalt_{k}}^{2}=1.$ Furthermore, we can write $\synthframeelrev_{m}=\sum_{k\in\frameset} \coef_{k} \synthframeelrev_{k}$ with $\coef_{k}=\inner{\synthframeelrev_m}{\analframeelrev_k}$. From \fref{eq:min1} it then follows that 
\begin{align*}
1&=\sum_{k\in\frameset}\abs{\coefalt_{k}}^{2}=\sum_{k\in\frameset}\abs{\coef_{k}}^{2}+\sum_{k\in\frameset}\abs{\coef_{k}-\coefalt_{k}}^{2}\\
&=\sum_{k\in\frameset}\abs{\coef_{k}}^{2}+\abs{\coef_{m}-\coefalt_{m}}^{2}+\sum_{k\ne m}\abs{\coef_{k}-\coefalt_{k}}^{2}\\
&=\sum_{k\in\frameset}\abs{\inner{\synthframeelrev_m}{\analframeelrev_k}}^{2}+\abs{\inner{\synthframeelrev_m}{\analframeelrev_m}-1}^{2}+\sum_{k\ne m}\abs{\inner{\synthframeelrev_m}{\analframeelrev_k}}^{2}\\
&=2\sum_{k\ne m}\abs{\inner{\synthframeelrev_m}{\analframeelrev_k}}^{2}+\abs{\inner{\synthframeelrev_m}{\analframeelrev_m}}^{2}+\abs{1-\inner{\synthframeelrev_m}{\analframeelrev_m}}^{2}
\end{align*}
and hence
\begin{equation*}
\sum_{k \ne m}\abs{\inner{\synthframeelrev_m}{\analframeelrev_k}}^{2}\,=\,\frac{1-\abs{\inner{\synthframeelrev_m}{\analframeelrev_m}}^{2}-
\abs{1-\inner{\synthframeelrev_m}{\analframeelrev_m}}^{2}}{2}.
\end{equation*}
\end{proof}

We are now able to formulate an equivalent condition for a frame to be exact.

\begin{thm}[{\citepar{heil89-12}}]
\label{thm:exactframecond}
Let $\setdualframerevmm$ be a frame for the Hilbert space \hilspace and $\setframerevmm$ its canonical dual frame. Then, 
\begin{enumerate}
	\item 
	\setdualframerevmm is exact if and only if $\inner{\synthframeelrev_{m}}{\analframeelrev_{m}}=1$ for all $m\in\frameset$;
	\item
	\setdualframerevmm is inexact if and only if there exists at least one $m\in\frameset$ such that $\inner{\synthframeelrev_{m}}{\analframeelrev_{m}}\ne 1$.
\end{enumerate}
\end{thm}	

\begin{proof}
We first show that if $\inner{\synthframeelrev_{m}}{\analframeelrev_{m}}=1$ for all $m\in\frameset,$ then $\{\synthframeelrev_{k}\}_{k\ne m}$ is incomplete for \hilspace (for all $m\in\frameset$) and hence $\setdualframerevmm$ is
an exact frame for \hilspace. Indeed, fix an arbitrary $m\in\frameset.$
From \fref{lem:innerprodexpr} we have
\begin{equation*}
\sum_{k\ne m}\abs{\inner{\synthframeelrev_m}{\analframeelrev_k}}^{2}\,=\,\frac{1-\abs{\inner{\synthframeelrev_m}{\analframeelrev_m}}^{2}-
\abs{1-\inner{\synthframeelrev_m}{\analframeelrev_m}}^{2}}{2}.
\end{equation*}
Since  $\inner{\synthframeelrev_{m}}{\analframeelrev_{m}}=1,$ we have 
$\sum_{k\ne m}\abs{\inner{\synthframeelrev_m}{\analframeelrev_k}}^{2}=0$ so that $\inner{\synthframeelrev_m}{\analframeelrev_k}=\inner{\analframeelrev_m}{\synthframeelrev_k}=0$ for all
$k \ne m$. 
But~$\analframeelrev_m \ne 0$ since $\inner{\synthframeelrev_m}{\analframeelrev_m}=1.$
Therefore, $\{\synthframeelrev_{k}\}_{k\ne m}$ is incomplete for \hilspace, because 
$\analframeelrev_m\ne 0$ is orthogonal to all elements of the set
$\{\synthframeelrev_{k}\}_{k\ne m}$.

Next, we show that if there exists at least one $m\in\frameset$ such that $\inner{\synthframeelrev_{m}}{\analframeelrev_{m}} \ne 1$, then $\setdualframerevmm$ is inexact. More specifically, we will show that
$\{\synthframeelrev_{k}\}_{k \ne m}$ is still a frame for \hilspace if $\inner{\synthframeelrev_{m}}{\analframeelrev_{m}} \ne 1$.
We start by noting that
\begin{equation}
	\label{eq:gmdecomp}
\synthframeelrev_{m}=\sum_{k\in\frameset} \big<\synthframeelrev_{m},\analframeelrev_{k}\big>\synthframeelrev_{k}=\big<\synthframeelrev_{m},\analframeelrev_{m}\big>\synthframeelrev_{m}+\sum_{k\ne m}\big<\synthframeelrev_{m},\analframeelrev_{k}\big>\synthframeelrev_{k}.
\end{equation}
If $\inner{\synthframeelrev_{m}}{\analframeelrev_{m}} \ne 1$, \fref{eq:gmdecomp} can be rewritten as 
\begin{equation*}
\synthframeelrev_{m}=\frac{1}{1-\inner{\synthframeelrev_{m}}{\analframeelrev_{m}}}\sum_{k \ne m}\inner{\synthframeelrev_{m}}{\analframeelrev_{k}}\synthframeelrev_{k},
\end{equation*}
and for every $\fun \in \hilspace$ we have
\begin{align*}
\abs{\inner{\fun}{\synthframeelrev_{m}}}^{2}&=\abs{\frac{1}{1-\inner{\synthframeelrev_{m}}{\analframeelrev_{m}}}}^{2}\abs{\sum_{k \ne m}
\inner{\synthframeelrev_{m}}{\analframeelrev_{k}}\inner{\fun}{\synthframeelrev_{k}}}^{2}\\
&\le\frac{1}{\abs{1-\inner{\synthframeelrev_{m}}{\analframeelrev_{m}}}^{2}}\left[ \sum_{k \ne m}
\abs{\inner{\synthframeelrev_{m}}{\analframeelrev_{k}}}^{2}\right] \left[\sum_{k \ne m} \abs{\inner{\fun}{\synthframeelrev_{k}}}^{2}\right].
\end{align*}
Therefore
\begin{align*}
\sum_{k\in\frameset}\abs{\inner{\fun}{\synthframeelrev_{k}}}^{2}&=\abs{\inner{\fun}{\synthframeelrev_{m}}}^{2}+\sum_{k \ne
m}\abs{\inner{\fun}{\synthframeelrev_{k}}}^{2}\\
&\le\frac{1}{\abs{1-\inner{\synthframeelrev_{m}}{\analframeelrev_{m}}}^{2}}\left[\sum_{k\ne m}\abs{\inner{\synthframeelrev_{m}}{\analframeelrev_{k}}}^{2}\right]\left[\sum_{k\ne m}\abs{\inner{\fun}{\synthframeelrev_{k}}}^{2}\right]+\sum_{k\ne m}\abs{\inner{\fun}{\synthframeelrev_{k}}}^{2}\\
&=\sum_{k\ne m}\abs{\inner{\fun}{\synthframeelrev_{k}}}^{2}\underbrace{\left[1+\frac{1}{\abs{1-\inner{\synthframeelrev_{m}}{\analframeelrev_{m}}}^{2} }\sum_{k\ne m}\abs{\inner{\synthframeelrev_{m}}{\analframeelrev_{k}}}^{2}\right]}_C\\
&=C\sum_{k\ne m}\abs{\inner{\fun}{\synthframeelrev_{k}}}^{2}
\end{align*}
or equivalently
\begin{equation*}
\frac{1}{C}\sum_{k\in\frameset}\abs{\inner{\fun}{\synthframeelrev_{k}}}^{2} \le \sum_{k\ne m} \abs{\inner{\fun}{\synthframeelrev_{k}}}^{2}.
\end{equation*}
With \fref{eq:framedef} it follows that
\begin{equation}
	\label{eq:somethingisframe}
\frac{\frameA}{C}\vecnorm{\fun}^{2} \le \frac{1}{C}\sum_{k\in\frameset} \abs{\inner{\fun}{\synthframeelrev_{k}}}^{2} \le
\sum_{k\ne m}\abs{\inner{\fun}{\synthframeelrev_{k}}}^{2} \le \sum_{k\in\frameset} \abs{\inner{\fun}{\synthframeelrev_{k}}}^{2} \le
\frameB\vecnorm{\fun}^{2},
\end{equation}
where $\frameA$ and $\frameB$ are the frame bounds of the frame $\{\synthframeelrev_{k}\}_{k\in\frameset}$.
Note that (trivially) $C>0$; moreover~$C<\infty$ since $\inner{\synthframeelrev_{m}}{\analframeelrev_{m}}\ne 1$ and  $\sum_{k\ne m}\abs{\inner{\synthframeelrev_{m}}{\analframeelrev_{k}}}^{2}<\infty$ as a consequence of $\setframerevmm$ being a frame for~\hilspace. This implies that~$\frameA/C>0$, and, therefore, \fref{eq:somethingisframe} shows 
 that $\{\synthframeelrev_{k}\}_{k\ne m}$ is a frame with frame
bounds~$\frameA/C$ and~$\frameB.$

To see that, conversely, exactness of $\setdualframerevmm$  implies that $\inner{\synthframeelrev_{m}}{\analframeelrev_{m}}=1$ for all $m\in\frameset$, we
suppose that $\setdualframerevmm$ is exact and $\inner{\synthframeelrev_{m}}{\analframeelrev_{m}} \ne 1$ for at least one $m\in\frameset$. But the condition $\inner{\synthframeelrev_{m}}{\analframeelrev_{m}}
\ne 1$ for at least one $m\in\frameset$
implies that $\setdualframerevmm$ is inexact, which results in a contradiction.
It remains to show that $\setdualframerevmm$ inexact implies $\inner{\synthframeelrev_{m}}{\analframeelrev_{m}} \ne 1$ for at least one $m\in\frameset$.
Suppose that $\setdualframerevmm$ is inexact and $\inner{\synthframeelrev_{m}}{\analframeelrev_{m}}=1$ for all $m\in\frameset$. But the condition $\inner{\synthframeelrev_{m}}{\analframeelrev_{m}}=1$ for all $m\in\frameset$
implies that $\setdualframerevmm$ is exact, which again results in a contradiction.
\end{proof}

Now we are ready to state the two main results of this section. The first result generalizes the biorthonormality relation \fref{eq:biorthog} to the  infinite-dimensional setting.
\begin{cor}[{\citepar{heil89-12}}]
\label{cor:exactbiorth}
Let $\setdualframerevmm$ be a frame for the Hilbert space \hilspace. If $\setdualframerevmm$ is  exact, then
$\setdualframerevmm$ and its canonical dual $\setframerevmm$ are biorthonormal, i.e.,
\begin{equation*}
\inner{\synthframeelrev_{m}}{\analframeelrev_{k}}=\begin{cases}
1,\mbox{ if } k=m\\
0,\mbox{ if } k\ne m. \end{cases}
\end{equation*}
Conversely, if $\setdualframerevmm$ and  $\setframerevmm$ are biorthonormal, then $\setdualframerevmm$ is exact.
\end{cor}
\begin{proof}
If $\setdualframerevmm$ is exact, then biorthonormality follows by noting that \fref{thm:exactframecond} implies
$\inner{\synthframeelrev_{m}}{\analframeelrev_{m}}=1$ for
all $m\in\frameset$, and \fref{lem:innerprodexpr} implies $\sum_{k\ne m}\abs{\inner{\synthframeelrev_{m}}{\analframeelrev_{k}}}^{2}=0$  for
all $m\in\frameset$ and thus $\inner{\synthframeelrev_{m}}{\analframeelrev_{k}}=0$ for all
$ k\ne m $.
To  show that, conversely,  biorthonormality of $\setdualframerevmm$ and  $\setframerevmm$ implies that the frame $\setdualframerevmm$ is exact, we simply note that $\inner{\synthframeelrev_{m}}{\analframeelrev_{m}}=1$ for all $m\in\frameset$, by \fref{thm:exactframecond}, implies that $\setdualframerevmm$ is exact.
\end{proof}

The second main result in this section states that the expansion into an exact frame is unique and, therefore, the concept of an exact frame generalizes that of a basis to infinite-dimensional spaces. 
\begin{thm}[\citepar{heil89-12}]
If $\setdualframerevmm$ is an exact frame for the Hilbert space \hilspace and $\fun=\sum_{k\in\frameset} \coef_k\synthframeelrev_k$ with~$\fun\in\hilspace$, then the coefficients $\{\coef_k\}_{k\in\frameset}$ are unique and are given by
\begin{equation*}
\coef_k=\inner{\fun}{\analframeelrev_k},
\end{equation*}
where $\setframerevmm$ is the canonical dual frame to $\setdualframerevmm$.
\end{thm}
\begin{proof}
	We know from~\fref{eq:signalexp2} that $\fun$ can be written as $\fun=\sum_{k\in\frameset} \inner{\fun}{\analframeelrev_k}\synthframeelrev_k$. Now assume that there is another set of coefficients $\{\coef_k\}_{k\in\frameset}$ 
	such that 
	\begin{equation}
		\label{eq:xdeccomp}
	\fun=\sum_{k\in\frameset} \coef_k\synthframeelrev_k.	
	\end{equation}
	Taking the inner product of both sides of~\fref{eq:xdeccomp} with $\analframeelrev_m$ and using the biorthonormality relation 
	\begin{equation*}
		\inner{\synthframeelrev_k}{\analframeelrev_m}=\begin{cases}1,\quad k=m\\0,\quad k\neq m\end{cases}
	\end{equation*}
	we obtain
	\begin{equation*}
		\inner{\fun}{\analframeelrev_m}=\sum_{k\in\frameset} \coef_k\inner{\synthframeelrev_k}{\analframeelrev_m}=\coef_m.
	\end{equation*}
	Thus, $\coef_m=\inner{\fun}{\analframeelrev_m}$ for all $m\in\frameset$ and the proof is completed. 
\end{proof}

\section{The Sampling Theorem}
\label{sec:shannonthm}
We now discuss one of the most important results in signal processing---the sampling theorem. 
We will then  show how the sampling theorem can be interpreted as a frame decomposition.

Consider a signal $\fun(\time)$ in the space of square-integrable functions $\sqintspace$. 
In general, we can not expect this signal to be uniquely specified by its samples $\{\fun(k\sinterval)\}_{k\in\integers}$, where $\sinterval$ is the sampling period. The sampling theorem tells us, however, that if a signal is %
strictly bandlimited, i.e., its Fourier transform  vanishes outside a certain finite interval, and if \sinterval is chosen small enough (relative to the signal's bandwidth), then the samples $\{\fun(k\sinterval)\}_{k\in\integers}$ do uniquely specify the signal and we can reconstruct $\fun(\time)$ from $\{\fun(k\sinterval)\}_{k\in\integers}$ perfectly. The process of obtaining the samples $\{\fun(k\sinterval)\}_{k\in\integers}$ from the continuous-time signal $\fun(\time)$ is called \adac conversion\footnote{Strictly speaking \adac conversion also involves quantization of the samples.}; the process of reconstruction of the signal $\fun(\time)$ from its samples is called \daac conversion.  We shall now formally state and prove the sampling theorem. 

Let $\ft\fun(\freq)$ denote the Fourier transform of the signal $\fun(\time)$, i.e.,
\begin{equation*}
	\ft\fun(\freq)=\int_{-\infty}^\infty \fun(\time)e^{-\iu2\pi \time\freq}d\time.
\end{equation*}
We say that $\fun(\time)$ is bandlimited to $\bandwidth\Hz$ if
$\ft\fun(\freq)=0$ for  $\abs{\freq}>\bandwidth$. Note that this implies that the total bandwidth of $\fun(\time)$, counting positive and negative frequencies, is $2\bandwidth$.  The Hilbert space of  $\sqintspace$ functions that are bandlimited to $\bandwidth\Hz$ is denoted as $\sqintspaceBL$.

Next, consider the sequence of samples $\bigl\{\fun[k]\define \fun(k\sinterval)\bigr\}_{k\in\integers}$ of the signal $\fun(\time)\in\sqintspaceBL$ and compute its \dtftac:
\begin{align}
	\ftd{\fun}(\freq)&\define\sum_{k=-\infty}^\infty \fun[k]e^{-\iu 2\pi k\freq}\nonumber\\
	&=\sum_{k=-\infty}^\infty \fun(k\sinterval) e^{-\iu 2\pi k\freq}\nonumber\\
	&=\frac{1}{T}\sum_{k=-\infty}^\infty \ft\fun\lefto(\frac{\freq+k}{\sinterval}\right),
	\label{eq:dftsampled}
\end{align}
where in the last step we used the Poisson summation formula\footnote{\label{ft:poisson} Let $\fun(\time)\in\hilfunspace$ with Fourier transform $\ft\fun(\freq)=\int_{-\infty}^{\infty}\fun(\time) e^{-\iu  2\pi  \time \freq} d\time$. The Poisson summation formula states that $\sum_{k=-\infty}^\infty \fun(k)=\sum_{k=-\infty}^\infty \ft\fun(k)$.}~\cite[Cor. 2.6]{stein71}. 

\begin{figure}
	\centering
	\subfigure[]{\label{fig:initsignal}\includegraphics[]{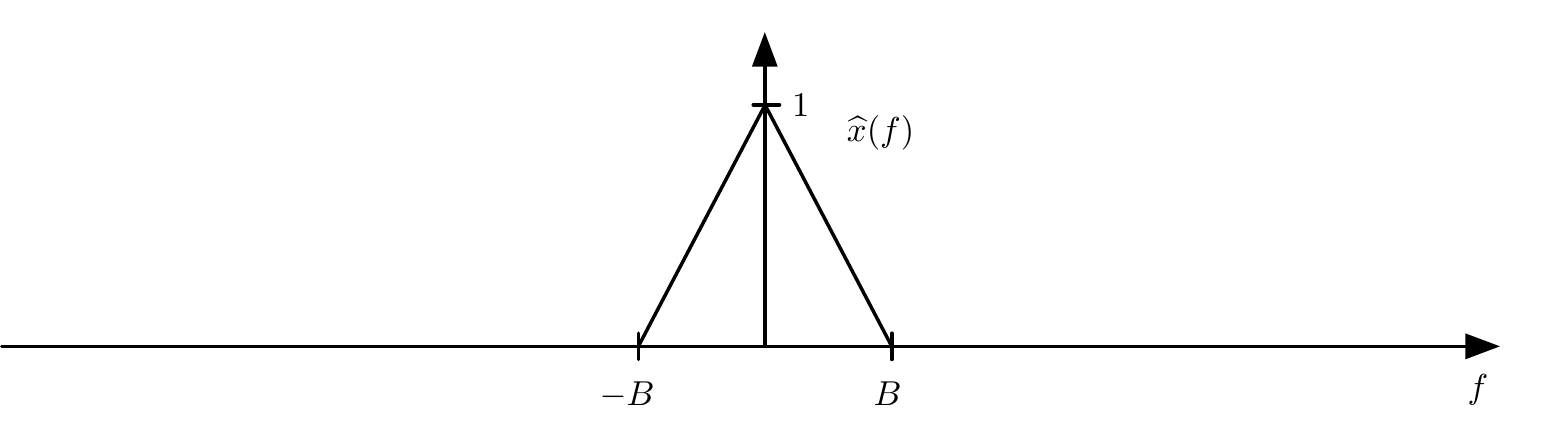}}\\
	\subfigure[]{\label{fig:nooverlap}\includegraphics[]{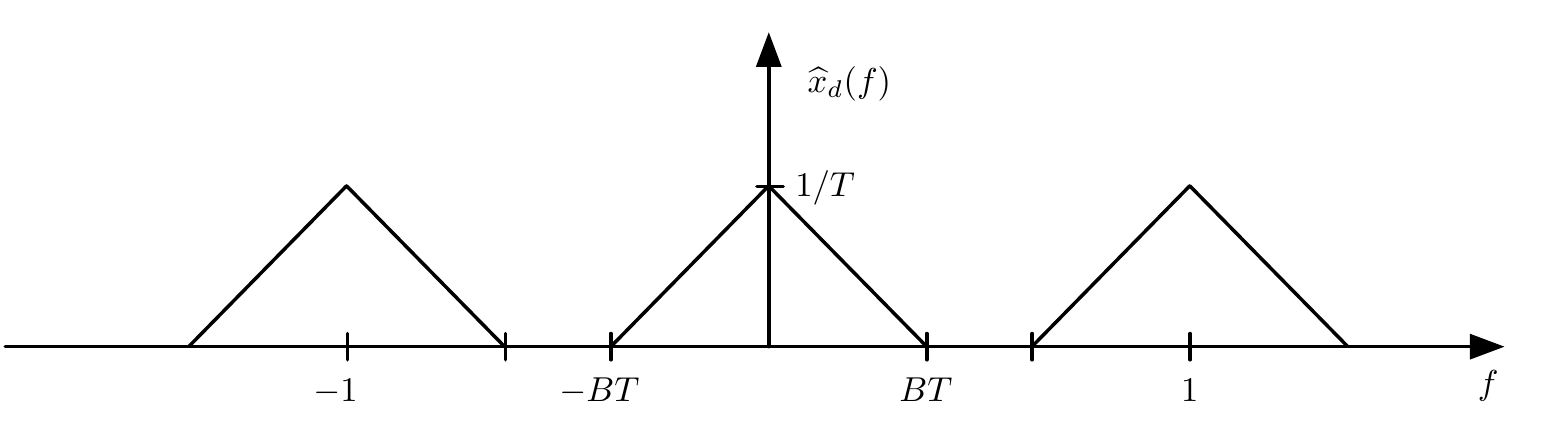}}\\
	\subfigure[]{\label{fig:overlap}\includegraphics[]{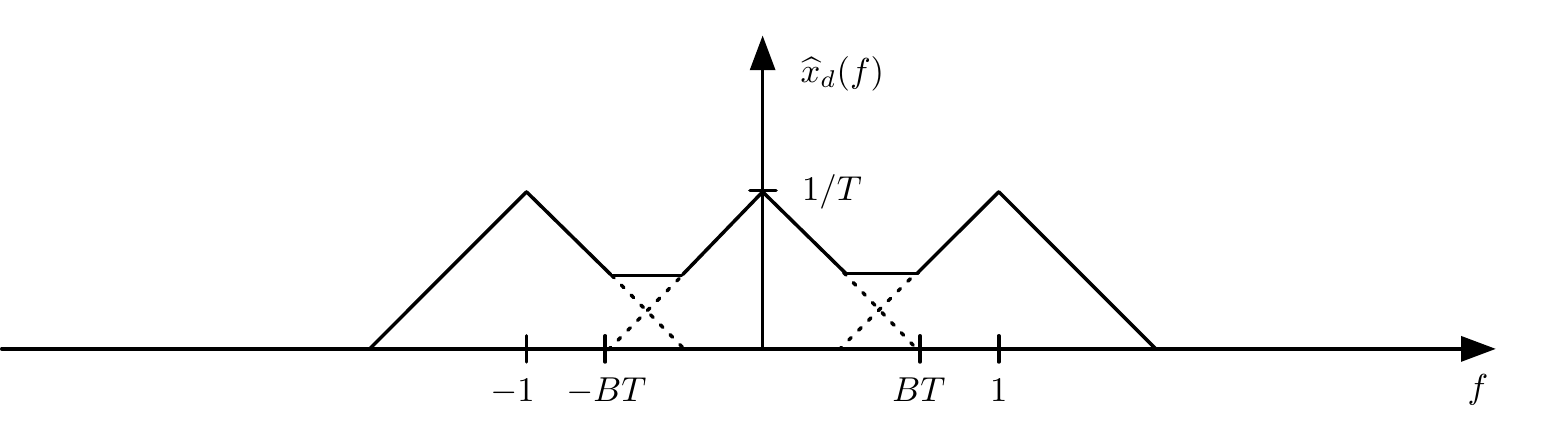}}
	\caption{Sampling of a signal that is band-limited to $\bandwidth\Hz$: (a) spectrum of the original signal; (b) spectrum of the sampled signal for $1/\sinterval>2\bandwidth$; (c) spectrum of the sampled signal for $1/\sinterval<2\bandwidth$, where aliasing occurs.}
	\label{fig:smapledspectrum}
\end{figure}
We can see that $\ftd{\fun}(\freq)$ is simply a periodized version of $\ft\fun(\freq)$. Now, it follows that for  $1/\sinterval\ge 2\bandwidth$ there is no overlap between the shifted replica of $\ft\fun(\freq/\sinterval)$, whereas for $1/\sinterval<2\bandwidth$, we do get the different shifted versions to overlap (see~\Figref{fig:smapledspectrum}). 
We can therefore conclude that for $1/\sinterval\ge 2\bandwidth$, $\ft\fun(\freq)$ can be recovered exactly from $\ftd{\fun}(\freq)$ by means of applying an ideal lowpass filter with gain $\sinterval$ and cutoff frequency $\bandwidth \sinterval$ to $\ftd{\fun}(\freq)$.
Specifically, we find that
\begin{equation}
	\label{eq:scaledspec}
	\ft\fun(\freq/\sinterval)=\ftd{\fun}(\freq)\, \sinterval\,\LPfilterf(\freq)
\end{equation}
with
\begin{equation}
	\label{eq:LPfilterf}
	\LPfilterf(\freq)=\begin{cases}1, \quad \abs{\freq}\le \bandwidth \sinterval \\0,\quad \text{otherwise}.\end{cases}
\end{equation}
From~\eqref{eq:scaledspec}, using \fref{eq:dftsampled}, we immediately see that we can recover the Fourier transform of $\fun(\time)$ from the sequence of samples $\{\fun[k]\}_{k\in\integers}$ according to 
\begin{equation}
	\ft\fun(\freq)
	=\sinterval \,\LPfilterf(\freq\sinterval)\sum_{k=-\infty}^\infty \fun[k]e^{-\iu 2\pi k\freq\sinterval}.
\end{equation}
We can therefore recover $\fun(\time)$ as follows:
\begin{align}
	\fun(\time)&=\int_{-\infty}^\infty \ft\fun(\freq) e^{\iu 2\pi \time\freq} d\freq\nonumber\\
	&=\int_{-\infty}^\infty \sinterval \LPfilterf(\freq\sinterval) \sum_{k=-\infty}^\infty  \fun[k]e^{-\iu 2\pi k\freq\sinterval} e^{\iu 2\pi \freq\time} d\freq\nonumber\\
	&=\sum_{k=-\infty}^\infty \fun[k]  \int_{-\infty}^\infty \LPfilterf(\freq\sinterval) e^{\iu 2\pi \freq\sinterval\left(\time/\sinterval-k\right) } d(\freq\sinterval)\nonumber\\
	\label{eq:reconstrLPfiltert}
	&=\sum_{k=-\infty}^\infty \fun[k]  \LPfiltert\lefto(\frac{\time}{\sinterval}-k\right)\\
	&=2\bandwidth \sinterval\sum_{k=-\infty}^\infty \fun[k] \sinc(2\bandwidth (\time-k\sinterval)),\nonumber
\end{align}
where $\LPfiltert(\time)$ is the inverse Fourier transform of $\LPfilterf(\freq)$, i.e,
\begin{equation*}
	\LPfiltert(\time)=\int_{-\infty}^\infty \LPfilterf(\freq)e^{\iu2\pi \time\freq}d\freq,
\end{equation*}
 and
\begin{equation*}
	\sinc(x)\define \frac{\sin(\pi x)}{\pi x}.
\end{equation*}
Summarizing our findings, we obtain the following theorem.
\begin{thm}[Sampling theorem~{\cite[Sec. 7.2]{oppenheim96}}]
	Let $\fun(\time)\in\sqintspaceBL$. Then 
	$\fun(\time)$ is uniquely specified by its samples $\fun(k\sinterval),\, k\in\integers,$ if $1/\sinterval\ge 2\bandwidth$.
	Specifically, we can reconstruct $\fun(\time)$ from $\fun(k\sinterval),\, k\in\integers,$ according to
	\begin{equation}
		\label{eq:samplingtheorem}
		\fun(\time)=2\bandwidth \sinterval\sum_{k=-\infty}^\infty \fun(k\sinterval) \sinc(2\bandwidth (\time-k\sinterval)).
	\end{equation}
\end{thm}

\subsection{Sampling Theorem as a Frame Expansion}
\label{sec:samplingframes}
We shall next show how the representation~\fref{eq:samplingtheorem} can be interpreted as a frame expansion.
The samples~$\fun(k\sinterval)$ can be written as the inner product of the signal $\fun(\time)$ with the functions 
\begin{equation}
	\label{eq:sincframe}
	\analframeel_k(\time)=2\bandwidth \sinc(2\bandwidth(\time-k\sinterval)),\quad k\in\integers.
\end{equation} 
Indeed, using the fact that the signal $\fun(\time)$ is band-limited to $\bandwidth\Hz$, we get
\begin{equation*}
	\fun(k\sinterval)=\int_{-\bandwidth}^\bandwidth \ft\fun(\freq) e^{\iu2\pi k \freq \sinterval } d\freq=\inner{\ft\fun}{\ft\analframeel_k},
\end{equation*}
where
\begin{equation*}
	\ft\analframeel_k(\freq)=\begin{cases}e^{-\iu2\pi k \freq \sinterval },\, \abs{\freq}\le\bandwidth \\0,\ \text{otherwise}\end{cases}
\end{equation*}
is the Fourier transform of $\analframeel_k(\time)$. We can thus rewrite \fref{eq:samplingtheorem} as
\begin{equation*}
	\fun(\time)= \sinterval \sum_{k=-\infty}^\infty \inner{\fun}{\analframeel_k}  \analframeel_k(\time).
\end{equation*}
Therefore, the interpolation of an analog signal from its samples $\{\fun(k\sinterval)\}_{k\in\integers}$ can be interpreted as the reconstruction of $\fun(\time)$ from its expansion coefficients $\fun(k\sinterval)=\inner{\fun}{\analframeel_k}$ in the
function set $\{\analframeel_k(\time)\}_{k\in\integers}$. 
We shall next prove that $\{\analframeel_k(\time)\}_{k\in\integers}$ is a frame for the space $\sqintspaceBL$. Simply note that for  $\fun(\time)\in\sqintspaceBL$, we have
\begin{equation*}
\vecnorm{\fun}^2=\inner{\fun}{\fun}=\inner{\sinterval \sum_{k=-\infty}^\infty \inner{\fun}{\analframeel_k}  \analframeel_k(\time)}{\fun}=\sinterval\sum_{k=-\infty}^\infty \abs{\inner{\fun}{\analframeel_k}}^2
\end{equation*}
and therefore
\begin{equation*}
	\frac{1}{\sinterval}\vecnorm{\fun}^2=\sum_{k=-\infty}^\infty \abs{\inner{\fun}{\analframeel_k}}^2.
\end{equation*}
This shows that $\{\analframeel_k(\time)\}_{k\in\integers}$ is, in fact, a tight frame for $\sqintspaceBL$ with frame bound $\frameA=1/\sinterval.$ 
We emphasize that the frame is tight irrespective of the sampling rate (of course, as long as $1/\sinterval>2\bandwidth$).

The analysis operator corresponding to this frame is given by $\analop:\sqintspaceBL\to\hilseqspace$ as
\begin{equation}
	\label{eq:analopsampling}
	\analop:\fun\to\{\inner{\fun}{\analframeel_k}\}_{k\in\integers},
\end{equation}
i.e., \analop maps the signal $\fun(\time)$ to the sequence of samples $\{\fun(k\sinterval)\}_{k\in\integers}$.

The action of the adjoint of the analysis operator $\ad\analop:\hilseqspace\to\sqintspaceBL$ is to perform interpolation according to 
\begin{equation*}
	\ad\analop:\{\coef_k\}_{k\in\integers}\to \sum_{k=-\infty}^\infty \coef_k \analframeel_k.
\end{equation*}
The frame  operator $\frameop:\sqintspaceBL\to\sqintspaceBL$ is given by $\frameop=\ad\analop\analop$ and acts as follows
\begin{equation*}
	\frameop:\fun(\time)\to \sum_{k=-\infty}^\infty \inner{\fun}{\analframeel_k} \analframeel_k(\time).
\end{equation*}
Since $\{\analframeel_k(\time)\}_{k\in\integers}$ is a tight frame for $\sqintspaceBL$ with frame bound $\frameA=1/\sinterval$, as already shown, it follows that $\frameop=(1/\sinterval)\iop_{\sqintspaceBL}$.

The canonical dual frame can be computed easily by applying the inverse of the frame operator to the frame functions $\{\analframeel_k(\time)\}_{k\in\integers}$ according to
\begin{equation*}
	\synthframeel_k(\time)=\frameop^{-1}	\analframeel_k(\time)=\sinterval\iop_{\sqintspaceBL}	\analframeel_k(\time)=\sinterval\analframeel_k(\time),\quad k\in\integers.
\end{equation*}

Recall that exact frames have 
a minimality property in the following sense: If we remove anyone element from an exact frame, the resulting set will be incomplete. In the case of sampling, we have an analogous situation: In the proof of the sampling theorem we saw that if we sample at a rate smaller than the \emph{critical sampling rate} $1/\sinterval=2\bandwidth$, we cannot recover the signal $\fun(\time)$ from its samples~$\{\fun(k\sinterval)\}_{k\in\integers}$. In other words, the set $\{\analframeel_k(\time)\}_{k\in\integers}$
in \fref{eq:sincframe} is \emph{not} complete for $\sqintspaceBL$ when $1/\sinterval<2\bandwidth$. This suggests that critical sampling $1/\sinterval=2\bandwidth$ could implement an exact frame decomposition. We show now that this is, indeed, the case.
Simply note that
\begin{equation*}
\inner{\analframeel_k}{\synthframeel_k}=\sinterval \inner{\analframeel_k}{\analframeel_k}=\sinterval\vecnorm{\analframeel_k}^2=\sinterval\vecnorm{\ft\analframeel_k}^2=2\bandwidth\sinterval,\quad \text{for all}\ k\in\integers.
\end{equation*}
For critical sampling $2\bandwidth\sinterval=1$ and, hence, $\inner{\analframeel_k}{\synthframeel_k}=1$, for all $k\in\integers$. \fref{thm:exactframecond} therefore allows us to conclude that $\{\analframeel_k(\time)\}_{k\in\integers}$ is an exact frame for $\sqintspaceBL$.

Next, we show that $\{\analframeel_k(\time)\}_{k\in\integers}$ is not only an exact frame, but, when properly normalized, even an \onbac for $\sqintspaceBL$, a fact well-known in sampling theory. To this end, we first renormalize the frame functions $\analframeel_k(\time)$  according to
\begin{equation*}
\analframeel'_k(\time)=\sqrt{\sinterval}\analframeel_k(\time)
\end{equation*}
so that
\begin{equation*}
\fun(\time)=\sum_{k=-\infty}^\infty \inner{\fun}{\analframeel'_k}\analframeel'_k(\time).
\end{equation*}
We see that $\{\analframeel'_k(\time)\}_{k\in\integers}$ is a tight frame for $\sqintspaceBL$ with $\frameA=1$. 
Moreover, we have
\begin{equation*}
	\vecnorm{\analframeel'_k}^2=\sinterval\vecnorm{\analframeel_k}^2=2\bandwidth\sinterval.
\end{equation*}
Thus, in the case of critical sampling,  $\vecnorm{\analframeel'_k}^2=1$, for all $k\in\integers$, and \fref{thm:tightframebasis} allows us to conclude that $\{\analframeel'_k(\time)\}_{k\in\integers}$ is an \onbac for $\sqintspaceBL$.

In contrast to  exact frames, inexact frames are redundant, in the sense that there is at least one element  that can be removed with the resulting set still being complete.
The situation is similar in the \emph{oversampled} case, i.e., when the sampling rate satisfies $1/\sinterval>2\bandwidth$. In this case, we collect more samples than actually needed for perfect reconstruction of $\fun(\time)$ from its samples. This suggests that $\{\analframeel_k(\time)\}_{k\in\integers}$  could be  an inexact frame for $\sqintspaceBL$ in the oversampled case. Indeed, according to \fref{thm:exactframecond} the condition 
\begin{equation}
	\label{eq:smone}
\inner{\analframeel_m}{\synthframeel_m}=2\bandwidth\sinterval<1,\quad \text{for all}\ m\in\integers,
\end{equation}
implies that the frame $\{\analframeel_k(\time)\}_{k\in\integers}$ is inexact for $1/\sinterval>2\bandwidth$. In fact, as can be seen from the proof of \fref{thm:exactframecond}, \fref{eq:smone} guarantees even more: for \emph{every} $m\in\integers$, the set $\{\analframeel_k(\time)\}_{k\ne m}$ is complete for~$\sqintspaceBL$. Hence,
the removal of \emph{any} sample $\fun(m\sinterval)$ from the set of samples $\{\fun(k\sinterval)\}_{k\in\integers}$ still leaves us with a frame decomposition so that $\fun(t)$ can, in theory, be recovered from the samples $\{\fun(k\sinterval)\}_{k\ne m}$. The resulting frame $\{\analframeel_k(\time)\}_{k\ne m}$ will, however, no longer be tight, which makes the  computation of the canonical dual frame  complicated, in general.

\subsection{Design Freedom in Oversampled A/D Conversion}
\label{sec:designfreedom}
In the critically sampled case, $1/\sinterval= 2\bandwidth$, the ideal lowpass filter of bandwidth $\bandwidth \sinterval$ with the transfer function specified in~\fref{eq:LPfilterf} is the only filter that provides perfect reconstruction of the spectrum~$\ft\fun(\freq)$ of~$\fun(\time)$ according to~\fref{eq:scaledspec} (see~\Figref{fig:filtersharp}). In the oversampled case, there is, in general, an infinite number of reconstruction filters that provide perfect reconstruction. The only requirement the reconstruction filter has to satisfy is that its transfer function be constant within the frequency range $-\bandwidth \sinterval\le\freq\le\bandwidth \sinterval$ 
(see~\Figref{fig:filternonsharpe}). Therefore, in the oversampled case one has more freedom in designing the reconstruction filter.  In \adac converter practice this design freedom is exploited to design reconstruction filters with desirable filter characteristics, like, e.g., rolloff in the transfer function.  
\begin{figure}[t]
	\centering
	\includegraphics[]{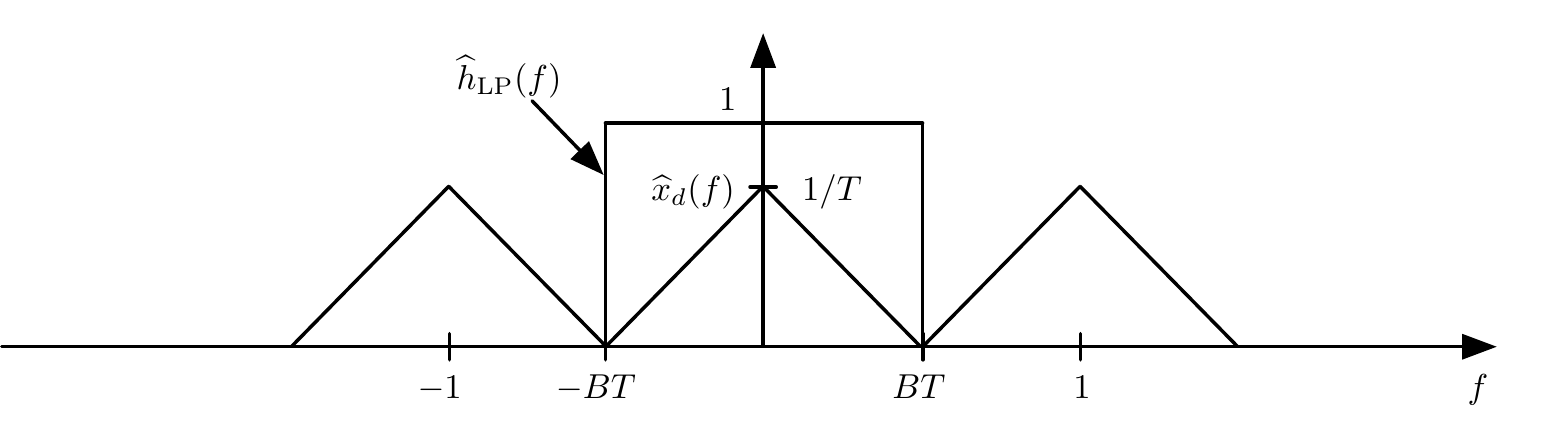}
	\caption{Reconstruction filter in the critically sampled case.}
	\label{fig:filtersharp}
\end{figure}

Specifically, repeating the steps leading from \fref{eq:scaledspec} to \fref{eq:reconstrLPfiltert}, we see that
\begin{equation}
	\label{eq:generalreconst}
	\fun(\time)
	=\sum_{k=-\infty}^\infty \fun[k]  \arbfiltert\lefto(\frac{\time}{\sinterval}-k\right),
\end{equation}
where the Fourier transform of $\arbfiltert(\time)$ is given by
\begin{equation}
	\label{eq:arbfilter}
\ft \arbfiltert(\freq)=\begin{cases} 1,\quad &\abs{\freq}\leq \bandwidth\sinterval \\ \arb(\freq), &\bandwidth\sinterval<\abs{\freq}\le \frac{1}{2} \\0,&\abs{f}>\frac{1}{2}\end{cases}.
\end{equation}
Here and in what follows $\arb(\cdot)$ denotes an arbitrary bounded function.
\begin{figure}[t]
	\centering
	\includegraphics[]{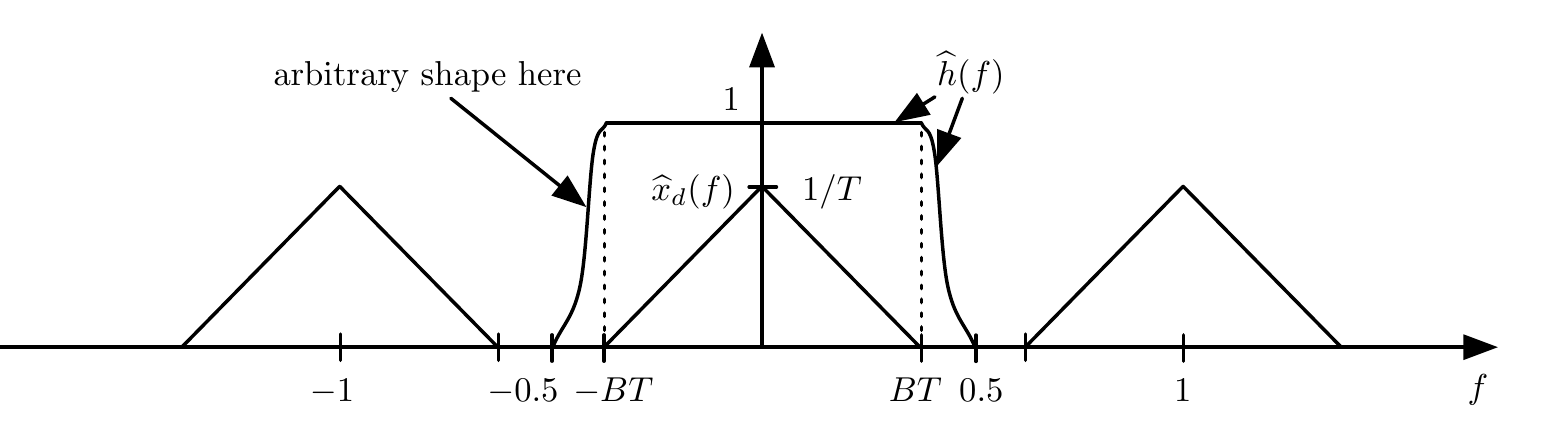}
	\caption{Freedom in the design of the reconstruction filter.}
	\label{fig:filternonsharpe}
\end{figure}%
In other words, every set $\{\arbfiltert(\time/\sinterval-k)\}_{k\in\integers}$ with the Fourier transform of $\arbfiltert(\time)$ satisfying~\fref{eq:arbfilter} is a valid dual frame for the frame 
$\{\analframeel_k(\time)=2\bandwidth \sinc(2\bandwidth(t-k\sinterval))\}_{k\in\integers}$. Obviously, there are infinitely many dual frames in the oversampled case.

\begin{figure}[t]
		\centering
		\includegraphics[]{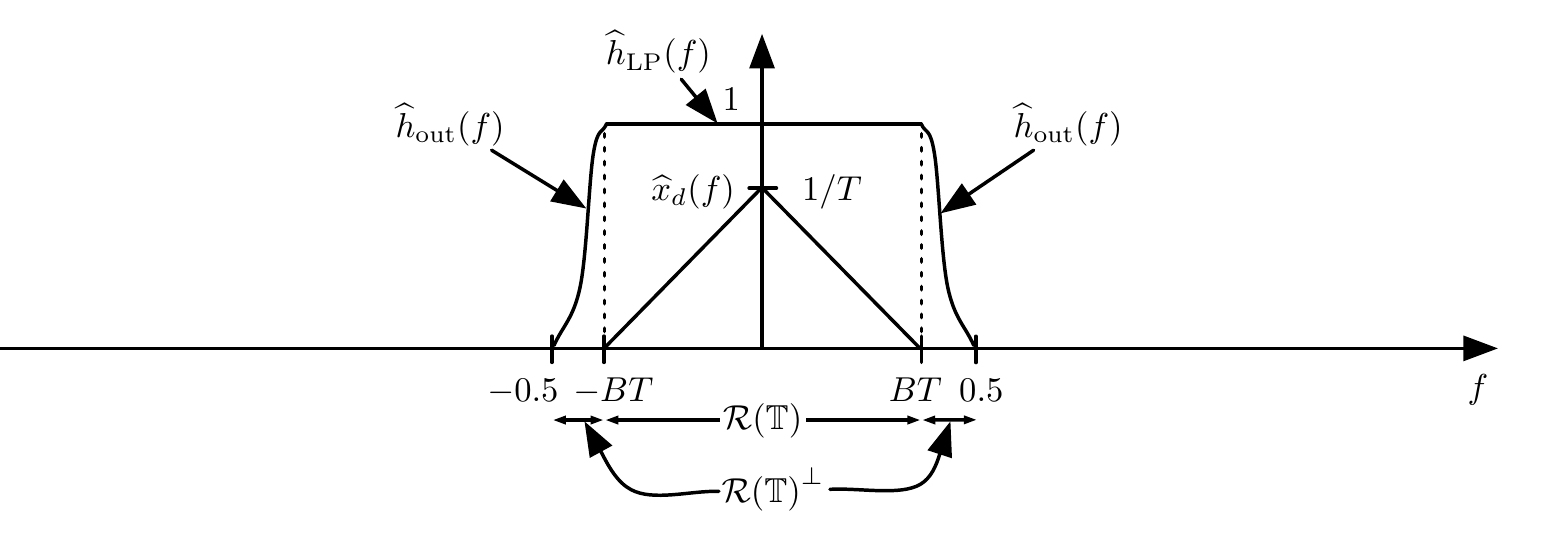}
		\caption{The reconstruction filter as a parametrized left-inverse of the analysis operator.}
		\label{fig:leftinvsamplingdecomp}
\end{figure}
We next show how the freedom in  the design of the reconstruction filter with transfer function specified in~\fref{eq:arbfilter} can be interpreted in terms of the freedom in choosing the left-inverse $\opL$ of the analysis operator $\analop$ as discussed in \fref{sec:signalexpansions}. Recall the parametrization~\fref{eq:generalliopframes1} of all left-inverses of the operator $\analop$:
\begin{equation}
	\label{eq:leftinvsampling}
		\opL=\ad\analdualop \pop +\opM(\iop_{\hilseqspace}-\pop),
\end{equation}
where $\opM:\hilseqspace\to\hilspace$ is an arbitrary bounded linear operator and $\pop:\sqsumspace\to\sqsumspace$ is the orthogonal projection operator onto the range space of $\analop$.
In~\fref{eq:dftsampled} we saw that  the \dtftac\footnote{The \dtftac is a periodic function with period one. From here on, we consider  the \dtftac as a function supported on its fundamental period $[-1/2,1/2]$.} of the sequence $\bigl\{\fun[k]= \fun(k\sinterval)\bigr\}_{k\in\integers}$ is compactly supported on the frequency interval $[-\bandwidth\sinterval,\bandwidth\sinterval]$ (see~\Figref{fig:leftinvsamplingdecomp}). In other words, the range space of the analysis operator \analop defined in \fref{eq:analopsampling} is the space of $\hilseqspace$-sequences with \dtftac supported on the interval  
$[-\bandwidth\sinterval,\bandwidth\sinterval]$  (see~\Figref{fig:leftinvsamplingdecomp}). 
It is left as an exercise to the reader to verify (see~\fref{xca:ocrange}), using Parseval's theorem,\footnote{\label{ft:parseval} Let $\{a_k\}_{k\in\integers}, \{b_k\}_{k\in\integers}\in\hilseqspace$ with \dtftac  $\ft a(\freq)=\sum_{k=-\infty}^\infty a_k e^{-\iu  2\pi  k \freq}$ and $\ft b(\freq)=\sum_{k=-\infty}^\infty b_k e^{-\iu  2\pi  k \freq} $, respectively. Parseval's theorem states that $\sum_{k=-\infty}^\infty a_k \conj{b}_k=\int_{-1/2}^{1/2} \ft a(\freq)\conj{\ft b}(\freq) d\freq$. In particular, $\sum_{k=-\infty}^\infty \abs{a_k}^2=\int_{-1/2}^{1/2} \abs{\ft a(\freq)}^2 d\freq$.} that the orthogonal complement of the range space of \analop
is the space of  $\hilseqspace$-sequences with \dtftac supported on the set $[-1/2,-\bandwidth\sinterval]\cup[\bandwidth\sinterval,1/2]$ (see~\Figref{fig:leftinvsamplingdecomp}). 
Thus, in the case of oversampled \adac conversion, the operator $\pop:\hilseqspace\to\hilseqspace$ is the orthogonal projection operator onto the subspace of $\hilseqspace$-sequences with \dtftac supported on the interval  
$[-\bandwidth\sinterval,\bandwidth\sinterval]$; the  operator $(\iop_{\hilseqspace}-\pop):\hilseqspace\to\hilseqspace$ is the orthogonal projection operator onto the subspace of  $\hilseqspace$-sequences with \dtftac supported on the set $[-1/2,-\bandwidth\sinterval]\cup[\bandwidth\sinterval,1/2]$. 

To see the parallels between~\fref{eq:generalreconst} and~\fref{eq:leftinvsampling}, let us decompose~$\arbfiltert\lefto(\time\right)$ as follows (see~\Figref{fig:leftinvsamplingdecomp})
\begin{equation}
	\label{eq:filterdecomp}
	\arbfiltert\lefto(\time\right) 
	=\LPfiltert\lefto(\time\right) + \outfiltert\lefto(\time\right),
\end{equation}
where the Fourier transform of $\LPfiltert(\time)$ is given by~\fref{eq:LPfilterf} and  the Fourier transform of $\outfiltert(\time)$ is
\begin{equation}
	\label{eq:zeropartfilter}
\outfilterf(\freq)=\begin{cases}  \arb(\freq), &\bandwidth\sinterval\le\abs{\freq}\le \frac{1}{2} \\0,\ \text{otherwise.}\end{cases}
\end{equation}
Now it is clear, and it is left to the reader to verify formally (see~\fref{xca:opA}), that the operator
$\opA:\hilseqspace\to\sqintspaceBL$ defined as  
\begin{equation}
	\label{eq:opAdfn}
	\opA:\{\coef_k\}_{k\in\integers}\to \sum_{k=-\infty}^\infty \coef_k \LPfiltert\lefto(\frac{\time}{\sinterval}-k\right)
\end{equation}
acts by first projecting the sequence $\{\coef_k\}_{k\in\integers}$ onto the subspace of  $\hilseqspace$-sequences with \dtftac supported on the interval  
$[-\bandwidth\sinterval,\bandwidth\sinterval]$ and then performs interpolation using the canonical dual frame elements $\synthframeel_k(\time)=\LPfiltert\lefto(\time/\sinterval-k\right)$. In other words $\opA=\ad\analdualop \pop$. 
Similarly, it is left to the reader to verify formally (see~\fref{xca:opA}), that the operator
$\opB:\hilseqspace\to\sqintspaceBL$ defined as
\begin{equation}
	\label{eq:opBdfn}
	\opB:\{\coef_k\}_{k\in\integers}\to \sum_{k=-\infty}^\infty \coef_k \outfiltert\lefto(\frac{\time}{\sinterval}-k\right)
\end{equation}
can be written as $\opB=\opM(\iop_{\hilseqspace}-\pop)$. Here, $(\iop_{\hilseqspace}-\pop):\sqsumspace\to \sqsumspace$ is the projection operator onto the subspace of  $\hilseqspace$-sequences with \dtftac supported on the set $[-1/2,-\bandwidth\sinterval]\cup[\bandwidth\sinterval,1/2]$;
the operator $\opM:\hilseqspace\to\hilfunspace$ is defined as
\begin{equation}
	\label{eq:opBdfn1}
	\opM:\{\coef_k\}_{k\in\integers}\to \sum_{k=-\infty}^\infty \coef_k h_M\lefto(\frac{\time}{\sinterval}-k\right)
\end{equation}
with the Fourier transform of $h_M(\time)$ given by
\begin{equation}
	\label{eq:zeropartfilter1}
\ft h_M(\freq)=\begin{cases}  \arb_2(\freq), &-\frac{1}{2}\le\abs{\freq}\le \frac{1}{2} \\0, &\text{otherwise},\end{cases}
\end{equation}
where $\arb_2(\freq)$ is an arbitrary bounded function that equals $\arb(\freq)$ for $\bandwidth\sinterval\le\abs{\freq}\le \frac{1}{2}$.
To summarize, we note that the operator $\opB$ corresponds to the second term on the right-hand side of~\fref{eq:leftinvsampling}. 

We can therefore write the decomposition~\fref{eq:generalreconst} as
\begin{align*}
	\fun(\time)
	&=\sum_{k=-\infty}^\infty \fun[k]  \arbfiltert\lefto(\frac{\time}{\sinterval}-k\right)\\
	&=\underbrace{\sum_{k=-\infty}^\infty \fun[k]  \LPfiltert\lefto(\frac{\time}{\sinterval}-k\right)}_{\ad\analdualop \pop\analop\fun(\time)} + \underbrace{\sum_{k=-\infty}^\infty \fun[k]\outfiltert\lefto(\frac{\time}{\sinterval}-k\right)}_{\opM(\iop_{\hilseqspace}-\pop)\analop\fun(\time)}\\
	&=\opL\analop\fun(\time).
\end{align*}

\subsection{Noise Reduction in Oversampled A/D Conversion}
\label{sec:noisead}
Consider again the bandlimited signal $\fun(\time)\in\sqintspaceBL$. Assume, as before, that the signal is sampled at a rate $1/\sinterval\ge 2\bandwidth$. 
Now assume that the corresponding samples $\fun[k]= \fun(k\sinterval),\, k\in\integers,$ are subject to noise, i.e., we observe
\begin{equation*}
	\fun'[k]= \fun[k]+\qnoise[k],\ k\in\integers,
\end{equation*} 
where the $\qnoise[k]$ are independent identically distributed zero-mean random variables, with variance $\Exop{\abs{\qnoise[k]}^2}=\qnoisevar$.
Assume that reconstruction is performed from the noisy samples $\fun'[k],\ k\in\integers,$ using the ideal lowpass filter with transfer function $\LPfilterf(\freq)$ of bandwidth $\bandwidth\sinterval$ specified in~\fref{eq:LPfilterf}, i.e., we reconstruct using the canonical dual frame according to
\begin{equation*}
	\fun'(\time)
	=\sum_{k=-\infty}^\infty \fun'[k]  \LPfiltert\lefto(\frac{\time}{\sinterval}-k\right).
\end{equation*}
Obviously, the presence of noise precludes perfect reconstruction. It is, however, interesting to assess the impact of oversampling on the variance of the reconstruction error defined as 
\begin{equation}
	\label{eq:noisevar}
	\mseo\define \Exop_\qnoise{\abs{\fun(\time)-\fun'(\time)}^2},
\end{equation}
where the expectation is with respect to the random variables $\qnoise[k],\ k\in\integers,$ and the right-hand side of \fref{eq:noisevar} does not depend on $\time$, as we shall see below.
 If we decompose $\fun(\time)$ as in~\fref{eq:reconstrLPfiltert}, 
we  see that
\begin{align}
	\label{eq:mseo}
	\mseo&= \Exop_\qnoise{\abs{\fun(\time)-\fun'(\time)}^2} \\
	&= \Exop_\qnoise{\abs{\sum_{k=-\infty}^\infty \qnoise[k]  \LPfiltert\lefto(\frac{\time}{\sinterval}-k\right)}^2} \nonumber\\
		&= \sum_{k=-\infty}^\infty \sum_{k'=-\infty}^\infty\Exop_\qnoise\{\qnoise[k]\conj\qnoise[k']\} \,\LPfiltert\lefto(\frac{\time}{\sinterval}-k\right)\conj{\LPfiltert}\lefto(\frac{\time}{\sinterval}-k'\right) \nonumber\\
		\label{eq:mseolast1}	&= \qnoisevar \sum_{k=-\infty}^\infty \abs{\LPfiltert\lefto(\frac{\time}{\sinterval}-k\right)}^2.
\end{align}
Next applying the Poisson summation formula (as stated in \Ftref{ft:poisson}) to the function $l(\time')\define \LPfiltert\lefto(\frac{\time}{\sinterval}-\time'\right) e^{-2\pi \iu \time'\freq}$ with Fourier transform $\ft l(\freq')= \LPfilterf(-\freq-\freq')e^{-2\pi \iu (\time/\sinterval) (\freq+\freq')}$,  we have
\begin{equation}
	\label{eq:psform}
	\sum_{k=-\infty}^\infty \LPfiltert\lefto(\frac{\time}{\sinterval}-k\right) e^{-2\pi \iu k\freq}=\sum_{k=-\infty}^\infty l(k)=\sum_{k=-\infty}^\infty \ft l(k)= \sum_{k=-\infty}^\infty \LPfilterf(-\freq-k)e^{-2\pi \iu (\time/\sinterval) (\freq+k)}.
\end{equation}
Since $\LPfilterf(\freq)$ is zero outside the interval $-1/2\le \freq\le 1/2$, it follows that
\begin{equation}
	\label{eq:psform2}
	\sum_{k=-\infty}^\infty \LPfilterf(-\freq-k)e^{-2\pi \iu (\time/\sinterval) (\freq+k)}=\LPfilterf(-\freq)e^{-2\pi \iu (\time/\sinterval) \freq},\quad \text{for}\ \freq\in[-1/2,1/2].
\end{equation}
We conclude from~\fref{eq:psform} and~\fref{eq:psform2} that the \dtftac of the sequence $\bigl\{a_k\define\LPfiltert\lefto(\time/\sinterval-k\right)\bigr\}_{k\in\integers}$ is given (in the fundamental interval $\freq\in[-1/2,1/2]$) by $\LPfilterf(-\freq)e^{-2\pi \iu (\time/\sinterval) \freq}$ and hence we can apply Parseval's theorem (as stated in \Ftref{ft:parseval}) and rewrite~\fref{eq:mseolast1} according to 
\begin{align}
	\mseo&= \qnoisevar \sum_{k=-\infty}^\infty \abs{\LPfiltert\lefto(\frac{\time}{\sinterval}-k\right)}^2\nonumber\\
	&=\qnoisevar\int_{-1/2}^{1/2} \abs{\LPfilterf(-\freq)e^{-2\pi \iu (\time/\sinterval) \freq}}^2 d \freq\nonumber\\
	&=\qnoisevar\int_{-1/2}^{1/2} \abs{\LPfilterf(\freq)}^2 d \freq\nonumber\\
	&= \qnoisevar 2 \bandwidth \sinterval. \label{eq:mseolast}
\end{align}
We see that the average mean squared reconstruction error is inversely proportional to the oversampling factor $1/(2 \bandwidth \sinterval)$. Therefore, each doubling of the oversampling factor decreases the mean squared error by  3\dB.

Consider now reconstruction performed using a general filter that provides perfect reconstruction in the noiseless case. Specifically, we have
\begin{equation*}
	\fun'(\time)
	=\sum_{k=-\infty}^\infty \fun'[k]  \arbfiltert\lefto(\frac{\time}{\sinterval}-k\right),
\end{equation*}
where $\arbfiltert\lefto(\time\right)$ is given by~\fref{eq:filterdecomp}.
 In this case, the average mean squared reconstruction error
can be computed repeating the steps leading from~\fref{eq:mseo} to~\fref{eq:mseolast} and is given by
\begin{equation}
	\label{eq:mseononidz}
	\mseo=\qnoisevar\int_{-1/2}^{1/2} \abs{\ft\arbfiltert(\freq)}^2 d \freq 
\end{equation}
where $\ft\arbfiltert(\freq)$ is the Fourier transform of $\arbfiltert(\time)$ and is specified in~\fref{eq:arbfilter}. 
Using~\fref{eq:filterdecomp}, we can now decompose $\mseo$ in~\fref{eq:mseononidz} into two terms according to
\begin{equation}
	\label{eq:mseononid}
	\mseo=\qnoisevar\underbrace{\int_{-\bandwidth\sinterval}^{\bandwidth\sinterval} \abs{\LPfilterf(\freq)}^2 d \freq}_{2 \bandwidth \sinterval}+\qnoisevar\int_{\bandwidth\sinterval\le \abs{\freq}\le 1/2} \abs{\outfilterf(\freq)}^2 d \freq.
\end{equation}
We see that two components contribute to the reconstruction error. 
Comparing~\fref{eq:mseononid} to~\fref{eq:mseolast}, we conclude that the first term in~\fref{eq:mseononid} corresponds to the error due to noise in the signal-band  $ \abs{\freq}\le \bandwidth\sinterval$ picked up by the ideal lowpass filter with transfer function $\LPfilterf(\freq)$. The second term in~\fref{eq:mseononid} is due to the fact that a generalized inverse passes some of the noise in the out-of-band region $\bandwidth\sinterval\le \abs{\freq}\le 1/2$.
The amount of additional noise in the reconstructed signal is determined by the bandwidth and the shape of the reconstruction filter's transfer function in the out-of-band region. In this sense, there exists a tradeoff between noise reduction and design freedom in oversampled \adac conversion. Practically desirable (or realizable) reconstruction filters (i.e., filters with rolloff) lead to additional reconstruction error.

\section{Important Classes of Frames}
There are two important classes of structured signal expansions that have found widespread use in practical applications, namely Weyl-Heisenberg (or Gabor) expansions and affine
(or wavelet) expansions.
Weyl-Heisenberg expansions provide a decomposition into time-shifted and modulated versions of a ``window function'' $g(t)$.
Wavelet expansions realize decompositions into time-shifted and dilated versions of a mother wavelet $g(t)$. 
Thanks to the strong structural properties of  Weyl-Heisenberg and wavelet expansions, there are efficient algorithms for applying the corresponding analysis and synthesis operators. 
Weyl-Heisenberg and wavelet expansions have been successfully used in signal detection, image representation, object recognition, and wireless communications.
We shall next show that these signal expansions can be cast into the language of frame theory. For a 
detailed analysis of these classes of frames, we refer the interested reader to~\cite{daubechies92}.

\subsection{Weyl-Heisenberg Frames}
We start by defining a linear operator that realizes time-frequency shifts when applied to a given function.
\begin{dfn}
\label{dfn:weylop}
The Weyl operator $\weylop^{(T,F)}_{m,n}:\hilfunspace\to\hilfunspace$ is defined as
\begin{equation*} 
\weylop^{(T,F)}_{m,n}: \fun(\time)\to e^{\iu2\pi  nF\time}\fun(\time-mT), 
\end{equation*}
where $m,n \in\integers$, and $T>0$ and $F>0$ are fixed time and frequency shift parameters,
respectively.
\end{dfn}

Now consider some prototype (or window)  function  $\window(\time)\in\hilfunspace$. Fix the parameters~$T>0$ and~$F>0$. By shifting the window function $\window(\time)$ in time by integer multiples of~$T$ and in frequency by integer multiples of~$F$, we get a highly-structured set of functions according to
\begin{equation*}
	\window_{m,n}(\time)\define \weylop^{(T,F)}_{m,n}\window(\time)= e^{\iu2\pi  nF\time}\window(\time-mT),\quad m\in\integers,\ n\in\integers.
\end{equation*}
The set $\left\{\window_{m,n}(\time)=  e^{\iu2\pi  nF\time}\window(\time-mT)\right\}_{m\in\integers,\, n\in\integers}$ is referred to as a \emph{\whac set} and is denoted by $(\window,T,F)$. 	When the \whac set $(\window,T,F)$ is a frame for $\hilfunspace$, it is called a \emph{\whac frame} for \hilfunspace.

Whether or not a \whac set $(\window,T,F)$ is a frame for $\hilfunspace$ is, in general, difficult to answer. The answer depends on the window function $\window(\time)$ as well as on the shift parameters $T$ and $F$. Intuitively, if the parameters $T$ and $F$ are ``too large'' for a given window function $\window(\time)$, the \whac set $(\window,T,F)$ cannot be a  frame for $\hilfunspace$. This is because a \whac set $(\window,T,F)$ with ``large'' parameters $T$ and $F$ ``leaves holes in the time-frequency plane'' or equivalently in the Hilbert space $\hilfunspace$. Indeed, this intuition is correct and the following fundamental result formalizes it:
\begin{thm}[{\hspace{-0.06mm}\cite[Thm. 8.3.1]{christensen96}}]
\label{thm:negativeframe}
	Let $\window(\time)\in\hilfunspace$ and $T,F>0$ be given. Then the following holds:
	\begin{itemize}
		\item If $TF>1$, then $(\window,T,F)$ is \emph{not} a frame for \hilfunspace.
		\item If $(\window,T,F)$ is a frame for $\sqintspace$, then $(\window,T,F)$ is an exact frame if and only if $TF=1$.
	\end{itemize}
\end{thm} 
We see that $(\window,T,F)$ can be a frame for $\sqintspace$ only if $TF\le 1$, i.e., when the shift parameters~$T$ and~$F$ are such that the grid they induce in the time-frequency plane is sufficiently dense. Whether or not a \whac set $(\window,T,F)$ with $TF\le 1$ is a frame for $\hilfunspace$ depends on the window function $\window(\time)$ and on the  values of~$T$ and~$F$. There is an important special case where a simple answer can be given.

\begin{example}[Gaussian, {\cite[Thm. 8.6.1]{christensen96}}]
Let $T,F>0$ and take $\window(\time)=e^{-\time^2}$. Then the \whac set
\begin{equation*}
	\left\{\weylop^{(T,F)}_{m,n}\window(\time)\right\}_{m\in\integers,\, n\in\integers}
\end{equation*}
is a frame for \hilfunspace if and only if $TF<1$.
\end{example}

\subsection{Wavelets}
Both for wavelet frames and \whac frames we deal with function sets that are obtained by letting a special class of parametrized operators act on a fixed function. In the case of \whac frames 
the underlying operator realizes time and frequency shifts. In the case of wavelets, the generating operator realizes time-shifts and scaling. Specifically, we have the following definition.
\begin{dfn}
\label{dfn:waveletop}
The operator $\waveop^{(T,S)}_{m,n}:\hilfunspace\to\hilfunspace$ is defined as
\begin{equation*} 
\waveop^{(T,S)}_{m,n}: \fun(\time)\to S^{n/2} \fun(S^n\time-mT), 
\end{equation*}
where $m,n \in\integers$, and $T>0$ and $S>0$ are fixed time and scaling parameters,
respectively.
\end{dfn}
Now, just as in the case of \whac expansions, consider a prototype function (or mother wavelet)  $\window(\time)\in\hilfunspace$. Fix the parameters $T>0$ and $S>0$ and consider the set of functions
\begin{equation*}
	\window_{m,n}(\time)\define \waveop^{(T,S)}_{m,n}\window(\time)= S^{n/2} \window(S^n\time-mT) ,\quad m\in\integers,\ n\in\integers.
\end{equation*}
This set is referred to as a \emph{wavelet set}.
	When the wavelet set
		$\left\{\window_{m,n}(\time)=  S^{n/2} \window(S^n\time-mT) \right\}_{m\in\integers,\ n\in\integers}$
	with parameters $T,S>0$ is a frame for $\hilfunspace$, it is called a \emph{wavelet frame}.

Similar to the case of Weyl-Heisenberg sets it is hard to say, in general, whether a given wavelet set forms a frame for \hilfunspace or not. The answer depends on the window function $\window(\time)$ and on the parameters $T$ and $S$ and explicit results are known only in certain  cases.
We conclude this section by detailing such a case.
\begin{example}[Mexican hat, {\cite[Ex. 11.2.7]{christensen96}}]
	Take $S=2$ and consider the mother wavelet
	\begin{equation*}
		\window(\time)=\frac{2}{\sqrt{3}}\pi^{-1/4}(1-\time^2) e^{-\frac{1}{2}\time^2}.
	\end{equation*}
	Due to its shape, $\window(\time)$ is called the Mexican hat function. It turns out that for each $T<1.97$, the wavelet set 
	\begin{equation*}
		\left\{\waveop^{(T,S)}_{m,n}\window(\time)\right\}_{m\in\integers,\ n\in\integers}
	\end{equation*}
	is a frame for $\hilfunspace$~\cite[Ex. 11.2.7]{christensen96}.
\end{example}	
\section{Exercises}
\begin{xca}[Tight frames~\cite{mallat09}]
	{\ } 
\begin{enumerate}
\item
Prove that if $K\in \integers\setminus\{0\}$, then the set of vectors 
\begin{equation*}
	\left\{\vanalframeel_{k}=\tp{\left[1\ \ \cex{k /(K\dimension)}\ \ \cdots\ \ \cex{k (\dimension-1)/(K\dimension)}\right]}\right\}_{0\le k<K\dimension} 
\end{equation*}
is a tight frame for $\complexset^{\dimension}$. Compute the frame bound.
\item
Fix $T\in \reals\setminus\{0\}$ and  define for every $k\in\integers$
\begin{equation*}
	\analframeel_{k}(t)=\begin{cases}
		\cex{k t/T}, &t\in [0,T]\\
		0,&\text{otherwise}.
	\end{cases}
\end{equation*}
Prove that $\{\analframeel_{k}(t)\}_{k\in\integers}$ is a tight frame for $\hilfunspace([0,T])$, the space of square integrable functions supported on the interval $[0,T]$. Compute the frame bound.
\end{enumerate}
\end{xca}
\begin{xca}[\dftac as a signal expansion]
The \dftac of an $\dimension$-point signal $\fun[\dtime],\, \dtime=0,\ldots,\dimension-1,$ is defined as 
\begin{equation*}
	\ft{\fun}[\dfreq]=\frac{1}{\sqrt{\dimension}}\sum_{\dtime=0}^{\dimension-1} \fun[\dtime] \cexn{\frac{\dfreq}{\dimension}\dtime}.
\end{equation*}
Find the corresponding inverse transform and show that the \dftac can be interpreted as a frame expansion in $\complexset^{\dimension}$. Compute the frame bounds. Is the underlying frame special?
\end{xca}
\begin{xca}[Unitary transformation of a frame]
Let $\{\analframeel_{k}\}_{k\in \frameset}$ be a frame for the Hilbert space \hilspace with frame bounds $\frameA$ and $\frameB$. Let 
$\opU:\hilspace\to\hilspace$ be a unitary operator. 
Show that the set $\{\opU\analframeel_{k}\}_{k\in \frameset}$ is again a frame for \hilspace and compute the corresponding frame bounds.
\end{xca}
\begin{xca}[Redundancy of a frame]
Let $\{\vanalframeel_{k}\}_{k=1}^\framesize$ be a frame for $\complexset^\dimension$ with $\framesize>\dimension$. Assume that the frame vectors are normalized such that $\vecnorm{\vanalframeel_{k}}=1,\, k=1,\ldots,\framesize$. The ratio $\framesize/\dimension$ is called the redundancy of the frame.  
\begin{enumerate}
	\item
	Assume that $\{\vanalframeel_{k}\}_{k=1}^\framesize$ is a tight frame with frame bound $\frameA$. 
Show that $\frameA=\framesize/\dimension$.
    \item 
Now assume that $\frameA$ and $\frameB$ are the frame bounds of $\{\vanalframeel_{k}\}_{k=1}^\framesize$. Show that $\frameA\le \framesize/\dimension\le \frameB$. 
\end{enumerate}
\end{xca}

\begin{xca}[Frame bounds~\cite{christensen96}]
Prove that the upper and the lower frame bound are unrelated: In an arbitrary Hilbert space \hilspace find a set $\{\analframeel_{k}\}_{k\in \frameset}$ with an upper frame bound $\frameB<\infty$ but with the tightest lower frame bound $\frameA=0$; find another set $\{\analframeel_{k}\}_{k\in \frameset}$ with lower frame bound $\frameA>0$ but with the tightest upper frame bound $\frameB=\infty$. Is it possible to find corresponding examples in the finite-dimensional space $\complexset^\dimension$?
\end{xca}

\begin{xca}[Tight frame as an orthogonal projection of an \onbac]
Let $\{\vbasisel_{k}\}_{k=1}^\framesize$ be an \onbac for an $\framesize$-dimensional Hilbert space \hilspace. For~$\dimension<\framesize$, let $\hilspace'$ be an $\dimension$-dimensional subspace of \hilspace. Let $\opP:\hilspace\to\hilspace$ be the orthogonal projection onto $\hilspace'$. Show that $\{\opP\vbasisel_{k}\}_{k=1}^\framesize$ is a tight frame for  $\hilspace'$. Find the corresponding frame bound.
\end{xca}

\begin{xca}[]
	\label{xca:ocrange}
Consider the space of  $\hilseqspace$-sequences with \dtftac supported on the interval $[\freq_1,\freq_2]$ with $-1/2<\freq_1<\freq_2<1/2$. Show that the orthogonal complement (in $\hilseqspace$) of this space is the space of  $\hilseqspace$-sequences with \dtftac supported on the set $[-1/2,1/2]\setdiff[\freq_1,\freq_2]$.	
[Hint: Use the definition of the orthogonal complement and apply Parseval's theorem.]
\end{xca}
	
\begin{xca}[]
	\label{xca:opA}
Refer to \fref{sec:designfreedom} and consider the operator $\opA$ in \fref{eq:opAdfn} and the operator $\analop$ in~\fref{eq:analopsampling}. 
\begin{enumerate}
	\item 
	Consider a sequence $\{a_k\}_{k\in\integers}\in \rng(\analop)$ and show that $\opA\{a_k\}_{k\in\integers}=\ad\analdualop\{a_k\}_{k\in\integers}$.
	\item
	Consider a sequence $\{b_k\}_{k\in\integers}\in \ocomp{\rng(\analop)}$ and show that $\opA\{b_k\}_{k\in\integers}=0$.
	\item
	Using the fact that every sequence $\{c_k\}_{k\in\integers}$ can be decomposed as $\{c_k\}_{k\in\integers}=\{a_k\}_{k\in\integers}+\{b_k\}_{k\in\integers}$ with $\{a_k\}_{k\in\integers}\in \rng(\analop)$ and $\{b_k\}_{k\in\integers}\in \ocomp{\rng(\analop)}$, show that $\opA=\ad\analdualop\pop$
	where $\pop:\hilseqspace\to \hilseqspace$ is the orthogonal projection operator onto $\rng(\analop)$.
\end{enumerate}
[Hints: Use the fact that $\rng(\analop)$  is the space of  $\hilseqspace$-sequences with \dtftac supported on the interval  
	$[-\bandwidth\sinterval,\bandwidth\sinterval]$; use the characterization of $\ocomp{\rng(\analop)}$ developed in~\fref{xca:ocrange}; work in the \dtftac domain.]
\end{xca}
	
\begin{xca}[]
	Refer to \fref{sec:designfreedom} and use the ideas from~\fref{xca:opA} to show that the  operator~$\opB$ in \fref{eq:opBdfn} can be written as $\opB=\opM(\iop_{\hilseqspace}-\pop)$, where  $\pop:\hilseqspace\to \hilseqspace$ is the orthogonal projection operator onto~$\rng(\analop)$ with $\analop$ defined in~\fref{eq:analopsampling}; and $\opM:\hilseqspace\to \hilfunspace$ is the interpolation operator defined in~\fref{eq:opBdfn1}.
	\end{xca}

\begin{xca}[Weyl operator~\cite{daubechies86-05}]
Refer to~\fref{dfn:weylop} and show the following properties of the Weyl operator.
\begin{enumerate}
 	\item 
 	The following equality holds:
	\begin{equation*}
	\weylop_{m,n}^{(T,F)} \weylop_{k,l}^{(T,F)}=e^{-\iu 2\pi mlTF} \weylop_{m+k,n+l}^{(T,F)}.
	\end{equation*}
	\item
	The adjoint operator of $\weylop_{m,n}^{(T,F)}$ is given by
	\begin{equation*}
	\ad{\left(\weylop_{m,n}^{(T,F)}\right)}=e^{-\iu 2\pi m n T F}\weylop_{-m,-n}^{(T,F)}.
	\end{equation*}
	\item
	The Weyl operator is unitary on \hilfunspace, i.e.,
	\begin{equation*}
	\weylop_{m,n}^{(T,F)}\ad{\left(\weylop_{m,n}^{(T,F)}\right)}=\ad{\left(\weylop_{m,n}^{(T,F)}\right)}\weylop_{m,n}^{(T,F)}=\opI_{\hilfunspace}.
	\end{equation*}
 \end{enumerate}
\end{xca}  

\begin{xca}[Dual WH frame~\cite{daubechies90-09}]
\label{ex:dualWHframe}
Assume that the \whac set $\left\{\window_{m,n}(\time)=  \weylop_{m,n}^{(T,F)}\analframeel(\time)\right\}_{m\in\integers,\ n\in\integers}$ is a frame for \hilfunspace with frame operator \frameop.
\begin{enumerate}
	\item 
	Show that the frame operator \frameop and its inverse $\frameop^{-1}$ commute with the Weyl operators, i.e.,
	\begin{align*}
	&\weylop_{m,n}^{(T,F)}\frameop=\frameop \weylop_{m,n}^{(T,F)}\\
	&\weylop_{m,n}^{(T,F)}\frameop^{-1}=\frameop^{-1} \weylop_{m,n}^{(T,F)}
	\end{align*}
	for $m,n\in \integers.$
	\item
	Show that the minimal dual frame $\{\synthframeel_{m,n}(\time)=(\frameop^{-1} \window_{m,n})(\time)\}$ is a \whac frame with prototype function~$\synthframeel(\time)=(\frameop^{-1} \window)(\time)$, i.e., that
	\begin{equation*}
	\synthframeel_{m,n}(\time)=\weylop_{m,n}^{(T,F)} \synthframeel(\time).
	\end{equation*}
\end{enumerate}
\end{xca} 

\begin{xca}[WH frames in finite dimensions]
This is a Matlab exercise. The point of the exercise is to understand what the abstract concept of the frame operator and the dual frame mean in linear algebra terms.

Consider the space $\complexset^\dimension$. Take $\dimension$ to be a large number, such that your signals resemble continuous-time waveforms, but small enough such that your Matlab program works. Take a prototype vector $\vanalframeel=\tp{[\vanalframeelc[1] \cdots \vanalframeelc[\dimension]]}\in \complexset^\dimension$. You can choose, for example, the $fir1(.)$ function in Matlab, or discrete samples of the continuous-time Gaussian waveform $e^{-x^2/2}$.
Next, fix the shift parameters $T,K\in \naturals$ in such a way that $L\define \dimension/T\in\naturals$. Now define
\begin{equation*}
	\vanalframeelc_{k,l}[n]\define \vanalframeelc[(n-lT)\!\!\! \mod \dimension]\, e^{\iu 2\pi k n/K}, \ k=0,\ldots, K-1, \ l=0,\ldots, L-1,\ n=0,\ldots, \dimension-1
\end{equation*}
and construct a discrete-time \whac set according to
\begin{equation*}
\left\{\vanalframeel_{k,l}=\tp{\left[\vanalframeelc_{k,l}[0] \ \cdots \ \vanalframeelc_{k,l}[\dimension-1]\right]}\right\}_{k=0,\ldots, K-1,\, l=0,\ldots, L-1}.
\end{equation*}

\begin{enumerate}
	\item 
	Show that the analysis operator $\analop:\complexset^\dimension\to\complexset^{KL}$ can be viewed as a $(KL\times \dimension)$-dimensional matrix. Specify this matrix in terms of $\vanalframeel, T, K,$ and $\dimension$. 
	\item
	Show that the adjoint of the analysis operator $\ad\analop:\complexset^{KL}\to \complexset^\dimension$ can be viewed as an $(\dimension\times KL)$-dimensional matrix.  Specify this matrix in terms of $\vanalframeel, T, K,$ and $\dimension$.
	\item 
	Specify the matrix corresponding to the frame operator $\fop$ in terms of $\vanalframeel, T, K,$ and $\dimension$. Call this matrix \matS. Compute and store this $\dimension\times \dimension$ matrix in Matlab.
	\item
	Given the matrix $\matS$, check, if the \whac system $\left\{\vanalframeel_{k,l}\right\}_{k=0,\ldots, K-1,\, l=0,\ldots, L-1}$ you started from  is a frame. Explain, how you can verify this. 

	\item
	Prove that for  $K=\dimension$ and $T=1$ and for every prototype vector $\vanalframeel\ne \veczero$, the set $\left\{\vanalframeel_{k,l}\right\}_{k=0,\ldots, K-1,\, l=0,\ldots, L-1}$ is a frame for $\complexset^\dimension$.
	\item
	For the prototype vector $\vanalframeel$ you have chosen, find two pairs of shift parameters $(T_1, K_1)$ and $(T_2, K_2)$ such that $\left\{\vanalframeel_{k,l}\right\}_{k=0,\ldots, K-1,\, l=0,\ldots, L-1}$ is a frame for $T=T_1$ and $K=K_1$ and is not a frame for $T=T_2$ and $K=K_2$. 
	For the case where $\left\{\vanalframeel_{k,l}\right\}_{k=0,\ldots, K-1,\, l=0,\ldots, L-1}$ is a frame, compute the frame bounds.
	\item
	Compute the dual prototype vector $\tilde \vanalframeel=\tp{[\tilde\vanalframeelc[1] \cdots \tilde\vanalframeelc[\dimension]]}=\matS^{-1} \vanalframeel$.
	Show that the dual frame $\left\{\tilde\vanalframeel_{k,l}\right\}_{k=0,\ldots, K-1,\, l=0,\ldots, L-1}$ is given by time-frequency shifts of  $\tilde \vanalframeel$, i.e.,
	\begin{equation*}
	\left\{\tilde\vanalframeel_{k,l}=\tp{\left[\tilde\vanalframeelc_{k,l}[0] \ \cdots\ \tilde\vanalframeelc_{k,l}[\dimension-1]\right]}\right\}_{k=0,\ldots, K-1,\, l=0,\ldots, L-1}
	\end{equation*}
	with 
	\begin{equation*}
		\tilde\vanalframeelc_{k,l}[n]\define \tilde\vanalframeelc[(n-lT)\!\!\! \mod \dimension]\, e^{\iu 2\pi k n/K}, \ k=0,\ldots, K-1, \ l=0,\ldots, L-1,\ n=0,\ldots, \dimension-1.
	\end{equation*}
\end{enumerate}
\end{xca} 

\newpage 
\section*{Notation}

\begin{tabular}[t]{@{}>{$}p{0.22\linewidth}<{$}@{}>{\raggedright\arraybackslash}p{0.78\linewidth}@{}}

	\veca, \vecb, \dots & vectors\\
	\matA, \matB, \dots & matrices\\
	\tp{\veca},\; \tp{\matA} & transpose of the vector~\veca and the
 		matrix~\matA\\
	\conj{\sca},\; \conj{\veca},\; \conj{\matA} & complex conjugate
 		of the scalar~\sca, element-wise complex conjugate of the vector~\veca, and the matrix~\matA\\
	\herm{\veca},\; \herm{\matA} & Hermitian transpose of the vector~\veca and
 		the matrix~\matA\\
 	\identity_N & identity matrix of size $N\times N$ \\
 	\rank(\matA)	& rank of the matrix~\matA\\
 	\eval(\matA) & eigenvalue of the matrix~\matA \\
 	\eval_{\min}(\matA), \eval_{\min}(\opA) &  smallest eigenvalue of the matrix~$\matA$,  smallest spectral value of the self-adjoint operator~$\opA$\\
 	\eval_{\max}(\matA), \eval_{\max}(\opA) &  largest eigenvalue of the matrix~$\matA$,  largest spectral value of the self-adjoint operator~$\opA$\\
	 		\iu & $\sqrt{-1}$\\
	 		\define & definition\\
			\setA,\setB,\dots	& sets\\
	 		\reals, \complexset, \integers, \naturals	& real line, complex plane, set of all integers, set of natural numbers (including zero)\\
	 		\hilfunspace & Hilbert space of complex-valued finite-energy functions\\
		\sqintspaceBL & space of square-integrable functions bandlimited to \bandwidth\Hz\\
			\hilspace & abstract Hilbert space\\
			\hilseqspace & Hilbert space of square-summable sequences\\
	 		\inner{\veca}{\vecb}	& inner product of the vectors~\veca and~\vecb: $\inner{\veca}{\vecb}\define\sum_{\sci} \elementof{\veca}{\sci}\conj{(\elementof{\vecb}{\sci})} $ \\
	 		\inner{\fun}{\altfun} & depending on the context: inner product in the abstract Hilbert space \hilspace or inner product of the functions~$\fun(\time)$ and~$\altfun(\time)$: $\inner{\fun}{\altfun}\define\int_{-\infty}^\infty \fun(\time) \conj{\altfun}(\time) d \time$\\
	 		\vecnorm{\veca}^2 & squared $\ell^{2}$-norm of the vector~$\veca$: $\vecnorm{\veca}^2\define\sum_{\sci} \abs{\elementof{\veca}{\sci}}^2$\\
	\vecnorm{\altfun}^2 & depending on the context: squared norm in the abstract Hilbert space \hilspace or squared $\hilfunspace$-norm of the function~$\altfun(\time)$: $\vecnorm{\altfun}^2 \define\int_{-\infty}^\infty \abs{\altfun(\time)}^2 d \time$\\
		\iop_\hilspace, \iop_{\hilseqspace}, \iop_{\hilfunspace},  \iop_{\sqintspaceBL} & identity operator in the corresponding space\\
		\rng(\opA) & range space of operator \opA\\
		\ad\opA & adjoint of operator \opA\\
		\ft{\fun}(\freq) & Fourier transform of $\fun(\time)$: $\ft{\fun}(\freq)\define\int_{-\infty}^{\infty}\fun(\time) e^{-\iu  2\pi  \time\freq} d\time$\\
		\ftd{\fun}(\freq)& Discrete-time Fourier transform of $\fun[k]$: $\ftd{\fun}(\freq)\define\sum_{k=-\infty}^\infty \fun[k]e^{-\iu 2\pi k\freq}$  
\end{tabular}
\newpage 
\bibliographystyle{IEEEtran}
\bibliography{IEEEabrv,publishers,confs-jrnls,vebib}

\end{document}